\documentclass[letterpaper,11pt]{article}  
\usepackage[sort,compress]{cite}
\usepackage[top=1.1in, bottom=1.13in, left=1.1in, right=1.1in]{geometry}
\usepackage[utf8]{inputenc}
\usepackage{amsmath}
\usepackage{amssymb}
\usepackage{amsfonts}
\usepackage{amsthm}
\usepackage{graphicx}
\usepackage{tikz}
\usepackage{xspace}
\usepackage{subcaption}
\usepackage{color}
\usepackage{comment}
\usepackage{titling}
\usepackage{enumitem}
\usepackage[font={small}]{caption}
\usepackage{changepage}
\usepackage[export]{adjustbox}
\usepackage{mathtools}
\usepackage{rotating}
\usepackage{float}
\usepackage[skip=0pt]{caption} %remove space between figure and caption

\usepackage{hyperref}

\makeatletter
\newcommand\footnoteref[1]{\protected@xdef\@thefnmark{\ref{#1}}\@footnotemark}
\makeatother

\setlist[itemize]{leftmargin=1em,topsep=0.2em,noitemsep}
\setlist[enumerate]{leftmargin=1.4em,topsep=0.2em,noitemsep,labelsep=0.3em}

\usetikzlibrary{decorations.pathmorphing}
\usetikzlibrary{shapes.misc} % for crosses
\usetikzlibrary{backgrounds}

\newtheorem{theorem}{Theorem}[section]

\newtheorem{lemma}[theorem]{Lemma}
\newtheorem{corollary}[theorem]{Corollary}

\newtheorem{observation}[theorem]{Observation}
\theoremstyle{definition}
\newtheorem{definition}[theorem]{Definition}

\newcommand{\REPL}{\mathsf{REPR}}
\newcommand{\frakC}{\mathfrak{C}}
\newcommand{\vect}{\protect\overrightarrow}

\newcommand{\resp}{respectively\xspace}
\newcommand{\termasm}[1]{\mathcal{A}_{\Box}[{#1}]}
\newcommand{\prodasm}[1]{\mathcal{A}[{#1}]}
\newcommand{\dom}[1]{{\rm dom}(#1)}
\newcommand{\Z}{\mathbb{Z}}
\newcommand{\N}{\mathbb{N}}
\newcommand{\R}{\ensuremath{\mathbb{R}}}
\newcommand{\calT}{\mathcal{T}}
\newcommand{\calU}{\mathcal{U}}

\newcommand{\prodT}{\prodasm{\mathcal{T}}}
\newcommand{\prodpaths}[1]{{\bf{P}}[{#1}]}
\newcommand{\prodpathsT}{\prodpaths{\mathcal{T}}}
\newcommand{\prodpathsU}{\prodpaths{\mathcal{U}}}
\newcommand{\ra}{\rightarrow}

\newcommand{\rev}[1]{\ensuremath{#1^\leftarrow}}

\newcommand{\pathassembly}[1]{\mathrm{asm}{(#1)}}
\newcommand{\asm}[1]{\pathassembly{#1}}
\newcommand{\gglue}[3]{\mathrm{glue}(#1_{#2} #1_{#3})}
\newcommand{\glue}[2]{\mathrm{glue}(P_{#1} P_{#2})}
\newcommand{\glueP}[2]{\mathrm{glue}(P_{#1} P_{#2})}
\newcommand{\gluePD}[2]{\mathrm{glue}(P_{#1} D_{#2})}

\newcommand{\glueQ}[2]{\mathrm{glue}(Q_{#1} Q_{#2})}

\newcommand{\namedGlue}[2]{\mathrm{glue}(#1  \,  #2)}

\newcommand{\pos}[1]{\mathrm{pos}(#1)}
\newcommand\type[1]{\mathrm{type}(#1)}
\newcommand\defeq{\mathrel{\overset{\makebox[0pt]{\mbox{\normalfont\tiny\sffamily def}}}{=}}}

\newcommand{\thmTM}{}   
\newcommand{\thmMain}{}

\newcommand\midpoint[2]{\mathrm{mid}(#1,#2)}

\newcommand{\pn}{\ensuremath{P^{n+1}}\xspace}

\newcommand{\pk}{\ensuremath{P^{k}}\xspace}
\newcommand{\fk}{\ensuremath{E^{k}}\xspace}
\newcommand{\gk}{\ensuremath{F^{k}}\xspace}

\newcommand{\hs}{$h$-successful\xspace}
\newcommand{\an}{\ensuremath{\alpha_n}\xspace}

\newcommand{\asmprefix}[2]{{}^{#1\leadsto\!}#2}

\newcommand{\embed}[1]{\ensuremath{\frak{E}_{#1}}}
\newcommand{\ep}{\ensuremath{\embed{P}}}

\newcommand{\pnviz}{\ensuremath{\gglue{P^{n+1}}{p}{p+1}}\xspace}

\newcommand{\midpnviz}{\ensuremath{\midpoint{P^{n+1}_p}{P^{n+1}_{p+1}}}\xspace}

\newcommand{\olq}{\ensuremath{\overline{q}}}

\makeatletter
\newtheorem*{rep@theorem}{\rep@title}
\newcommand{\newreptheorem}[2]{%
\newenvironment{rep#1}[1]{%
 \def\rep@title{#2 \ref{##1}}%
 \begin{rep@theorem}}%
 {\end{rep@theorem}}}
\makeatother
\newreptheorem{theorem}{Theorem}

\pgfdeclaredecoration{penciline}{initial}{
    \state{initial}[width=+\pgfdecoratedinputsegmentremainingdistance,
    auto corner on length=1mm,]{
        \pgfpathcurveto%
        {% From
            \pgfqpoint{\pgfdecoratedinputsegmentremainingdistance}
                      {\pgfdecorationsegmentamplitude}
        }
        {%  Control 1
        \pgfmathrand
        \pgfpointadd{\pgfqpoint{\pgfdecoratedinputsegmentremainingdistance}{0pt}}
                    {\pgfqpoint{-\pgfdecorationsegmentaspect
                     \pgfdecoratedinputsegmentremainingdistance}%
                               {\pgfmathresult\pgfdecorationsegmentamplitude}
                    }
        }
        {%TO 
        \pgfpointadd{\pgfpointdecoratedinputsegmentlast}{\pgfpoint{1pt}{1pt}}
        }
    }
    \state{final}{}
}

\title{The non-cooperative tile assembly model is not intrinsically universal or capable of bounded Turing machine simulation}

\author{Pierre-\'Etienne Meunier\\
  Inria  \\
\href{mailto:pierre-etienne.meunier@inria.fr}{pierre-etienne.meunier@inria.fr}\thanks{This work was carried out while at Inria, Paris, France, and the Department of Computer Science, Aalto University, Finland,
and  
Aix Marseille Universit\'e, CNRS, LIF UMR 7279, 13288, Marseille, France,
and
 LIAFA UMR 7089, Paris~7, France, and California Institute of Technology, Pasadena, CA 91125, USA. Supported in part by National Science Foundation Grant CCF-1219274.}
\and
Damien Woods\\
Inria \\
\href{mailto:damien.woods@inria.fr}{damien.woods@inria.fr}\thanks{This work was carried out while at Inria, Paris, France, as well as California Institute of Technology, Pasadena, CA 91125, USA, and during a brief visit to  LIAFA (UMR 7089), Paris 7, France.  Supported by National Science Foundation grants CCF-1219274, 0832824 (The Molecular Programming Project), CCF-1219274, and CCF-1162589, USA, a visiting professor award from Paris~7, and by INRIA.}}
\date{}

\begin{document}

\numberwithin{figure}{section}
\newcommand\turinginothermodels{}  
\maketitle

% !TEX root = t1notiu.tex

\begin{abstract}
The field of algorithmic self-assembly is concerned with the computational and expressive power of nanoscale self-assembling molecular systems. In the well-studied cooperative, or temperature 2, abstract tile assembly model it is known that there is a tile set to simulate any Turing machine and an intrinsically universal tile set that simulates the shapes and dynamics of any instance of the model, up to spatial rescaling. It has been an open question as to whether the seemingly simpler noncooperative, or temperature 1, model is capable of such behaviour. Here we show that this is not the case, by showing that there is no tile set in the noncooperative model that is intrinsically universal, nor one capable of time-bounded Turing machine simulation within a bounded region of the plane.

Although the noncooperative model intuitively seems to lack the complexity and power of the cooperative model it was not immediately obvious how to prove this. 
One reason is that there have been few tools to analyse the structure of complicated paths in the plane. This paper provides a number of such tools. A second reason is that almost every obvious and small generalisation to the model (e.g.\ allowing error, 3D, non-square tiles, signals/wires on tiles, tiles that repel each other, parallel synchronous growth) endows it with great computational, and sometimes simulation, power. Our main results show that all of these generalisations provably increase computational and/or simulation power. Our results hold for both deterministic and nondeterministic noncooperative systems. Our first main result stands in stark contrast with the fact that for both the cooperative tile assembly model, and for 3D noncooperative tile assembly, there are respective intrinsically universal tilesets. Our second main result gives a new technique (reduction to simulation) for proving negative results about computation in tile assembly. However, our results leave as an open problem whether there might be other ways noncooperative systems compute. 

\end{abstract}

\thispagestyle{empty}

\let\oldref=\ref
\renewcommand\ref[1]{\oldref{ext-#1}}
\renewcommand\ref[1]{\oldref{#1}}

% !TEX root = t1notiu.tex

\section{Introduction}
The design and laboratory fabrication of nanoscale molecular systems that implement sophisticated computation is a goal held by many. 
If we are to have such an engineering discipline that exploits  the idea that molecules can compute, then we need a firm foundation of the kind of computational theory that is relevant to such systems. 
The field of algorithmic tile assembly provides one such theoretical framework targeted specifically at molecular self-assembling systems.  
One of most well-studied models of computation for molecular self-assembly systems is the abstract tile assembly model, put forward by Winfree~\cite{Winf98}. The model describes crystal-like growth process where, starting from a small connected arrangement of square tiles,  called a seed assembly,  
a growth process takes place where other unit-size square tiles stick to the ever-larger growing assembly. Local rules specify which tiles can stick at each location along the boundary of the assembly. Growth happens asynchronously and in parallel; the model is a kind of asynchronous nondeterministic cellular automaton. Winfree~\cite{Winf98} showed that the model can simulate Turing machines, Winfree and Rothemund showed that it can efficiently self-assemble squares~\cite{RotWin00,Roth01}, and Winfree and Soloveichik~\cite{SolWin07} used bounded-space simulation of time-/space-bounded Turing machines to exhibit for each finite connected shape a Kolmogorov-efficient tile set that assembles a scaled version of that shape.  
Recently, it has been shown that there is even a single {\em intrinsically universal} tile set set that faithfully simulates the geometry (shapes) and dynamics of any instance of the model, up to spatial rescaling~\cite{IUSA}.

These results were all shown for the so-called {\em cooperative} (or {\em temperature 2}) model, where tiles bind to the growing assembly if they, or at least some of them, bind on two or more sides. This provides a kind of ``context sensitivity'' in the growth process. What happens if we allow {\em noncooperative} (or {\em temperature 1})  growth where tiles bind if they match on at least one side?  Growth like this looks like growing and branching tips in 2D. Tendrils snake out from the seed, possibly crashing into each other, and more often than not they seem to  merely form simple structures (cycles and/or repeated path segments), and certainly not the kind of structures needed for computation. 
Putting proofs behind this intuition has been a challenge and the literature has seen a number of unproven conjectures about the limitations of temperature 1. In this paper, we settle two  such questions. 

Our first main result is on the topic of simulation in tile assembly. As noted, it has been shown that there is an intrinsically universal tile set  for the cooperative model; that is, a tile set is capable of simulating any instance of the cooperative model~\cite{IUSA}. More precisely,   there is a tile set~$U$ that given as input any instance~$\calT$ of the tile assembly model (encoded as a seed assembly), tiles from $U$ self-assemble (at temperature 2) to  simulate  the geometry (shapes) and dynamics of~$\calT$ perfectly, modulo a spatial rescaling.   By spatial rescaling we mean that each unit-sized square tile in~$\calT$ is simulated by an $m \times m$ square block of tiles over $U$.   The result is a kind of completeness result for the abstract tile assembly model: the tile set $U$ is ``hard'' for all tile assembly systems in the sense it is able, via extensive use of cooperative binding, to capture all possible production and dynamics of all systems, and of course every instantiation of $U$ is itself also a valid tile assembly system. Since then,  it has been shown~\cite{Meunier-2014} that the noncooperative tile assembly model can not simulate the cooperative model  but it was left open  (Conjecture~1.4~\cite{Meunier-2014})  whether the noncooperative model can  simulate itself. So although the noncooperative model is weak, perhaps it is just strong enough for self-simulation? In other words, is there a noncooperative tile set that is  ``hard'' for the noncooperative model? We answer this conjecture by showing that there is no such intrinsically universal tile set for the noncooperative~model. 

\renewcommand\turinginothermodels{\cite{RotWin00,SolWin07,Versus,OneTile,Cook-2011,Roth01,Patitz-2011,Signals,Hendricks-2014,Fekete2014,gilbert2015continuous,Winf98,Winfree98simulationsof}}

Our second main result is on computation in the noncooperative model. We show, that it is impossible to simulate a time-bounded Turing machine in a bounded rectangular region of the plane in the noncooperative model (see Theorem~\ref{thm:no_finite_TM} for the formal statement).
Although this statement has caveats (i.e.\ both instances of the word ``bounded''),  it implies that the noncooperative model  can not simulate Turing machines using any method with a geometry remotely similar to any of the known ways to simulate Turing machines in any known tile assembly model~\turinginothermodels.\footnote{I.e.\ by simulating a time $t(n)$ and space $s(n)$-bounded Turing machine in a $O(f(t(n))) \times O(g(s(n)))$ region for finite functions $f$ and $g$.} It is important to note that the  negative result about simulation in ref~\cite{Meunier-2014} does not say anything about computation in the model; in fact that particular negative result also holds in the 3D noncooperative model, despite the fact that model can simulate Turing machines. It is also important to point out that many generalisations~\cite{Versus,2HAMIU,OneTile,geotiles,Cook-2011,Roth01,Patitz-2011,BMS-DNA2012a,Signals,Jonoska2014,Fekete2014,Hendricks-2014,gilbert2015continuous}
%(enumerated below) 
of the classical 2D noncooperative model can indeed carry out ``bounded'' simulation of Turing machines; thus our result formally separates these generalised noncooperative models from the classical 2D model.

\paragraph{New tools for noncooperative tile assembly.}
Besides showing limitations on noncooperative growth in terms of simulation and computation power, we contend that this paper brings some new techniques to the table. Generally in tile assembly systems, in order to carry out nontrivial computation for finite or infinite shape-building one often has the goal of building structures that (a) are large but (b) not too large (e.g. neither hardcoding a small shape nor filling the entire plane could be reasonably regarded as interesting computation---the interesting {\em algorithmic} stuff lies in-between).  In this paper we provide two tools to analyse, and prove negative results on, building such shapes in the noncooperative model. The first is a method to show that any {\em any}  path of tiles $P$ that travels a long enough horizontal distance while staying above some horizontal line can be either pumped forever or else blocked by growing something else. Hence if $P$ was supposed to form part of some interesting shape, then our first tool (Lemma~\ref{lem:onepath}) makes it so that we can use $P$ to make another path that either goes outside the shape ($P$ is pumpable), or else prevents $P$ from growing to completion ($P$ is blocked).  This contrasts with previous works, e.g. \cite{Doty-2011,Manuch-2010}, since here we \emph{use} non-pumpable paths to prove that other ``unintended'' assemblies can be produced. In fact one of the main new ideas in our work is to prove strong properties about non-pumpable paths.

Our second tool (Theorem~\ref{thm:shapes}) builds on this to simultaneously block multiple paths, despite the fact they may interact with each other in very complicated ways. 
More precisely, given a set of paths of tiles, we define a total ordering on those paths so that we can iteratively apply the first tool to infinitely pump and/or block all of the paths.  
This methodology seems general enough that it might find future application. 
See Section~\ref{sec:intuition} for a proof overview.  

Another contribution of this work that might prove useful in the future is a set of definitions (Section~\ref{sec:defs-paths}) and lemmas (Section~\ref{secvisibility}) that capture a number of basic properties about producible paths at temperature 1. 
Two properties we reason about again and again are (a) visibility of a glue $g$ from the south meaning that no glue on the path lies immediately below $g$, and (b) the notion of one path being more right-turning, or more {\em right-priority}, than another. 
We often force right-priority paths $P$ to grow a branch until that branch crashes into (intersects the position of) a prefix of the path, then we embed the new crashed path in $\mathbb{R}^2$ and reasoning using the visibility of some glue along $P$ to argue that the crashed path encloses a component of the plane along the {\em left-hand side} of $P$. 
 Our  conventions and tricks for reasoning about paths of tiles via  embeddings in $\mathbb{R}^2$ could be applied to a variety of models. 
 They in turn allow us to frequently use reasoning that is at the abstraction level of {\em paths in the plane} as opposed to the more low-level of individual tiles and glues.  
Together this collection of tools allow us to disrupt any attempt to build shapes of a certain kind, and they work whether or not nondeterminism is deployed as a tool by the ill-fated programmer.  We hope  these ideas may find use  independently of the two main problems we solve here.

Intrinsic universality, and simulation between tile assembly systems, is giving rise to a kind of complexity theory for comparing models of self-assembly\cite{woods2015ntrinsic}. It is interesting to note that in this setting sometimes it is possible to prove negative results on the {\em simulation power} of models that are already known to be Turing universal~\cite{2HAMIU, Hendricks-2014,  Meunier-2014}.
Here we show that one can obtain a negative result on {\em Turing machine-style computation itself}, via a negative result on simulation between tile assembly systems.
Hence we show (for the first time) that simulation between tile assembly systems is a new method to obtain negative results on Turing computation in tile assembly.

\paragraph{Previous and future work.}
A large number of papers have conjectured or discussed that in one sense or another, sophisticated computation such as Turing machine simulation or building shapes with few tile types is impossible in the noncooperative model~\cite{RotWin00, Roth01, Cook-2011,  Doty-2011, Reif-2012, Manuch-2010, Patitz-2011, BMS-DNA2012a}. 
Our Theorem~\ref{thm:no_finite_TM} implies that any claimed simulator of Turing machines by noncooperative (temperature~1) systems would have to look very different from the known methods for cooperative abstract tile assembly model~\cite{Winf98,RotWin00,SolWin07} and its generalizations such as the two-handed~\cite{Versus,2HAMIU} or polygon~\cite{OneTile,geotiles} models, as well as variants of the noncooperative models, such as 3D tiles~\cite{Cook-2011}, probabilistic simulations~\cite{Cook-2011,Roth01}, negative glues~\cite{Patitz-2011}, staged and stepwise assembly~\cite{BMS-DNA2012a}, active signals~\cite{Signals,Jonoska2014}, polyomino-shaped tiles \cite{Fekete2014,Hendricks-2014} and polygons~\cite{gilbert2015continuous}).

Rothemund and Winfree~\cite{RotWin00} gave the first negative result on 2D temperature 1 systems: building an $N \times N$ square requires $N^2$ tile types if we insist that the square is {\em fully connected}. They conjecture this holds in the absence of that assumption. 
 Ma\v{n}uch, Stacho, and Stoll~\cite{Manuch-2010} show that 2D  temperature 1 systems {\em without mismatches} require at least $2N-1$ tile types to uniquely self-assemble $N\times N$ squares.  Tile assembly systems that always build a single terminal  assembly are said to be {\em directed}.
  Doty, Patitz and Summers~\cite{Doty-2011} conjecture that every directed~2D noncooperative system is {\em pumpable} meaning, roughly speaking, that every sufficiently long path of tiles has a segment that can be producibly repeated infinitely often. (They conjecture this for directed systems since  %that always build a single final assembly, called ``directed systems''; 
  by a result of Cook, Fu and Schweller~\cite{Cook-2011} we know that non-directed systems simulate Turing machines, with some error.) Their paper shows that if this conjecture holds then certain forms of computation (e.g.\ infinite computation) are impossible for directed~2D  temperature 1 systems. Proof of that conjecture would not imply our main results which are concerned with bounded (finite) computation and simulation, nor do our results imply that temperature 1 systems are pumpable (i.e.\ the present paper leaves the pumpability question open). Also our negative results do not make any assumptions about pumpablity,  mismatches nor directedness. 

As already noted, it has been shown~\cite{Meunier-2014} that noncooperative tile assembly can not simulate the cooperative model, here we answer the main open question from that paper (Conjecture~1.4~\cite{Meunier-2014}).

Meunier~\cite{meunier2015} gives  positive results for 2D noncooperative systems. First, by showing the existence of relatively simple noncooperative tile assembly systems that always build finite assemblies that contain a path where at least one tile type is repeated. A second, more general, construction gives for each real number $\epsilon > 0$, a tileset $T_\epsilon$ which, started from a single tile seed, produces only finite terminal assemblies, all of height $(2-\epsilon)|T_\epsilon|$.
So although general-purpose computation seems impossible at temperature~1, we know that one form of algorithmic self-assembly is possible, namely building long(ish) paths by re-using tile types.

One of the main reasons one simulates Turing machines with tile assembly systems is to build shapes.
Theorem~\ref{thm:no_finite_TM} shows that none of the standard ways to make shapes in models that are generalisations of the noncooperative model can possibly work in the noncooperative model itself. This gives one formal sense in which shape building via computation is impossible at temperature~1.
We leave (all!) others open. 

Beyond  self-assembly,  the combinatorics of self-avoiding walks in the plane, first introduced by Flory in 1953~\cite{Flory53} in the context of polymer chemistry,  has provided long-standing open problems attracting attention from mathematicians and computer scientists~\cite{knuth:math,Bousquet2010}. 
Our setting and results can be interpreted as memory-bounded versions of this topic: indeed, noncooperative self-assembly is exactly the process of building self-avoiding paths in $\mathbb{Z}^2$, but with a {\em memory} encoded by tile types. It would be interesting to see if our techniques could be applied to that domain to
shed an algorithmic light on the problem of counting or sampling self-avoiding walks.

There are a large number of papers on temperature 1 models that are generalisations of the classical temperature 1 model that we study~\cite{Cook-2011, Roth01, Patitz-2011,BMS-DNA2012a,Signals,Jonoska2014,Fekete2014,Hendricks-2014,gilbert2015continuous}.
Since those models achieve Turing universality, na\"ive application of our techniques to those models is  provably impossible. But often we care more about shape-building than computation and our techniques give a method to edit producible shapes, hence we ask:   Can our techniques be generalised to show limitations to the classes of shapes efficiently producible  in those models?  Another question: Is there a non-trivial hierarchy of simulation power within the noncooperative model? We leave this as an open research direction to further clarify and investigate the power of noncooperative self-assembly,\footnote{The question is not without merit, as recently it was shown that the two-handed, or hierarchical, model of self-assembly has an infinite set of hierarchies with each level in the hierarchy more power than the one below~\cite{2HAMIU}. However, the dearth of positive results for the temperature~1 abstract tile assembly model suggests positive results are unlikely to be easily found.} that would certainly require new techniques beyond what we've seen to date.

Our results do not close the problem of determining what can be built \emph{computationally} at temperature 1; there are many potential forms of computation that could in principle be exhibited by noncooperative systems beyond those formally encapsulated by Theorems~\ref{thm:main} and~\ref{thm:no_finite_TM}.
We close with a conjecture that attempts to eliminate many of these.
All known temperature~1 tile assembly systems that reuse tile types without producing infinite terminal assemblies produce assemblies that place tiles at a small Manhattan distance from the seed. 
For example, using $|T| = 2N-1$ tile types to build a size $N\times N$ square with Manhattan diameter $|T|+1$~\cite{RotWin00}, or, for all real numbers $\epsilon > 0$, using $|T_\epsilon|$ tile types to build finite terminal assemblies, all of height at least $(2-\epsilon)|T_\epsilon|$~\cite{meunier2015}.
 We conjecture that if a temperature~1 tile assembly system with~$|T|$ tile types produces only finite terminal assemblies, then these terminal assemblies place tiles at Manhattan distance no more than $2|T|$ from the seed.
This bound is just large enough so that the techniques exploited in this paper --- that require reuse of a visible glue type along a path of tiles --- could potentially find application, but small enough to almost meet the lower bound in~\cite{meunier2015}. More importantly, our conjectured bound severely limits the kinds of {\em finite} computations achievable in the temperature~1 abstract tile assembly model.

\subsection{Results}

We give an overview of our two main results, although a number of notions have yet to be formally defined (see Section~\ref{sec:defs} for definitions). Our first main result shows that the noncooperative abstract tile assembly model is not intrinsically universal: 

\renewcommand{\thmMain}{The noncooperative abstract tile assembly model is not intrinsically universal. In other words, there is no tileset $U$ that at temperature 1 simulates all noncooperative tile assembly systems.}
\begin{theorem}\label{thm:main}
\thmMain
\end{theorem}
\noindent The intuition behind the proof is given in Section~\ref{sec:intuition}, and the  proof is given in Sections~\oldref{sec:onepath} and~\oldref{sec:manypaths}.

Our second main result, that is almost immediate from our main theorem, shows that temperature 1 systems are severely limited in their ability to simulate Turing machines.  The standard published methods to simulate Turing machines in 2D in the abstract tile assembly model and its generalizations~\turinginothermodels, are (or can be easily modified to be) such that simulation of a $s(n)$ space bounded, and $t(n)$ time bounded Turing machine $M$ can be achieved in a $O(s(n)) \times O(t(n))$ rectangle with (a)~a seed assembly (encoding $M,x$)
contained in the leftmost $O(1)$ columns (coordinates), (b)~an output assembly (encoding the output of $M$ on input $x$) that
includes a unique tile type appearing on the rightmost column, and (c)~no tile ever goes outside this rectangle.  
The following theorem states, in a formal way,  that simulating Turing machines in a bounded rectangular region, without error and with the accept/reject answer given as a tile on the rightmost column is impossible for the  2D noncooperative abstract tile assembly model, for deterministic or even nondeterministic tile assembly systems. In the theorem statement it is important to note that the  ``bounding function'' $B_{M}$ is arbitrary in the sense that it allows a potential simulator tile assembly system to use much more space  than the actual running time or  space usage  of the Turing machine~$M$; this generality serves to strengthen the theorem statement (e.g.\ bounded Turing machine simulation is impossible even if we allow the tile assembly system to use, say, exponential, or doubly exponential, or indeed any finite spatial scaling). 

\renewcommand{\thmTM}{
Let $t: \N  \rightarrow\N $, $s:\N \rightarrow\N $ and let $B_M:\N\ra\N$  such that $\forall n\in\N$, $B_M(n)\geq s(n)$.
Let $M$ be any Turing machine that halts on all inputs $x \in \{ 0,1\}^\ast$ in time $t(|x|)$ using space~$s(|x|)$. 
There is no pair $(V, B_M)$ where  $V$ is a tileset and $B_M$  is a function  such that for all  $x\in \{0,1\}^\ast$, $|x|=n$,  there is a seed assembly $\sigma_{M,x}$ and tile assembly system $\mathcal{V}_x =(V,\sigma_{M,x},1)$  
 such that:
\begin{enumerate}
\item $\dom{\sigma_{M,x}}\subseteq \{0,1,\ldots,B_M(n)-1\}\times \{0,1,\ldots,B_M(n)-1\}$
\item for all $\alpha\in\termasm{\mathcal{V}}$, 
   $   \dom{\alpha} \subseteq \{0,1,\ldots, t(n) B_M(n)  -1\}  \times \{0,1,\ldots, B_M(n)-1 \}  $, 
$\dom{\alpha} \cap \left( \left\{ b+1 ,b +2 , \ldots, b + B_M(n) -1 \right\} \times \left\{ 0,1,\ldots,B_M(n)-1 \right\}\right) \neq \emptyset$
where $b=B_M(n) (t(n)-1)$ and 
$\alpha$ has at least one occurrence of a special tile type $H\in V$
on the rightmost column of $\dom{\alpha}$, and nowhere else, if and only if $M$ accepts $x$.
\end{enumerate} }
\begin{theorem}\label{thm:no_finite_TM}
\thmTM
\end{theorem}
\noindent The formalism simply states that there is no tile set $V$, such that when $V$ is instantiated as a noncooperative (temperature 1) tile assembly system $\mathcal{V}_x= (V,\sigma_{M,x},1)$, with an input seed assembly $\sigma_{M,x}$ (that somehow encodes a Turing machine $M$ and its input $x$), then $\mathcal{V}_x$ simulates~$M$ on $x$ within a finite rectangular region, writing a yes/no answer as $H$/``no tile'' anywhere on the rightmost column of tiles. Since the ``bounding function'' $B_M$ in the theorem statement can be arbitrarily large, the theorem holds even~if~we allow the noncooperative system to use an {\em arbitrarily large}, but finite, rectangular bounding box for the simulation. Section~\oldref{sec:intuition} gives an intuitive overview of the proof, and the actual proof is given in Section~\oldref{sec:no_finite_TM}. 

% !TEX root = t1notiu.tex
\section{Definitions and preliminaries}\label{sec:defs}
\label{definitions}

Let $\Z$ be the integers, $\mathbb{Z}^+ = \{1,2,3,\ldots \}$ and $\N = \{0,1,2,3,\ldots \}$.

When referring to the relative placements of positions in the grid graph of $\mathbb{Z}^2$, or in the plane~$\mathbb{R}^2$, we say that a position $P = (x_P, y_P)$ is \emph{to the right of} (\resp, \emph{to the left of}, \emph{above}, \emph{below}) of another position $Q = (x_Q, y_Q)$ if $x_P\geq x_Q$ (\resp $x_P\leq x_Q$, $y_P\geq y_Q$, $y_P\leq y_Q$).
This definition should not be confused with the definitions of right and left turns, nor with the definition of right-hand side and left-hand side, all of which are defined  below.

Moreover, unless stated otherwise, vectors of $\Z^2$ and $\R^2$ are \emph{column vectors}, i.e. $\vect{u} = \left(\begin{array}{c}x_u\\y_u\end{array}\right)$.

\subsection{Abstract tile assembly model}\label{sec:atam}

The abstract tile assembly was introduced by Winfree~\cite{Winf98}. In this paper we study a restriction of the abstract tile assembly model called the temperature 1 abstract tile assembly model, or noncooperative abstract tile assembly model. For definitions of the full model, as well as intuitive explanations, see for example~\cite{RotWin00,Roth01}.

A \emph{tile type} is a unit square with four sides,
each consisting of a glue \emph{type} and a nonnegative integer \emph{strength}. Let  $T$  be a a finite set of tile types.
In all sets of tile types used in this paper, we assume the existence of a well-defined total ordering that we call the {\em canonical ordering}.

The sides of a tile type are respectively called  north, east, south, and west, as shown  in the following picture:
\begin{center}
\vspace{-1ex}
\begin{tikzpicture}[scale=0.8]
\draw(0,0)rectangle(1,1);
\draw(0,0.5)node[anchor=east]{West};
\draw(1,0.5)node[anchor=west]{East};
\draw(0.5,0)node[anchor=north]{South};
\draw(0.5,1)node[anchor=south]{North};
\end{tikzpicture}
\vspace{-1ex}\end{center}

An \emph{assembly} is a partial function $\alpha:\mathbb{Z}^2\dashrightarrow T$ where $T$ is a set of tile types and the domain of $\alpha$ (denoted $\dom{\alpha}$) is connected.\footnote{Intuitively, an assembly is a positioning of unit-sized tiles, each from some set of tile types $T$, so that their centers are placed on (some of) the elements of the discrete plane $\mathbb{Z}^2$ and such that those elements of $\mathbb{Z}^2$ form a connected set of points.} 
A {\em tile} is a  pair $((x,y),t)  \in \mathbb{Z}^2 \times T$ where $(x,y)$ is a {\em position} and $t$ is a tile type. 
Hence the elements of an assembly are tiles. 
We let $\mathcal{A}^T$ denote the set of all assemblies over the set of tile types $T$. 
In this paper, two tile types in an assembly are said to  {\em bind} (or \emph{interact}, or are
\emph{stably attached}), if the glue types on their abutting sides are
equal, and have strength $\geq 1$.  An assembly $\alpha$ induces a
weighted \emph{binding graph} $G_\alpha=(V,E)$, where $V=\dom{\alpha}$, and
there is an edge $\{ a,b \} \in E$ if and only if $a$ and $b$ interact, and
this edge is weighted by the glue strength of that interaction.  The
assembly is said to be $\tau$-stable if every cut of $G$ has weight at
least $\tau$.

A \emph{tile assembly system} is a triple $\mathcal{T}=(T,\sigma,\tau)$,
where $T$ is a finite set of tile types, $\sigma$ is a $\tau$-stable assembly called the \emph{seed}, and
$\tau \in \mathbb{N}$ is the \emph{temperature}.
Throughout this paper,  $\tau=1$.

Given two $\tau$-stable assemblies $\alpha$ and $\beta$, we say that $\alpha$ is a
\emph{subassembly} of $\beta$, and write $\alpha\sqsubseteq\beta$, if
$\dom{\alpha}\subseteq \dom{\beta}$ and for all $p\in \dom{\alpha}$,
$\alpha(p)=\beta(p)$.
We also write
$\alpha\rightarrow_1^{\mathcal{T}}\beta$ if we can obtain $\beta$ from
$\alpha$ by the binding of a single tile type, that is:  $\alpha\sqsubseteq \beta$, $|\dom{\beta}\setminus\dom{\alpha}|=1$ and the tile type at the position $\dom{\beta}\setminus\dom{\alpha}$ stably binds to $\alpha$ at that position.  We say that $\gamma$ is
\emph{producible} from $\alpha$, and write
$\alpha\rightarrow^{\mathcal{T}}\gamma$ if there is a (possibly empty)
sequence $\alpha_1,\alpha_2,\ldots,\alpha_n$ where $n \in \N \cup \{ \infty \} $, $\alpha= \alpha_1$ and $\alpha_n =\gamma$, such that
$\alpha_1\rightarrow_1^{\mathcal{T}}\alpha_2\rightarrow_1^{\mathcal{T}}\ldots\rightarrow_1^{\mathcal{T}}\alpha_n$. 
A sequence of $n\in\mathbb{Z}^+ \cup \{\infty\}$ assemblies
$\alpha_0,\alpha_1,\ldots$ over $\mathcal{A}^T$ is a
\emph{$\mathcal{T}$-assembly sequence} if, for all $1 \leq i < n$,
$\alpha_{i-1} \to_1^\mathcal{T} \alpha_{i}$.

The set of \emph{productions}, or \emph{producible assemblies}, of a tile assembly system $\mathcal{T}=(T,\sigma,\tau)$ is the set of all assemblies producible
from the seed assembly $\sigma$ and is written~$\prodasm{\mathcal{T}}$. An assembly $\alpha$ is called \emph{terminal} if there is no $\beta$ such that $\alpha\rightarrow_1^{\mathcal{T}}\beta$. The set of all terminal assemblies of $\mathcal{T}$ is denoted~$\termasm{\mathcal{T}}$.  

As mentioned, in this paper $\tau=1$.  Also throughout this paper, we make the simplifying assumption that all glue types have strength 0 or 1: it is not difficult to see that this assumption does not change the behavior of the model (if a glue type $g$ has strength $s_g \geq 1$, in the $\tau =1$ model then a tile with glue type $g$ binds to a matching glue type on an assembly border irrespective of the exact value of $s_g$).

\subsection{Simulation between tile assembly systems and intrinsic universality} \label{sec:simulation_def}

To state our main result, we must formally define what it means for one tile assembly system to ``simulate'' another.  A number of definitions of simulation have been put forward for various self-assembly models~\cite{USA, IUSA, Meunier-2014, OneTile, 2HAMIU, Fekete2014}, here and in Appendix~\ref{sec:appendix:addSimDefs} we use those from~\cite{Meunier-2014}.

Let $T$ be a tile set, and let $m\in\Z^+$.
An \emph{$m$-block supertile} over $T$ is a partial function $\alpha : \Z_m^2 \dashrightarrow T$, where $\Z_m = \{0,1,\ldots,m-1\}$. Let $B^T_m$ be
the set of all $m$-block supertiles over~$T$.  The $m$-block with no domain is
said to be $\emph{empty}$.  For a general assembly $\alpha:\Z^2 \dashrightarrow T$
and $(x,y)\in\Z^2$, define $\alpha^m_{(x,y)}$ to
be the $m$-block supertile defined by $\alpha^m_{(x,y)}(x',y') = \alpha(mx+x',my+y')$ for
all $x',y'\in\{0,1,\ldots,m-1\}$.
For some tile set $S$, a partial function $R: B^{S}_m
\dashrightarrow T$ is said to be a \emph{valid $m$-block supertile
  representation} from $S$ to $T$ if for any $\alpha,\beta \in B^{S}_m$ such
that $\alpha \sqsubseteq \beta$ and $\alpha \in \dom R$, then $R(\alpha) =
R(\beta)$.

For a given
valid $m$-block supertile representation function $R$ from tile set~$S$ to tile
set $T$, define the \emph{assembly representation function}\footnote{Note that
  $R^*$ is a total function since every assembly of $S$ represents \emph{some}
  assembly of~$T$; the functions $R$ and $\alpha$ are partial to allow undefined
  points to represent empty space.}  $R^*: \mathcal{A}^{S} \rightarrow
\mathcal{A}^T$ such that $R^*(\alpha') = \alpha$ if and only if
$\alpha(A) = R\left(\alpha'^m_{A}\right)$
for all $A \in \Z^{2}$.  For an assembly $\alpha' \in \mathcal{A}^{S}$
such that $R^*(\alpha') = \alpha$, $\alpha'$ is said to map
\emph{cleanly} to $\alpha \in \mathcal{A}^T$ under $R^*$ if for all non empty
blocks $\alpha'^m_{A}$,
$A+\vec{u} \in \dom \alpha$ for some $\vec{u}\in\mathbb{Z}^2$ such that
$\|\vec{u}\|_2\leq 1$. %, or if $\alpha'$ has at most one non-empty $m$-block~$\alpha^m_{0,\ldots, 0}$.
In other words, $\alpha'$ may have tiles on supertile blocks representing empty
space in~$\alpha$, but only if that position is adjacent to a tile in $\alpha$.
We call such growth ``around the edges'' of~$\alpha'$ \emph{fuzz} and thus
restrict it to be adjacent to only valid supertiles, but not diagonally adjacent
(i.e.\ we do not permit \emph{diagonal fuzz}).

Below, let $\mathcal{T} = \left(T,\sigma_T,\tau_T\right)$
be a tile assembly system, let $\mathcal{S} = \left(S,\sigma_S,\tau_S\right)$ be a tile assembly system, and let $R$ be an $m$-block representation function $R:B^S_m
\rightarrow T$.

\begin{definition}
  \label{def-equiv-prod}
  \label{def:equiv-shape} We say that $\mathcal{S}$ and $\mathcal{T}$ have
\emph{equivalent terminal shapes} (under $R$) if $\{\dom{R^{*}(\alpha)} \mid \alpha\in\termasm{\mathcal S}\} = \{\dom{\beta} \mid \beta\in\termasm{{\mathcal T}_{n}}\}$. 
\end{definition}

Our main negative result on simulation (Theorem~\ref{thm:main}) shows that any claimed intrinsically universal noncooperative tileset $U$ does not satisfy Definition~\ref{def:equiv-shape} when used to simulate  certain  noncooperative tile assembly systems. Intrinsically universal tilessets must satisfy Definition~\ref{def:equiv-shape} (see Observation~\ref{obs:prodshapes}) 
 and moreover must satisfy a significantly stronger set of definitions than Definition~\ref{def:equiv-shape}; such stronger definitions are given in Appendix~\ref{sec:appendix:addSimDefs}.

\subsection{Paths and non-cooperative self-assembly}\label{sec:defs-paths}
This definition sections introduces quite a number of key definitions and concepts that will be used extensively throughout the paper. 

Let $n \in \mathbb{N}$ and let $T$ be a set of tile types.  As already defined in Section~\ref{sec:atam}, a {\em tile} is a pair $((x,y),t)$ where $(x,y) \in \mathbb{Z}^2$ is a position and $t\in T$ is a tile type. 

Intuitively, a path is a finite or one-way-infinite simple (non-self-intersecting) sequence of tiles placed on points of $\mathbb{Z}^2$ so that each tile in the sequence interacts with the previous one, or more precisely: 

\begin{definition}[Path]\label{def:path}
  A {\em path} is a (finite or infinite) sequence  $P = P_0 P_1 P_2 \ldots$  of tiles   $P_i = ((x_i,y_i),t_i) \in \mathbb{Z}^2 \times T$, such that:
\begin{itemize}
\item for all $P_j$ and $P_{j+1}$ defined on $P$ it is the case that~$t_{j}$ and~$t_{j+1}$ interact, and
\item for all $P_j,P_k$ such that $j\neq k$ it is the case that $ (x_j,y_j) \neq (x_k,y_k)$.
\end{itemize}
\end{definition}

Whenever $P$ is finite, i.e. $P = P_0P_1P_2\ldots P_{n-1}$ for some $n$, $n$ is termed the {\em length} of $P$.
By definition, paths are simple (or self-avoiding), and this fact will be repeatedly used through the paper. A \emph{position} of $P$ is an element of $\mathbb{Z}^2$ that appears in $P$ (and therefore appears exactly once), and an \emph{index} $i$ of $P$ is simply an integer in $\{0,1,\ldots,n-1\}$.
For a path $P = P_0 \ldots  P_i P_{i+1} \ldots P_j  \ldots $, we define the notation $P_{i,i+1,\ldots,j} = P_i P_{i+1} \ldots P_j$, i.e.\  ``the subpath of $P$ between indices $i$ and $j$, inclusive''.

Although a path is not an assembly, we know that each adjacent pair of tiles in the path sequence interact implying that the set of path positions forms a connected set in $\Z^2$ and hence every path uniquely represents an assembly containing exactly the tiles of the path, more formally:  
For a path $P = P_0 P_1 P_2 \ldots$ we define the set of tiles  $\pathassembly{P} = \{ P_0, P_1, P_2, \ldots\}$ which we observe is  an assembly\footnote{I.e.  $\pathassembly{P}$ is  a partial function from $\Z^2$ to tile types that is defined on a connected set.} and  we call $\pathassembly{P}$ a {\em path assembly}. 
A {\em path $P$ is said to be producible} by some tile assembly system $\calT = (T,\sigma,1)$ if the assembly $(\asm{P} \cup \sigma ) \in \prodT$ is producible, and we call such a $P$ a {\em producible path}. 
We define $$\prodpathsT = \{ P  \mid P \textrm{ is a path and } (\asm{P}\cup\sigma) \in \prodT \} $$  to be the set of producible paths of $\calT$.\footnote{Intuitively, although  producible paths  are not assemblies, any  producible path $P$ has the nice property that it encodes an unambiguous description of how to grow $\asm{P}$ from the seed $\sigma$, in ($P$) path  order, to produce  the assembly $\sigma \cup \asm{P}$.} 

For any path $P = P_0 P_1 P_2, \ldots$ and integer $i\geq 0$, we write $\pos{P_i} \in \mathbb{Z}^2$, or  $(x_{P_i},y_{P_i}) \in \mathbb{Z}^2$, for the position of $P_i$ and $\type{P_i}$ for the tile type of $P_i$. Hence if  $P_i = ((x_i,y_i),t_i) $ then $\pos{P_i} =  (x_{P_i},y_{P_i}) = (x_i,y_i) $ and $\type{P_i} = t_i$.

If two paths, or two assemblies, or a path and an assembly, share a common position we say they {\em intersect} at that position. Furthermore, we say that two paths, or two assemblies, or a path and an assembly,  {\em agree} on a position if they both place the same tile type at that position and {\em conflict} if they place a different tile type at that position. 

Note that, since the domain of a producible assembly is a connected set in $\Z^2$, and since in an assembly sequence of some TAS $\calT = (T,\sigma,1)$ each tile binding event $\beta_i \rightarrow_1^\mathcal{T} \beta_{i+1} $ adds a single node $v$ to the binding graph $G_{\beta_{i}}$ of $\beta_i$  to give a new binding graph $G_{\beta_{i+1}}$, and adds at least one  weight-1 edge joining $v$ to the subgraph $G_{\beta_i} \in G_{\beta_{i+1}}$, then for any tile $((x,y),t) \in \alpha$  in a producible assembly $\alpha \in \prodT$, there is a edge-path (sequence of edges) in the binding graph of $\alpha$ from $\sigma$ to $((x,y),t)$. From there, the following important fact about temperature 1 tile  assembly is straightforward to see.

\begin{observation}
Let $\calT = (T,\sigma,1)$ be a tile assembly system and let $\alpha \in \prodT$. 
For any tile $((x,y),t) \in \alpha$ there is a producible path $P \in \prodpathsT$ that for some $i \in \N$ contains $P_i = ((x,y),t)$.
\end{observation}

For $A,B\in\mathbb{Z}^2$, we define $\vect{AB} = B - A$ to be the vector from $A$ to $B$, and  for two tiles $P_i = ((x_i,y_i),t_i)$ and $P_j = ((x_j,y_j),t_j)$ we define $\vect{P_i P_j} = \pos{P_j} - \pos{P_i}$ to mean the vector from $\pos{P_i}=(x_i,y_i)$ to $\pos{P_j}=(x_j,y_j)$.
The translation of a path $P$ by a vector $\vect{v} \in \mathbb{Z}^2$, written $P+\vect{v}$, is  the path $Q$ where
and for all indices $i$ of $P$,  
$\pos{Q_i}=\pos{P_i}+\vect{v}$ 
and 
$\type{Q_i}=\type{P_i}$. 
As a convenient notation, for a path $PQ$ composed of subpaths $P$ and $Q$,  when we write $PQ +\vect{v} $ we mean $(PQ)+\vect{v}$ (i.e.\ the translation of all of $PQ$ by $+\vect{v}$).
The translation of a path $P$ by a vector $\vect{v} \in \mathbb{Z}^2$, written $P+\vect{v}$, is  the path $Q$ where
and for all indices $i$ of $P$,  

The translation of an assembly $\alpha$ by a vector $\vect{v}$, written $\alpha+\vect{v}$, is the assembly $\beta$ defined on the set  $\dom\alpha+\vect{v}$ as $\beta(x,y)=\alpha((x,y)-\vect{v})$ where $(x,y) \in \Z^2$.
A \emph{column} $x\in \mathbb{Z}$ is the set of all points of $\mathbb{Z}^2$ with x-coordinate $x$, and a \emph{row} $y\in\mathbb{Z}$ is the set of all points of $\mathbb{Z}^2$ with y-coordinate~$y$.

\newcommand{\torture}{\ensuremath{P_{i+1+ ( (k-i-1)\! \mod(j-i) )}+ \lfloor (k-i-1)/ (j-i) \rfloor\vect{P_iP_j}}}

Next, for a path $P$ and two indices $i,j$ on $P$, we will define a (not necessarily producible) sequence called the pumping of $P$ between $i$ and $j$. 
\begin{definition}[pumping of $P$ between $i$ and $j$]
  \label{def:pumpingPbetweeniandj}
Let $\calT = (T,\sigma,1)$ be a tile assembly system and $P\in\prodpathsT$.
We say that the \emph{``pumping of $P$ between $i$ and $j$''} is the sequence $\olq$ of elements from $\Z^2\times T$ defined by:

\begin{equation*}
\olq_k =
\begin{cases}
 P_k &\qquad \textrm{for } 0\leq k \leq i \\
 \torture &\qquad \textrm{for }  i < k  ,
\end{cases}
\end{equation*}
\end{definition}

Hence, intuitively, $\olq$ has two parts. It begins with a  finite sequence $P_{0,1,\ldots,i}$. Then appended to that, there is an infinite  sequence where the tile types appear with positions at regular intervals in the plane.    
We formalize the latter intuition in the following Lemma:
\begin{lemma}
  \label{lem:torture}
  Let $P$ be a path with tiles from some tileset $T$, $i<j$ be two integers, and $q$ be the pumping of $P$ between $i$ and $j$.
  Then for all integer $k\geq i$, $q_{k+(j-i)} = q_k + \vect{P_iP_j}$.
\end{lemma}
\begin{proof}
  By the definition of $q$:
  \begin{eqnarray*}
    q_{k + (j-i)} &=& P_{i+1+((k + (j-i) - i - 1) \mod (j-i))} + \left\lfloor\frac{k + (j-i) - i-1}{j-i}\right\rfloor \vect{P_iP_j}\\
        &=& P_{i+1+((k - i - 1) \mod (j-i))} + \left\lfloor\frac{k-i-1}{j-i} + 1\right\rfloor \vect{P_iP_j}\\
        &=& P_{i+1+((k - i - 1) \mod (j-i))} + \left(\left\lfloor\frac{k-i-1}{j-i} \right\rfloor + 1\right) \vect{P_iP_j}\\
        &=& q_{k}+ \vect{P_iP_j}
      \end{eqnarray*}
\end{proof}

The following  definition gives the notions of  pumpable and finitely pumpable that are used in our proofs. It is followed by a less formal but more intuitive description.

\begin{definition}[Pumpable]\label{def:pumpable}
Let $\calT = (T,\sigma,1)$ be a tile assembly system.  We say that a producible path $P \in\prodpathsT$, is {\em infinitely pumpable}, or simply {\em pumpable}, if there are two integers $i<j$ such that the pumping of $P$ between $i$ and $j$ is a producible (infinite) path, i.e. $\olq \in\prodpathsT$.
\end{definition}

In other, more intuitive, words, a producible path $P\in\prodpathsT$ is \emph{infinitely pumpable}, or simply \emph{pumpable}, if there is a producible infinite assembly $\alpha \in \prodT$ and two indices  $i<j$ on~$P$, such that~$\alpha$ contains exactly $\sigma$, then $\asm{P_{0,1,\ldots,j}}$, and then infinitely many occurrences of the ``pumpable segment'' $\asm{ P_{i,i+1,\ldots,j-1}}$ each translated by successive positive integer multiples $\{1,2,3,\ldots \}$ of $\vect{P_iP_j}$, where these occurrences do not intersect $\sigma$, $\asm{P_{0,1,\ldots,j-1}}$ or themselves, each tile along this path assembly is bound to the previous, and $\alpha$ contains no other tiles.\footnote{We remark that  this definition of ``infinitely pumpable'' intentionally excludes pumping that intersects with and {\em agrees with} the seed, $P_{0,1,\ldots j}$ or some translated (pumped) segment.}

For all $i$ such that both $P_i$ and $P_{i+1}$ are defined (i.e. for all $i\in \mathbb{N}$ if $P$ is infinite, and for all $i < |P|-1$ otherwise), we define the
``\emph{output side of $P_i$}'' to be the side of $\type{P_i}$ adjacent to $\type{P_{i+1}}$, and for all
$i>0$, we define the ``\emph{input side of $P_i$}'' to be the side  of $\type{P_i}$ adjacent to~$\type{P_{i-1}}$. The sides of $\type{P_i}$ that are neither output sides nor input sides of $P_i$ are said to be \emph{free}, as are the glues of those sides.\footnote{By this definition of input and output sides the first tile of a path does not have an input side, and the last one does not have an output side. We also remark that this definition of input/output sides is defined relative to a specific  \emph{path}, and is {\em not} a property of the tiles themselves; moreover, the tiles, including the first and last tiles, may have other glue types, i.e.\ {\em free} glue types, not used by the path.  Despite the fact free sides may have strength 1 glue types we typically ignore this in our analysis of paths---this is because our proofs typically analyse paths one at a time and thus require us to consider only the non-free tiles sides that actually bind the tiles along the path assembly and thus don't require us to make statements about free sides.}

Let $P=P_0P_1\ldots $. For $i > 0 $, we say that a \emph{right turn} (\resp \emph{left turn})  from $P$ at index~$i$  is a path with prefix $P_0P_1\ldots P_i x$ for some $x\in\mathbb{Z}^2\times T$ adjacent to $P_i$  such that orientated in the direction $\vect{P_{i-1}P_{i}}$, $\pos{x}$ is clockwise (\resp anti-clockwise) from $P_{i+1}$. More formally, let $\vec{u} = \vect{P_{i}P_{i-1}}$ (the unit column vector from $\pos{P_{i}}$ to $\pos{P_{i-1}}$), let  
$\rho	 = \bigl(\begin{smallmatrix}
0&1 \\ -1&0
\end{smallmatrix} \bigr)$, and let $\tau = (\rho\cdot \vec{u}, \rho\cdot\rho\cdot \vec{u}, \rho\cdot\rho\cdot\rho\cdot \vec{u})$, then we say that $P_0P_1\ldots P_i x$ is a right turn from $P_0P_1\ldots P_i P_{i+1}$ if $\vect{P_{i} x}$ appears after $\vect{P_{i} P_{i+1}}$ in $\tau$.

We define $\glue{i}{i+1} = (g, i)$  [i.e.\ $(g, i)$ is a pair of the form (glue type, path index)] where $g$ is the  shared glue type between consecutive tiles $P_i$ and $P_{i+1}$ on the path $P$. 
When we say ``glue'' in the context of a path, we mean a pair of the form (glue type, path index). 
We define $\type{\glue{i}{i+1}} = g$ to denote the glue type of $\glue{i}{i+1}$, 
we write $\pos{\glue{i}{i+1}} = (\pos{P_i}, \pos{P_{i+1}})$ (the ``position of glue $\glue{i}{i+1}$'') to denote the edge (of the grid graph of $\mathbb{Z}^2$) of $\glue{i}{i+1}$, oriented from $\pos{P_i}$ to $\pos{P_{i+1}}$.
Moreover, for $A,B \in \R^2$ we define $\midpoint{A}{B} = A+\frac{1}{2}\vect{AB} \in \mathbb{R}^2$ to be the  midpoint of the line segment $[A,B]\subsetneq \mathbb{R}^2$, and for a pair of tiles $P_i,P_j$ we define $\midpoint{P_i}{P_j}$ to be the midpoint 
of the line segment $[\pos{P_i},\pos{P_j}]\subsetneq \mathbb{R}^2$.

\begin{definition}[The right priority path of a set of paths]\label{def:rp}
Let $P$ and $Q$, where $P \neq Q$, be two paths 
with $\pos{P_0} = \pos{Q_0} $ and $\pos{P_1} = \pos{Q_1}$.  
Let $i$ be the smallest index such that $i \geq 0$ and  $P_i\neq Q_i$.
We say that $P$ is the {\em right priority path} of $P$ and $Q$  if either (a) $P_{0,1,\ldots, i}$ is a right turn from $Q$ or (b) $\pos{P_i}=\pos{Q_i}$ and the type of $P_i$ is smaller than the type of $Q_i$ in the canonical ordering of tile types.

For sets of paths, we extend this definition as follows: let $p_{0} \in \mathbb{Z}^{2},  p_{1} \in \mathbb{Z}^{2}$ be two adjacent positions. If $S$ is a set of paths such that for all $P\in S$, $P_0=p_0$ and $P_1=p_1$, we call the \emph{right-priority path of $S$} the path that is right-priority path of all other paths~in~$S$.

The {\em left priority path of a set of paths} is defined symmetrically: swap left for right in Definition~\ref{def:rp}.

\end{definition}

\newcommand{\range}[1]{\mathrm{range}(#1)}

\subsubsection{Curves: embedding paths in $\R^2$}

A {\em curve}, or a {\em curve in $\mathbb{R}^2$},  is defined to be a continuous function $f:[0,1] \rightarrow \R^2$. We say that~$f$ is continuous at some $x_0\in \dom{f}$ if $\forall \varepsilon, \exists \eta, \forall x, |x-x_0|\leq \eta\Rightarrow \|f(x)-f(x_0)\|_2\leq \varepsilon$, where for all $a,b\in[0,1]$, $\|(a,b)\|_2 = \sqrt{a^2+b^2}$, and we say that $f$ is continuous if and only if $f$ is continuous at all $x_0\in [0,1]$.

Intuitively, we will define the concatenation of a finite sequence of curves $f_0,f_1,\ldots,f_{k-1}$ to be  a function $F : [0,1] \rightarrow \R^2 $ that for each $i \in \{0,1,\ldots, k-1 \}$ represents $f_i$ by rescaling the  domain of  $f_i$ to be in the interval $[\frac{i}{k}, \frac{i+1}{k} ]$. Thus $F$ is defined on $[0,1]$ and has range  $\bigcup_{i=0}^{k-1} \left( \range{f_i} \right)$.
 This is defined as follows:
\begin{definition}[Concatenation of curves in $\R^2$]\label{def:concat}
Given a finite sequence of $k\in\N$ curves $f_0,f_1,\ldots,f_{k-1}$ in $\R^2$ their {\em concatenation} is the function $F: [0,1] \rightarrow \R^2$ defined for all $i$ such that $0\leq i< k$ and all $x\in \left[\frac{i}{k}, \frac{i+1}{k}\right]$ as $F(x) = f_i(x k-i) $. 
\end{definition}

For example, Figure~\ref{fig:rhs}(c) shows the concatenation of two curves:  the curve in Figure~\ref{fig:rhs}(b) and a unit-length vertical line segment. 

The following observation states that the concatenation $F$ of $k$ continuous functions $f_0,f_1,\ldots,f_{k-1}$, that have the property  $f_i(1)= f_{i+1}(0)$ for $0\leq i < k-1 $, is itself a continuous function and  although the proof is straightforward, it is worth explicitly stating since it is used extensively in this paper: 
\begin{observation}\label{obs:concatGivesCurve}
Let $f_0,f_1,\ldots,f_{k-1}$ be a finite sequence of curves in $\R^2$ that have the property that for all $i \in \{0,1,\ldots,k-2 \}$, $f_i(1)= f_{i+1}(0)$ and let  $F$ be  the concatenation of $f_0,f_1,\ldots,f_{k-1}$. Then $F$ is  a curve in $\R^2$. 
\end{observation}
\begin{proof}
First note that  for each $i\in \{0,1,\ldots,k-1 \}$,  $f_i$ is a continuous function and that rescaling (shrinking) the domain of $f_i$ from $[0,1]$ to $[\frac{i}{k}, \frac{i+1}{k} ]$ preserves continuity. Secondly, since for each $i\in \{0,1,\ldots,k-2 \}$, $f_i(1)= f_{i+1}(0)$ and  $F$ contains $ \frac{i+1}{k} \rightarrow f_i(1) $, the function $F$ is continuous on the (``double-length'') interval  $[\frac{i}{k} , \frac{i+2}{k}]$. Since this holds for all such $i$, $F$ is continuous on its entire domain $[0,1]$, and thus is a curve in~$\R^2$.\end{proof}

\begin{observation}\label{obs:ClosedSimpleCurve}
Let $f_0,f_1,\ldots,f_{k-1}$ be a finite sequence of finite-length  simple curves in $\R^2$ such that for all $i \in \{0,1,\ldots,k-2 \}$, $f_i(1)= f_{i+1}(0)$, also $f_{k-1}(1) = f_0(0)$, and 
those $k$ nonempty intersections between $f_0,f_1,\ldots,f_{k-1}$  are the only nonempty intersections between them. 
Let  $F$ be  the concatenation of $f_0,f_1,\ldots,f_{k-1}$. 
Then $F$ is  a finite-length closed simple curve in $\R^2$. 
\end{observation}
\begin{proof}  The hypotheses of Observation~\ref{obs:concatGivesCurve} are satisfied hence $F$ is a curve. 
$F$ is composed of a finite set of $k$ finite length component curves so $F$ is of finite length. 
$F$ is closed because for $i\in \{0,1,\ldots,k-2 \}$,  $f_i(1)= f_{i+1}(0)$, and  $f_{k-1}(1) = f_0(0)$, and $F$ is simple because those are the only nonempty intersections between $f_0,f_1,\ldots,f_{k-1}$. 
\end{proof}

Also, we will sometimes need other curves that are not defined by paths:

\begin{definition}[Line segment]
  Let $A, B\in\R^2$. The \emph{line segment from $A$ to $B$}, which we write $[A,B]$, is the curve defined for all $x\in[0,1]$ by $f(x) = A+x\vect{AB}$.
  
\end{definition}
\begin{definition}\label{def:embed}
For any path $P$ we define $\frak{E}_P$ to be the {\em canonical embedding} of $P$ where $\frak{E}_{P}:[0,1]\rightarrow\mathbb{R}^2$, such that
for all $s$ such that $0\leq s < 1$  
$$\frak{E}_P(s) = \pos{P_{\lfloor s\cdot(|P|-1)\rfloor}} + (s\cdot(|P|-1) - \lfloor s\cdot(|P|-1)\rfloor) \vect{P_{\lfloor s\cdot(|P|-1)\rfloor}P_{\lfloor s\cdot(|P|-1)\rfloor + 1}}$$ 
and 
$$\frak{E}_P(1) = \pos{P_{|P|-1}} \, .$$
\end{definition}

Note that by Definition~\ref{def:embed}, the canonical embedding of a path  is a curve, i.e. the canonical embedding is a continuous function from $[0,1]$ to $\R^2$. Figure~\ref{fig:rhs}(a) shows an example path $P$ and Figure~\ref{fig:rhs}(b) shows its canonical embedding $\frak{E}_P$. 

This paper frequently uses the Jordan curve theorem, which is a statement about curves in~$\mathbb{R}^2$: any  simple closed (and hence finite) curve in $\R^2$ partitions $\R^2$ into exactly two connected components, a bounced one and an unbounded one.

In our proofs, we will often reason about right turns and left turns from a curve, and also about on which side of a closed simple curve is the bounded connected component. Since all of the  closed simple curves $c$ we will define will be simple finite polygons (with a finite number of finite-length sides), their left-hand side and right-hand side can be defined by taking any point $A$ on a segment of the polygonal curve $c$, not at a corner, and reasoning as follows. Since $c$ is locally a straight line around $A$, $c$ is differentiable at $A$. Also $c$ has a direction (from domain element 0 to domain element 1). The left-hand side of $c$ is therefore the connected component to the left of $A$ when positioned at $A$ orientated in the direction from $0$ to $1$ along $c$, and the right-hand side of $c$ is the connected component to the right of $A$. By defining curves within a very small distance of $c$,  we can show that the left-hand side of $c$ is connected, and the right-hand side of $c$ is also connected. For example, Figure~\ref{fig:rhs}(c) shows  such a polygonal closed simple curve $c$, with its left-hand side highlighted in grey.

\begin{figure}[t]
\begin{center}
\begin{adjustbox}{max width=1\textwidth}
\tikzset{cross/.style={cross out, draw=black, minimum size=2*(#1-\pgflinewidth), inner sep=0pt, outer sep=0pt},
cross/.default={1.5pt}}
\newcommand{\rhsFigScale}{0.64}
\newcommand{\rhsFigSpace}{9ex}

\begin{tikzpicture}[scale=\rhsFigScale,font=\small]

\begin{scope}%[local bounding box=figrhsScopeA]

\foreach \i in {-2,...,4} 
  \foreach \j in {-1,...,5} {
     \draw (\i,\j) node[cross,rotate=45] {} ;
}

% draw tiles
\draw[fill=black!10!white](2.5,3.5) rectangle (3.5,4.5); %first tile
\foreach \i in {-1,...,2} {% top row of tiles
   \draw[fill=white](\i-0.5,3.5) rectangle (\i+0.5,4.5);
}
\draw[fill=white](-1.5,2.5) rectangle (-0.5,3.5);
\draw[fill=white](-1.5,1.5) rectangle (-0.5,2.5);
\draw[fill=white](-0.5,1.5) rectangle (0.5,2.5);
\draw[fill=white](-0.5,0.5) rectangle (0.5,1.5);
\draw[fill=white](0.5,0.5) rectangle (1.5,1.5);  
\draw[fill=white](1.5,0.5) rectangle (2.5,1.5);  
\draw[fill=white](1.5,-0.5) rectangle (2.5,0.5);  
\foreach \i in {0,...,2} {% long vertical line on right of pic
   \draw[fill=white](2.5,\i-0.5) rectangle (3.5,\i+0.5);
}
\draw[fill=white](2.5,2.5) rectangle (3.5,3.5);

\draw(3,4)node[]{$P_0$};
\draw(2,4)node[]{$P_1$};
\draw(1,4)node[]{$P_2$};
\draw(0,4)node[]{$P_3$};
\draw(-1,4)node[]{$P_4$};
\draw(-1,3)node[]{$P_5$};
\draw(-1,2)node[]{$P_6$};
\draw(0,2)node[]{$P_7$};
\draw(0,1)node[]{$P_8$};
\draw(1,1)node[]{$P_9$};
\draw(2,1)node[]{$P_{10}$};
\draw(2,0)node[]{$P_{11}$};
\draw(3,0)node[]{$P_{12}$};
\draw(3,1)node[]{$P_{13}$};
\draw(3,2)node[]{$P_{14}$};
\draw(3,3)node[] (RHSa) {$P_{15}$};

\draw(-1.5,-0.6)node{(a)};

\end{scope}
\begin{scope}[xshift=9cm]

\foreach \i in {-2,...,4} 
  \foreach \j in {-1,...,5} {
     \draw (\i,\j) node[cross,rotate=45] {} ;
}

\draw[very thick,->](3,4)--(-1,4)--(-1,2)--(0,2)--(0,1)--(2,1)--(2,0)--(3,0)--(3,3);
%-1.25,4.25
\draw(-1.3,4.3)node[]{$\frak{E}_P$};
\draw(3.8,4.3)node{$\frak{E}_P(0)$};
\draw(4,2.7)node{$\frak{E}_P(1)$};

\draw(-1.5,-0.6)node{(b)};

\end{scope}
\begin{scope}[xshift=18cm]

\draw[fill=black!10!white](3,4)--(-1,4)--(-1,2)--(0,2)--(0,1)--(2,1)--(2,0)--(3,0)--cycle;

\foreach \i in {-2,...,4} 
  \foreach \j in {-1,...,5} {
     \draw (\i,\j) node[cross,rotate=45] {} ;
}

% embedded curve
\draw[very thick,->](3,4)--(-1,4)--(-1,2)--(0,2)--(0,1)--(2,1)--(2,0)--(3,0)--(3,4);

\draw[->](1.5,1)--(1.5,1.5);
\draw(1.5,1.5)node[anchor=south]{\tiny LHS};

\draw[->](0,1.5)--(-0.5,1.5);
\draw(-0.5,1.5)node[anchor=east]{\tiny RHS};

\draw[->](-0.5,4)--(-0.5,3.5);
\draw(-0.5,3.5)node[anchor=north]{\tiny LHS};

\draw(-1.25,4.25)node[]{$c$};
\draw(3.0,4.4)node{$c(0)\!=\!c(1)$};

\draw(-1.5,-0.6)node{(c)};
\end{scope}

\end{tikzpicture}
\end{adjustbox}
\end{center}
\caption{Crosses denote points in $\Z^2$.  
(a)~A path $P=P_0,P_1,\ldots, P_{15} $. 
(b)~The curve $\frak{E}_P$,~the canonical embedding of $P$  in $\R^2$.
(c)~The curve $c$ defined as the concatenation  of  $\frak{E}_P$ and the unit-length line segment $[\pos{P_{15}}, \pos{P_{0}} ] = [\frak{E}_P(1),\frak{E}_P(0)]$. Two endpoints of $c$ are identical, $c(0)=c(1)$, moreover $c$  is a closed simple curve in $\R^2$.  A single thick arrowhead on $c$ indicates the direction of~$c$ and the position of~$c$'s identical endpoints. 
By the Jordan curve theorem, such a closed simple curve in $\mathbb{R}^2$ partitions the plane into two connected components exactly one of which is `bounded' (has finite area).  
The `right-hand side' (RHS) and `left-hand side' (LHS) of~$c$ are indicated by small thin arrows with the LHS being the bounded component  highlighted in grey.}
\label{fig:rhs}
\end{figure}

% !TEX root = t1notiu.tex

\section{A family of tile assembly systems \texorpdfstring{$\calT_{N}$}{T\_N}}
\label{sec:T}

Definition~\ref{def:tn} defines a (very simple) infinite family of noncooperative tile assembly systems $\{\mathcal{T}_N \mid N\in \Z^+  \}$. The proof of our main theorem shows that there is no tile set $U$ that for all~$N$ simulates $\calT_{N}$.
\begin{definition}\label{def:tn}
For each $N\in \Z^+$, let ${\cal T}_N=(T_N,\sigma_N,1)$ be the tile assembly
system that assembles the infinite assembly shown in Figure~\ref{fig:tn}.
\end{definition}
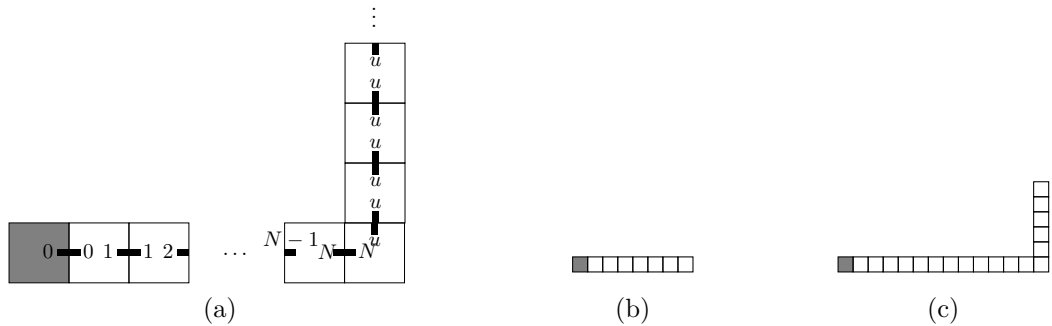
\begin{figure}[ht]
\begin{subfigure}[t]{0.42\textwidth}
\begin{center}
\includegraphics[scale=0.8,viewport=0 50 200 180]{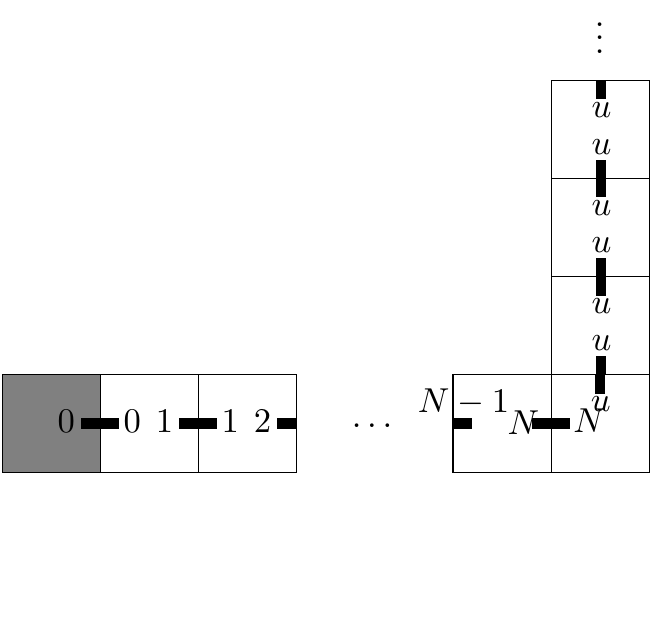}
\end{center}
\vspace{-2ex}
\caption{}
\end{subfigure}
\begin{subfigure}[t]{0.25\textwidth}
\begin{center}
\begin{tikzpicture}[scale=0.2]\draw[fill=gray](0,0)rectangle(1,1);\draw(1,0)rectangle(2,1);\draw(2,0)rectangle(3,1);\draw(3,0)rectangle(4,1);\draw(4,0)rectangle(5,1);\draw(5,0)rectangle(6,1);\draw(6,0)rectangle(7,1);\draw(7,0)rectangle(8,1);\end{tikzpicture}
\end{center}
\vspace{-2ex}
\caption{}
\end{subfigure}
\begin{subfigure}[t]{0.25\textwidth}
\begin{center}
\begin{tikzpicture}[scale=0.2]\draw[fill=gray](0,0)rectangle(1,1);\draw(1,0)rectangle(2,1);\draw(2,0)rectangle(3,1);\draw(3,0)rectangle(4,1);\draw(4,0)rectangle(5,1);\draw(5,0)rectangle(6,1);\draw(6,0)rectangle(7,1);\draw(7,0)rectangle(8,1);\draw(8,0)rectangle(9,1);\draw(9,0)rectangle(10,1);\draw(10,0)rectangle(11,1);\draw(11,0)rectangle(12,1);\draw(12,0)rectangle(13,1);\draw(13,0)rectangle(14,1);\draw(13,1)rectangle(14,2);\draw(13,2)rectangle(14,3);\draw(13,3)rectangle(14,4);\draw(13,4)rectangle(14,5);\draw(13,5)rectangle(14,6);\end{tikzpicture}
\end{center}
\vspace{-2ex}
\caption{}
\end{subfigure}
\vspace{2ex}
\caption{(a) The ``flipped-L'' TAS ${\cal T}_N = (T_N,\sigma_N,1)$, that
  deterministically assembles a single, infinite assembly. 
  $T$ contains $N+3$ tile types as shown, and the grey tile is the seed $\sigma_N$ which is placed at the origin $(0,0)$. 
  $\calT_N$ grows from the seed, distance $N+1$ to the east, and then grows infinitely to the north.
  Hence~$\calT_{N}$ builds an infinite path assembly. 
  Productions of an example such tile assembly system for $N=12$ are shown: (b) $\mathcal{T}_{12}$ after~7 tile additions, and (c) $\mathcal{T}_{12}$ after~18 tile additions.}
\label{fig:tn} 
\end{figure}

% !TEX root = t1notiu.tex

\section{Intuition behind the proofs of Theorems~\ref{thm:main} and~\ref{thm:no_finite_TM} }\label{sec:intuition}

We begin with a description of the high-level intuition behind the proof of our main result, Theorem~\ref{thm:main}. One of the main difficulties of this result is that for {\em any} finite number of non-cooperative tile assembly systems, there is in fact a single non-cooperative simulator for all of them:  simply let the tiles of the simulator be the disjoint union of all tilesets of the simulated systems. Moreover, it is known~\cite{Meunier-2014} that in 3D, there is a tileset, operating at temperature 1 (i.e.\ noncooperative), that simulates all non-cooperative tile assembly systems. 
Hence our proof is going to crucially make use of the fact that any claimed simulator tileset is of finite size, and must work in the plane. 

First, we assume, for the sake of contradiction, that there is a single tileset~$U$, that simulates all noncooperative (or temperature 1) tile assembly systems. Hence, in particular, the tileset $U$ simulates the class of systems $\{ \calT_{N} \mid N \in \Z^+ \}$ described in Section~\ref{sec:T} and shown in Figure~\ref{fig:tn}. Hence, for all $\calT_{N}$, there is a tile assembly system $\mathcal{U}_{\calT_N} = (U,\sigma_{\calT_N},1)$ and scale factor $m \in \Z^+$ such that~$\mathcal{U}_{\calT_N}$ simulates~$\mathcal{T}_N$.\footnote{\label{ft:dropTN}Later in the paper we drop the $\calT_N$ subscripts from $\calU_{\calT_N}, \sigma_{\calT_N}$ and simply say that $\mathcal{U} = (U,\sigma,1)$ simulates~$\calT_N$.}
In particular, by the definition of simulation (Section~\ref{sec:simulation_def}), and in particular by Definition \ref{def-equiv-prod}, this implies that for each terminal producible assembly~$\alpha$ of~$\mathcal{T}_N$ there is a producible assembly $\alpha' \in  \mathcal{U}_{\calT_N}$ that represents the shape of $\alpha$, and vice-versa.  Figure~\ref{fig:fuzz} shows what such a simulation should look like. 
\begin{figure}[ht]
\begin{center}
\definecolor{cdcdcdc}{RGB}{220,220,220}
\definecolor{c888888}{RGB}{136,136,136}

\begin{tikzpicture}[y=0.80pt, x=0.8pt,yscale=-1, inner sep=0pt, outer sep=0pt]
  \path[fill=cdcdcdc,rounded corners=0.0000cm] (76.8060,925.5563) rectangle
    (88.3338,937.0841);
  \path[cm={{0.0,-1.0,1.0,0.0,(0.0,0.0)}},fill=cdcdcdc,rounded corners=0.0000cm]
    (-1040.8345,65.2782) rectangle (-937.0841,99.8617);
  \path[fill=cdcdcdc,rounded corners=0.0000cm] (-38.4722,1017.7787) rectangle
    (88.3339,1052.3622);
  \path[fill=c888888,rounded corners=0.0000cm] (-38.4722,1029.3065) rectangle
    (76.8060,1040.8343);
  \path[cm={{0.0,-1.0,1.0,0.0,(0.0,0.0)}},fill=c888888,rounded corners=0.0000cm]
    (-1040.8345,76.8060) rectangle (-937.0841,88.3338);
  \path[fill=cdcdcdc,rounded corners=0.0000cm] (-50.0000,1029.3065) rectangle
    (-38.4722,1040.8343);
  \path[draw=black,fill=black,line join=miter,line cap=butt,line width=0.160pt]
    (-30.9043,1031.1180) .. controls (-32.5996,1031.2815) and (-34.6209,1030.4043)
    .. (-35.8793,1032.4593) .. controls (-36.7214,1033.8833) and
    (-36.6579,1036.1831) .. (-35.9497,1037.6568) .. controls (-34.5528,1039.7144)
    and (-32.2998,1038.8339) .. (-30.4734,1038.9900) .. controls
    (-28.6290,1039.7808) and (-26.4661,1038.3713) .. (-26.3916,1035.6051) ..
    controls (-25.6623,1033.3515) and (-27.3921,1031.6568) .. (-28.9447,1032.3035)
    .. controls (-29.6452,1032.0494) and (-30.1909,1031.3277) ..
    (-30.9043,1031.1180) -- cycle;
  \path[draw=black,line join=miter,line cap=butt,line width=0.184pt]
    (-26.3916,1035.6051) .. controls (-20.1347,1048.4574) and (-45.9018,1039.0194)
    .. (-38.4722,1032.2798) .. controls (-27.5683,1022.3885) and
    (-7.9477,1026.0153) .. (7.1390,1025.5205) .. controls (22.8242,1025.6345) and
    (30.8530,1043.9199) .. (46.4222,1046.1576) .. controls (54.4934,1045.0199) and
    (67.2427,1038.6536) .. (74.8159,1034.6730) .. controls (81.1923,1029.5636) and
    (74.7898,1023.8999) .. (71.3265,1019.6223) .. controls (66.8782,1012.0347) and
    (71.3840,1003.0655) .. (73.6493,995.2122) .. controls (77.0974,983.2586) and
    (89.2618,974.8344) .. (90.4456,961.8822) .. controls (90.2830,952.8454) and
    (82.8369,944.5338) .. (82.3090,929.1187);
  \path[draw=black,line join=miter,line cap=round,<->=,miter limit=4.00,line
    width=0.222pt] (231.5278,1052.3622) -- (231.5278,1017.7792);
  \path[draw=black,line join=miter,line cap=butt,miter limit=4.00,line
    width=0.222pt] (-78.4722,1017.7792) -- (272.7351,1017.7792);
  \path[fill=black] (-80.364609,1038.3201) node[above right] (text3854) {$3m$};
  \path[fill=black] (79.75,923.86218) node[above right] (text3183) {$\vdots$};
  \path[draw=black,line join=miter,line cap=butt,line width=0.800pt]
    (53.2782,981.3622) -- (109.8617,981.3622);
  \path[fill=black] (115,984.86218) node[above right] (text3226) {$h$};
  \path[draw=black,line join=miter,line cap=butt,line width=0.800pt]
    (-16.5000,1007.5705) -- (-16.5000,1064.1539);
  \path[fill=black] (-12.5,1067.1539) node[above right] (text3226-8) {$\ell$};
  \path[draw=black,line join=miter,line cap=rect,->=,miter limit=4.00,line
    width=0.400pt] (-109.2710,1033.8087) -- (176.4565,1033.8087);
  \path[draw=black,line join=miter,line cap=rect,->=,miter limit=4.00,line
    width=0.400pt] (-30.9043,1087.1180) -- (-30.9043,899.3904);
  \path[fill=black,line join=miter,line cap=butt,line width=0.800pt]
    (-43.133514,893.6167) node[above right] (text4156) {$y$};
  \path[fill=black,line join=miter,line cap=butt,line width=0.800pt]
    (186.32265,1036.0988) node[above right] (text4160) {$x$};
  \path[draw=black,line join=miter,line cap=round,<->=,miter limit=4.00,even odd
    rule,line width=0.222pt] (231.5280,1017.7792) -- (231.5280,923.3152);
  \path[fill=black,line join=miter,line cap=butt,line width=0.800pt]
    (241.20605,1043.8977) node[above right] (text4166) {Horizontal part};
  \path[fill=black,line join=miter,line cap=butt,line width=0.800pt]
    (243.71216,961.85254) node[above right] (text4170) {Vertical part};
  \path[draw=black,line join=miter,line cap=round,<->=,miter limit=4.00,even odd
    rule,line width=0.400pt] (-58.4722,1017.7792) -- (-58.4722,1052.3622);

\end{tikzpicture}
\end{center}
\caption{An example (claimed) simulation of some $\calT_N$, $N\in \Z^+$, by $\calU_{\calT_N}= (U,\sigma_{\calT_N},1)$. The seed assembly of $\calU_{\calT_N}$ is shown in black on the left. The simulator $\calU_{\calT_N}$ is free to place tiles anywhere in the {\em simulation zone}, defined to be the union of the dark gray (tile-representing supertiles) and light gray (fuzz) regions, and an example valid path of tiles (that has a \hs prefix) is shown in those regions as a thin black curve. The scale factor is $m$ and $\calU_{\calT_N}$ places $N+2$ horizontal tiles,  hence the width of the dark gray region is $m(N+2)$. Our main result is proven by showing that any simulator that claims to produce a valid terminal assembly must also produce one with an {\em incorrect shape}---one that either (a)~places tiles outside of the simulation zone (e.g.\ because some path is infinitely pumpable to the right or can be modified to grow upwards at an incorrect location that is  outside of the vertical part), or (b) is finite (e.g.\ does not grow infinitely upwards). The seed supertile contains the seed (in black) and the origin~$(0,0)$. The lines $\ell$ (at x-coordinate $|U|(3m+1)+m+1$) and $h$ (at y-coordinate $10m$) are used in many of the proofs.}
 \label{fig:fuzz}
\end{figure}
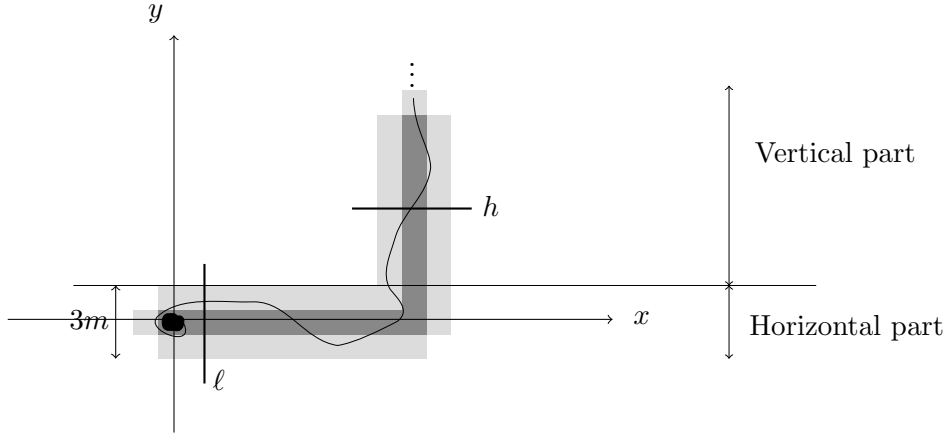

The proof is then broken into two stages. 
First, in Section~\oldref{sec:onepath}, we consider any path  that can grow in the claimed simulator long enough so it places at least one tile on a horizontal line at some height $h$. An example such path is  shown in Figure~\ref{fig:fuzz}. 
For any such path, we let $P$ denote its shortest prefix that contains exactly one tile at height $h=10m$ and where that tile is $P$'s last tile. Any assemblable path of this form is called ``\hs'' (see Definition~\ref{def:hsuccpaths}). 
We show that any  \emph{$h$-successful} path $P$ can be modified in two different ways:
 (1) $P$ is modified so that it grows at an invalid position, either by
  (1.1) \emph{infinite pumping}, by which we mean $P$ can be modified to give another path $P'$ which grows to form an assembly that is infinitely long horizontally to the right, and hence is not a simulation of the ``flipped-L'' shaped $\mathcal{T}_N$, or
  modifying it so that (1.2) the vertical arm grows displaced to the left or to the right (and hence is grown in the wrong position) or (1.3)  grows something in the wrong position by finite pumping (repeating a path segment that gets blocked).
 (2) $P$ is blocked, by which we mean {\em another} assembly (a path $R$) can be grown that blocks~$P$ (forcing~$P$ to be of finite length) 
 and thus the simulator can not rely on the growth of path~$P$ to obtain a valid simulation. 
(1) and (2) together show that no {\em single path} can carry out a valid simulation, which is stated formally in Theorem~\oldref{thm:onepath}. This can be regarded as our first technical tool for handling long paths.

However this leaves open the possibility that many paths could simultaneously and nondeterministically grow and interact in a way that carries out a valid simulation. 
In particular, in~(2) above, if we first block a path $P^1$ using another path $R^1$, and then  attempt  to block another path $P^2$ using a path~$R^2$, then $R^2$ may itself get blocked by $R^1$,
hence $P^2$ could possibly ``escape'' to become $h$-successful and carry out the simulation. 
In the second part of the proof, in Section~\oldref{sec:manypaths}, we consider the ways that \hs paths $P$, and the paths $R$ we use to block them, can interact. Based on this, we provide an explicit growth order for paths, such that if we apply (2) to each such path in turn, we  guarantee that the claimed simulator  fails.
This can be regarded as our second technical tool:  a method to control the interactions of multiple long paths.

The proof of our second main result Theorem~\ref{thm:no_finite_TM} is very short and uses the following intuition.	
We note that, for the sake of contradiction, if for each  time-bounded Turing machine $M$ there is a noncooperative tile set~$V$ that simulates   $M$ using arbitrarily large 2D space, bounded by a rectangle, then $V$ could be easily modified to build a family of  assemblies that are of the same (scaled) shape as  systems~$\calT_{N}$, 
thus contradicting our result that there is no such tile set.  Hence we give a reduction from a Turing machine prediction problem to the problem of simulating the class of all $\calT_{N}$ systems, a new technique for proving negative results about computation in tile assembly.
   
% !TEX root = t1notiu.tex
\section{Pumping or blocking any sufficiently wide and tall path}
\label{sec:onepath}

Suppose, for the sake of contradiction, that there is a tile set $U$, such that for $N=10|U|$, there is a scale factor $m \in \mathbb{Z}^{+}$ and a seed $\sigma$ for which ${\mathcal U} = (U,\sigma,1)$ simulates ${\mathcal T}_N$ at scale factor~$m$.  
Then it is the case that there is a path $P$ such that the simulator produces the assembly $\asm{P}$ and  $\asm{P}$ grows up to meet a line at height $h=10m$ above the horizontal arm of the assembly (see Figure~\ref{fig:fuzz}).
In this section, we prove that we can use the existence of such a ``\hs'' path $P$ to force the simulator to produce another assembly $\alpha\in\prodasm{\calU}$ that either (I) illegally places tiles outside of the simulation zone or else (II) is finite and blocks $P$ from growing (i.e.~$\alpha \cup \asm{P} $ is not producible from the assembly $\alpha \in \prodasm{U}$). 
Case (I) contradicts that $\calU$ simulates $\calT$ and we are done with the proof in that case. Case (II) gives a way to block the arbitrary single path $P$, but may not prevent other paths from growing and this is handled later  in Section~\ref{sec:manypaths}.

Since  showing impossibility of simulating a single  tile assembly system ${\mathcal T}_N$ is sufficient to prove our result, we will set $N=10|U|$ in the rest of the paper, and call $\sigma$ the seed and $m$ the scale factor at which our claimed simulator ${\mathcal U}=(U,\sigma,1)$ simulates~${\mathcal T}_N = {\mathcal T}_{10|U|}$.

\subsection{Glue visibility:  definitions and basic results about $V_{P}^{+}$ and $V_{P}^{-}$}\label{secvisibility}
\label{sec:vis}
We begin this section with a definition that will be used in  every proof in the rest of the paper, and is illustrated in Figure~\ref{fig:visible}. This notion of visibility is {\em from the south}. 

\begin{definition}[Visible glue (from the south)]
Let $P$ be a path and let $i,j\in\{0,1,\ldots,|P|-2\}$. Let  $\glue{i}{i+1} = (g, (x,y), d)$ where $g$ is the interacting glue type between $P_i$ and $P_{i+1}$, $(x,y) = \midpoint{P_i}{P_{i+1}}$ and $d = \vect{P_iP_{i+1}}$\footnote{Note that since paths are defined to be simple, for any $(x,y)\in\Z^2$, there is at most one index $i$ on $P$ such that $(x,y)$ is the position of $\glueP{i}{i+1}$.}. We say that $\glue{i}{i+1}$ is \emph{visible relative to $P$} if (a) there is a vertical  ray $r$ in $\mathbb{R}^2$ that starts from
$\midpoint{P_i}{P_{i+1}}\in \mathbb{R}^2$
and goes infinitely to the south  and (b) for all $j\neq i$,  $r$~does not contain  $\midpoint{P_j}{P_{j+1}}\in \mathbb{R}^2$.
\end{definition}

We write $V_P$ for the set of (glue type, path index) pairs of $P$ that are visible relative to some path $P$. We define $V^+_P$ (\resp, $V^-_P$) to be the set of those (glue type, path index) pairs of $V$  that are on {\em east} (\resp, west) of output sides of tiles of $P$.  Since, on a path $P$, each tile has exactly one output side, it follows that $(V^+_P,V^-_P)$ is a partition of $V_P$.

For a path $P$ and a vertical line $l = \{ (x,y) \mid y \in R \}$ at some position $x \in \{ \frac{x'}{2} \mid x' \in \Z \}$  such that  $l \cap \ep \neq \emptyset$,  we write ``$P$'s visible glue on $l$''  to mean the unique glue $\gglue{P}{i}{i+1}$ such that $\midpoint{P_i}{ P_{i+1}} \in\R^2$ is on $l$ and for all $j$, such that $0\leq j \leq |P|-2$, $j\neq i$ it is not the case that  $\midpoint{P_j}{P_{j+1}}$ has a smaller $y$-coordinate on $l$ than $\midpoint{P_i}{P_{i+1}}$. Moreover ``the position of the visible glue of $P$ on $l$'' is the point $\midpoint{P_i}{P_{i+1}}$. For a  glue $\gglue{P}{i}{i+1}$ that is visible relative to $P$ we write ``the position of $\gglue{P}{i}{i+1}$'' to mean the point $\midpoint{P_i}{P_{i+1}} \in\R^2$.
For a glue $\gglue{P}{i}{i+1}$ that is visible relative to $P$ we write ``the visibility ray of $\gglue{P}{i}{i+1}$'' to mean the infinite ray starting  at the point $\midpoint{P_i}{P_{i+1}} \in\R^2$ and going vertically to the south.

It is important to note that ``visible glue'' is defined  relative to a particular path. Hence, we often say that a glue of $P$ is ``visible  \emph{relative to $P$}'', although when $P$ is clear from the context, we may simply say that a glue is ``visible''.
 Figure~\ref{fig:visible} gives examples of glues on a path that are and are not visible. Intuitively, note that although a visibility ray starts at the ``position'' of a visible glue on the path $P$, the ray is permitted to ``touch'' other free glues (Section~\ref{sec:defs-paths}) on tiles of~$P$ (recall free glues are by definition not on the path since they are not input/output glues along the path, hence they are not visible  nor can they prevent some other glue from being visible). See Figure~\ref{fig:visible} for examples.

\begin{figure}
\begin{center}
\includegraphics[angle=270,scale=0.65]{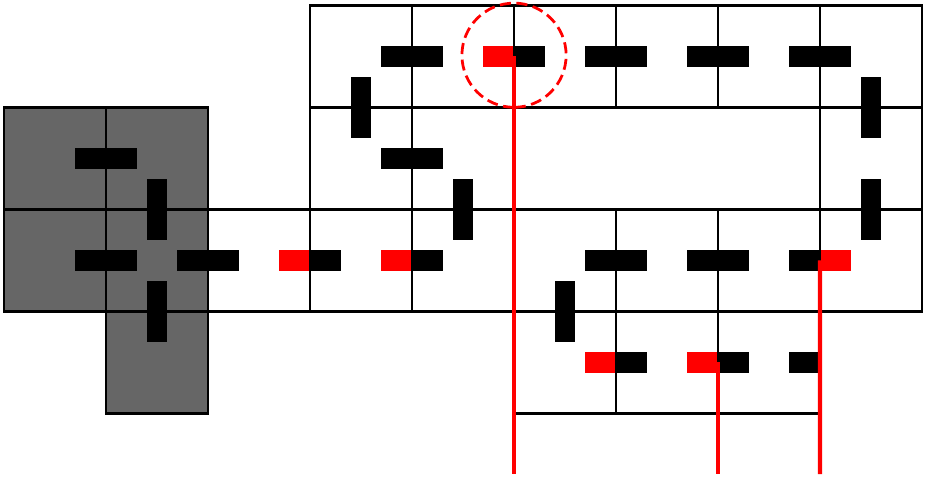}
\end{center}
\caption{A short path with six glues that are visible (from the south). Growth of this path begins from the seed which is shown in grey.  Glues here are coloured not by their type, but by their visibility (and direction).  Five red-black glues are in $V^{+}_{P}$ and one black-red glue (to the right) is in $V^{-}_{P}$ giving a total of six glues in $V_{P}$. None of the other glues (black-black, or black-) are  visible. Three red rays to the south testify to the visibility of three of the glues, and are called the ``visibility rays'' of these glues.}
\label{fig:visible}
\end{figure}

In the following lemma we show that for a path $P$ that has both  $V^+_P$ glues and  $V^-_P$ glues, the  $V^+_P$ glues are positioned to the right of its $V^-_P$ glues. The intuition behind the proof is as follows: suppose otherwise, then draw a finite length curve $c$ in $\R^2$ that runs from a $V_P^+$ glue along the positions of $P$ to a $V_P^-$ glue, then includes segments of the visibility rays  (to the south) of these two glues and finally includes a horizontal line that lies far below $P$ and runs between those two visibility rays. It turns out that $c$ is simple and closed and thus cuts the plane into an unbounded component and a bounded component $\mathcal{C}$. It turns out that $P$ must both go inside $\mathcal{C}$ for a time and then  leave $\mathcal{C}$, but since $c$ was defined using curves that $P$ never crosses we get a contradiction.  This rough intuition is more rigorously formalised in the proof. 
\begin{lemma}[$V^+_P$ glues are to the right of $V^-_P$ glues]
\label{lem:leftright}
Let $P \in  \prodpathsU$ be an assemblable path of any tile assembly system~$\calU=(U, \sigma, 1)$, and let $i$ and~$j$ be two indices such that $\glue{i}{i+1} \in V^+_P$, $\glue{j}{j+1} \in V^-_P$. If $P$ has at least one tile $P_k$ 
where $k > \max(i,j)$ and $\pos{P_k}$ is strictly to the right of, or strictly above, all tiles of $P_{\min(i,j),\min(i,j)+1,\ldots,\max(i,j)}$, then $x_{P_i}> x_{P_j}$. In other words, the glues of $V^+_P$ are all to the right of the glues of $V^-_P$. \end{lemma}
\begin{proof}
  We have argued above that $(V_P^+, V_P^-)$ is a partition of $V_P$.
  Hence since $\glue{i}{i+1}\in V^+_P$, $\glue{j}{j+1}\in V^-_P$ it is the case that $i \neq j$ which in turns implies $x_{P_i}\neq x_{P_j}$.   Assume, for the sake of contradiction, that there are two integers $i$ and $j$ that satisfy the lemma hypotheses (in particular that $\glue{i}{i+1}\in V^+_P$, $\glue{j}{j+1}\in V^-_P$) but where $x_{P_i}<x_{P_j}$.

Since $\glue{i}{i+1}\in V^+_P$, there is a vertical ray $l_i$ that starts from the point $\midpoint{P_i}{P_{i+1}} \in \R^2$, goes infinitely to the south, and  does not contain any point of the canonical embedding $\frak{E}_P$ of $P$ in $\R^2$ besides $\midpoint{P_i}{P_{i+1}}$. 
Likewise, let $l_j$ be the vertical ray to the south starting from the point $\midpoint{P_j}{P_{j+1}}$, and observe that since $\glue{j}{j+1}\in V^-_P$, then $l_j$ does not contain any point of $\frak{E}_P$ besides $\midpoint{P_j}{P_{j+1}}$.  
We will use $l_i$ and $l_j$ to define the three line segments $\rev{s_i}$, $s_j$ and $s_{i,j}$. First let $y_0$ be a y-coordinate below all of $P$, for instance
$$y_0 = \min\{y_{P_0},y_{P_{1}}, \ldots, y_{P_{|P-1|}}\} - 10$$
We then define:
\begin{equation}
  \label{eq:si}
  s_i = [ \midpoint{P_i}{P_{i+1}} , (  X(\midpoint{P_i}{P_{i+1}} ) , y_0)] \subsetneq l_i
\end{equation}
where $X(x,y) = x$ and we note that $s_i \subsetneq l_i  $ since $X(\midpoint{P_i}{P_{i+1}} ) , y_0)$ is directly to the south of $\midpoint{P_i}{P_{i+1}}$.  Let $\rev{s_i}$ be the ``reverse direction'' of the line segment $s_i$, more precisely:
\begin{equation}
  \label{eq:revsi}
  \rev{s_i} = [  (  X(\midpoint{P_i}{P_{i+1}} ) , y_0), \midpoint{P_i}{P_{i+1}}]
\end{equation}
Also, define the line segment 
\begin{equation}
  \label{eq:sj}
  s_j = [ \midpoint{P_j}{P_{j+1}} , (  X(\midpoint{P_j}{P_{j+1}} ) , y_0)] \subsetneq l_j
\end{equation}
And the line segment 
\begin{equation}
  \label{eq:sji}
  s_{j , i} = [  (X(\midpoint{P_j}{P_{j+1}} ) , y_0), (  X(\midpoint{P_i}{P_{i+1}} ), y_0)]
\end{equation}

There are two (almost identical) cases, (a) and (b).
\paragraph{Claim (a): $i<j$.} 
We let $c$ be the {\em concatenation} (see Definition~\ref{def:concat})  of the following six curves: 
\begin{equation*}
  \begin{split}
    &  \rev{s_i} \\
    & [\midpoint{P_i}{P_{i+1}}, \pos{ P_{i+1}} ] \\
    &  \frak{E}_{P_{i+1, i+2, \ldots, j} } \\
    &  [\pos{ P_j}, \midpoint{P_{j}}{P_{j+1}}] \\
    & s_j\\
    &s_{j,i}
  \end{split}
\end{equation*}
By Observation~\ref{obs:ClosedSimpleCurve},  $c$ is a finite closed simple curve\footnote{\label{ft}To see this one needs to check that the components of $c$ satisfy the hypotheses of Observation~\ref{obs:ClosedSimpleCurve}. Less formally but more intuitively, it can be seen that $c$ is  of finite length because its components are, also $c$ is closed as each of the components are  curves and their endpoints are pairwise equal in such a way to satisfy closure, and finally $c$ is simple since the components are simple and their only intersection is at their endpoints  in the order they are given.}  and thus defines a bounded  connected component $\mathcal C$.  (See Figure~\ref{fig:rightof} for an example.)

We claim that $\pos{P_{j+1}}$  is inside $\mathcal{C}$. First, note that  $\pos{P_{j+1}}$ is to the left of $\pos{P_{j}}$ (because $\glueP{j}{j+1}\in V_P^-$), and therefore, $\pos{P_{j+1}}$ is to the left of $l_j$. But since $x_{P_i} < x_{P_j}$ by assumption,    $X(\midpoint{P_i,P_{i+1}}) \neq   X(\midpoint{P_j,P_{j+1}}  $ by visibility,  and $P_{j+1}$ is unit distance to the left of $P_j$, 
then $\pos{P_{j+1}}$ is in fact between $l_i$ and $l_j$ (i.e.\ to the right of $l_i$ and to the left of $l_j$). Secondly, $\pos{P_{j+1}}$  is above the horizontal line $s_{j,i}$. 
Consider the ray $r$ at x-coordinate $X( \pos {P_{j+1}}) +  0.25$ that comes from the south
and stops at position $ p = \left( X( \pos{ P_{j+1} }) + 0.25,  Y(\pos{ P_{j+1} }) \right) $. Observe that $r$ crosses $c$ at the segment $s_{j,i}$ exactly once, and crosses $c$ nowhere else, and that due to the visibility of $\glueP{j}{j+1}$ we get that~$r$  does not intersect $\frak{E}_{P_{i+1, i+2, \ldots, j} }$, and that by its definition  $r$ is positioned away from the other four components of $c$. 
 Furthermore, since $c$ does not cross the short line segment $[\pos{ P_{j+1}},  p]$ then starting at the point $p $, one can walk (westwards) along the segment $[\pos{ P_{j+1}},  p]$ to the  point $\pos{P_{j+1}}$, without crossing $c$. 
 Hence $\pos{P_{j+1}}$  is inside $\mathcal{C}$ as claimed. 

Since, from the lemma statement, $P$ has at least one tile $P_k$ after $P_j$ (i.e.\ $k>j$)  positioned to the right of, or above, all tiles of $P_{i,i+1,\ldots,j}$, then $P$ has tiles positioned outside of $\mathcal{C}$ after~$P_{j}$.
But since $P$ is a path, it does not cross itself. Therefore, $\frak E_{P_{j+1,j+2,\ldots,|P|-1}}$
must leave $\mathcal{C}$, thus crossing $c$ by crossing at least one of $ l_i$ or $l_j$ and contradicting that both~$\glue{i}{i+1}$ and $\glue{j}{j+1}$ are visible. Thus $x_{P_i} > x_{P_j}$.

\begin{figure}
  \begin{center}
    \includegraphics[scale=1.1]{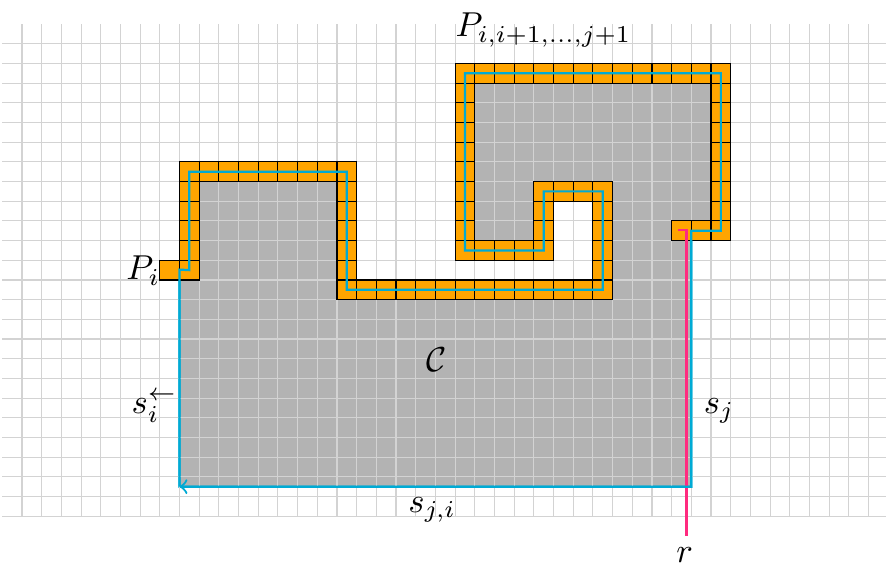}
  \end{center}
  \caption{An example illustrating the proof of Lemma~\ref{lem:leftright}, Case (a):  $i<j$. In this figure, $\glueP{i}{i+1}\in V_P^+$, $\glueP{j}{j+1}\in V_P^-$, and $\pos{P_i}$ is to the left of $\pos{P_j}$. The canonical embedding of the path  $P_{i+1, i+2,\ldots, j}$ in $\R^2$ and  five line segments are used to define a closed simple curve $c$ (light blue) and thus a  bounded connected component $\mathcal{C}$ of $\R^2$ (two of the line segments are unlabelled in this figure). Then, the  ray $r$ is used to verify that $\pos{P_{j+1}}$ is inside $\mathcal{C}$. Thus $P_{j+1,j+2,\ldots}$ ``grows inside'' $\mathcal{C}$, but must leave $\mathcal{C}$ at some point leading to a contradiction as $P$ can not cross any of the components that define $c$.}
  \label{fig:rightof}
\end{figure}

\paragraph{Case (b): $j<i$.}
We define the curve $c$ as the  concatenation (see Definition~\ref{def:concat})  of the following six curves:
\begin{equation*}
  \begin{split}
    & s_j^\leftarrow \\
    & [\midpoint{P_j}{P_{j+1}}, \pos{P_{j+1}}] \\
    & \frak{E}_{P_{j+1, j+2, \ldots, i}} \\
    & [\pos{P_i}, \midpoint{P_{i}}{P_{i+1}}]\\
    & s_i\\
    & s_{j,i}^\leftarrow
  \end{split}
\end{equation*}

By Observation~\ref{obs:ClosedSimpleCurve}, 
$c$ is a finite closed simple curve\footnoteref{ft}, and thus defines a bounded connected component~$\mathcal{C}$. (See an example in Figure~\ref{fig:rightof2}.)

By a similar\footnote{Specifically, consider the vertical ray $r$ with x-coordinate $X(\pos{P_{i+1})}-0.25$  that comes from the south
and  stops at position $ p = \left( X( \pos{ P_{i+1} }) - 0.25,  Y(\pos{ P_{i+1} }) \right) $. Observe that $r$ crosses $c$ exactly once at the segment $s_{j,i}^\leftarrow$,  that due to the visibility of $\glueP{i}{i+1}$ we get that $r$  does not intersect $\frak{E}_{P_{j+1, j+2, \ldots, i} }$, and that by its definition~$r$ is positioned away from the other four components of $c$. Furthermore, since $c$ does not cross the short line segment $[p, \pos{ P_{i+1}}]$ then starting at the point $p $ one can walk (eastwards) along the segment $[p, \pos{ P_{i+1}}]$ to the  point $\pos{P_{i+1}}$, without crossing $c$. Thus $\pos{P_{i+1}}$ is inside $\mathcal{C}$. See Figure~\ref{fig:rightof2} for an example of the argument.} argument as Case (a) (where $i<j$), $\pos{P_{i+1}}$ is inside $\mathcal{C}$.
Since $P$ has at least one tile $P_k$ after $P_i$ (i.e.\ $k > i$) to the right of or above $P_{j,j+1,\ldots,i}$, then $P_{i+1,i+2,\ldots,|P|-1}$ has tiles positioned outside $\mathcal{C}$.
But since $P$ is a path, it does not cross itself, nor does it place a tile below $s_{j,i}^\leftarrow$. Therefore, $\frak E_{P_{i+1,i+2,\ldots,|P|-1}}$ must leave $\mathcal{C}$, thus crossing $c$ by crossing at least one of $ l_i$ or $l_j$ and contradicting that both contradicting that both~$\glue{i}{i+1}$ and $\glue{j}{j+1}$ are visible. Thus $x_{P_i} > x_{P_j}$. 
\begin{figure}
  \begin{center}
    \includegraphics[scale=1.1]{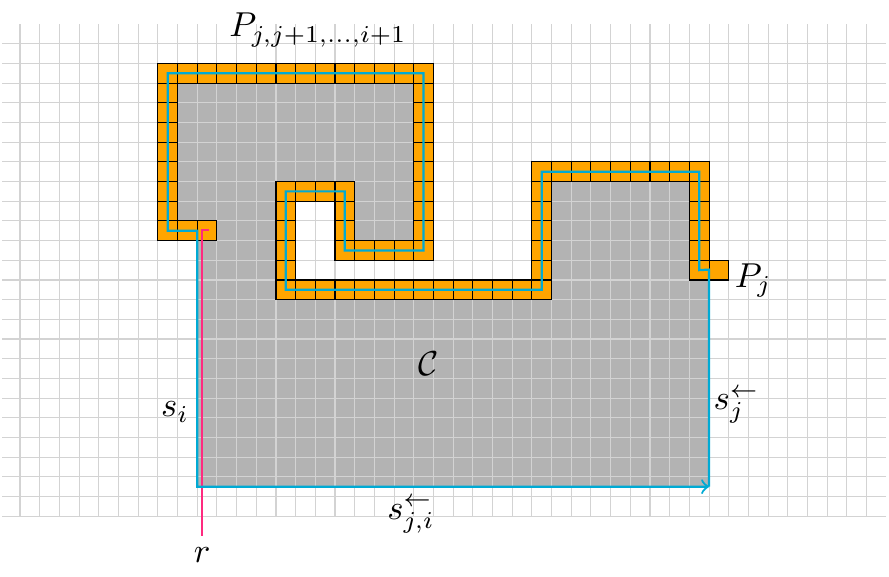}
  \end{center}
  \caption{An example illustrating the proof of Lemma~\ref{lem:leftright}, Case (b): $j<i$. In this figure, $\glueP{i}{i+1}\in V_P^+$, $\glueP{j}{j+1}\in V_P^-$, and $\pos{P_i}$ is to the left of $\pos{P_j}$ (i.e.\ $P_{j, j+1,\ldots}$ begins at the rightmost orange tile). The canonical embedding of the path $P_{j+1, j+2,\ldots, i}$ in $\R^2$  and five line segments are used to define a closed simple curve $c$ (light blue) and thus a  bounded connected component $\mathcal{C}$ of $\R^2$ (two of the line segments are unlabelled in this figure).  Then, the ray $r$ is used to verify that $\pos{P_{i+1}}$ is inside $\mathcal{C}$. Thus $P_{i+1,i+2,\ldots}$ ``grows inside'' $\mathcal{C}$, but must leave $\mathcal{C}$ at some point which leads to a contradiction as $P$ can not cross any of the components that define $c$.}
  \label{fig:rightof2}
\end{figure}
\end{proof}

Lemma~\ref{lem:leftright} was our first statement proven using visibility, and the following lemma~(\ref{lem:columns}) similarly exploits the technique of using rays given by visibility, an embedded path, and other line segments to enclose a connected component of the plane,  enabling  us to reason about how visible glues are organised along a path.
Together Lemmas~\ref{lem:leftright} and~\ref{lem:columns} and Corollary~\ref{cor:columns} give properties of how tiles of $V_P^+$ and $V_P^-$ are arranged in the plane. This is formalized in greater detail in Lemma~\ref{lem:visibility-setup}.

\begin{lemma}[$V^+_P$ glue order preserves path order]
\label{lem:columns}
  Let $P \in \prodpathsU$  be a path producible by any tile assembly system $\calU=(U,\sigma,1)$, and let $i,j$ be such that  $\glueP{i}{i+1} \in V_P^+$, $\glueP{j}{{j+1}} \in V_P^+$ and $x_{P_i}<x_{P_j}$. If $P$ has at least one tile $P_k$ after $P_i$ and $P_j$ (i.e.\ $i < k, j < k$), where $\pos{P_k}$ is to the right of, or above, all tiles of $P_{\min(i,j),\min(i,j)+1,\ldots,\max(i,j)}$, then $i<j$.
\end{lemma}
\begin{proof}
  Define $s_i$ as in Equation~\eqref{eq:si}, $\rev{s_i}$ as in Equation~\eqref{eq:revsi}, $s_j$ as in Equation~\eqref{eq:sj} and $s_{j,i}$ as in Equation~\eqref{eq:sji}.

  First assume, for the sake of contradiction, that $j<i$.
  We let $c$ be the {\em concatenation} (see Definition~\ref{def:concat})  of the following six curves: \begin{equation*}
    \begin{split}
      &  \rev{s_i} \\
      & [\midpoint{P_i}{P_{i+1}}, \pos{P_{i}}] \\
      &  \frak{E}_{P_{j+1, j+2, \ldots, i} }^\leftarrow \\
      &  [\pos{P_{j+1}}, \midpoint{P_{j}}{P_{j+1}}] \\
      & s_j\\
      &s_{j,i}
    \end{split}
  \end{equation*}

  By Observation~\ref{obs:ClosedSimpleCurve},  $c$ is a simple closed curve, hence partitions $\R^2$ into two connected components, exactly one of which is bounded. Let $\mathcal{C}$ be that bounded connected component.
  By a similar argument\footnote{Here, $r$ comes from the south with  x-coordinate $X(\pos{P_{i+1}})-0.25$, crosses $s_{j,i}$ (thus entering $\mathcal{C}$), ends at y-coordinate $Y(\pos{P_{i+1}})$, and then we walk from that end point to $\pos{P_{i+1}}$ staying inside $\mathcal{C}$.}   using ray $r$ as in the proof of Lemma~\ref{lem:leftright},  
  $P_{i+1,i+2,\ldots,|P|-1}$ starts inside~$\mathcal{C}$.
   However, since $P_{i+1,i+2,\ldots,|P|-1}$ has at least one tile to the right of, or above $P_{j,j+1,\ldots,i}$, then $P_{i+1,i+2,\ldots,|P|-1}$ cannot be entirely inside $\mathcal{C}$.
  Therefore, $P$ needs to cross the border of $\mathcal{C}$, contradicting either the fact that $P$ is simple, that paths do not place tiles below their lowest tile, or that $\glueP{i}{i+1}$ or $\glueP{j}{j+1}$ are visible. \end{proof}

By flipping the use of ``+'' and ``-'', and ``left'' and ``right'', in the statement of the previous lemma, we immediately get the following corollary:
\begin{corollary}[$V^-_P$ glue order preserves path order]
  \label{cor:columns}
  Let $P \in \prodpathsU$  be a path producible by any tile assembly system $\calU=(U,\sigma,1)$,   and let $i,j$ be such that  $\glueP{i}{i+1} \in V_P^-$,  $\glueP{j}{{j+1}} \in V_P^-$ and $x_{P_i}>x_{P_j}$. If $P$ has at least one tile $P_k$ after $P_i$ and $P_j$ (i.e.\ $i<k$, $j<k$), where $\pos{P_k}$  is to the right of, or above all tiles of $P_{\min(i,j),\min(i,j)+1,\ldots,\max(i,j)}$, then $i<j$.
\end{corollary}

In the proof of Lemma~\ref{lem:onepath}, we will attempt to pump a path segment. We will make use of the following  lemma stating that the $V_P^+$ glues on the prefix remain visible even if that prefix is pumped. In other words {\em visibility survives pumping}.
\newcommand{\ola}{\ensuremath{\overline{A}}}

\begin{lemma}[Visibility survives pumping for $V_{P}^{+}$]
\label{lem:semiprecious}
Let $P \in \prodpathsU$  be a path producible by any tile assembly system $\calU=(U,\sigma,1)$, and $i$, $j$ be two integers such that 
$i < j$, $\glueP{i}{i+1} \in V_{P}^{+}$, $\glueP{j}{j+1} \in V_{P}^{+}$,
and $\type{\glueP{i}{i+1}} = \type{\glueP{j}{j+1}}$.
Let~$\olq$ be the pumping of $P$ between $i$ and $j$ (as defined in Definition~\ref{def:pumpingPbetweeniandj}), and let $Q$ be the maximal  prefix of $\olq$ that is an assemblable path.

Then  $V_{P_{0,1,\ldots,j}}^{+} \subseteq V_Q^{+}$. Intuitively, this means that all ``+'' glues visible relative to  $P_{0,1,\ldots,j}$ are also visible relative to  $Q$ (note that $Q$ contains $P_{0,1,\ldots,j}$ as a prefix).
\end{lemma}
\begin{proof}
  {\bf Setup.}
  Let $l_i$ be the visibility ray of $\glueP{i}{i+1}$ and $l_j$ be the visibility ray of $\glueP{j}{j+1}$.
  Now let $P_{\text{left}}$ be a leftmost tile of $\sigma\cup \asm{P_{0,1,\ldots,j}}$, and let $l_{\mathrm{left}}$ be the horizontal ray to the west starting from $\pos{P_{\mathrm{left}}}$.
  Moreover, let $P_{\text{bottom}}$ be a lowest tile of $P_{0,1,\ldots,j}$ (i.e.\ with $y$-coordinate 
   $Y(\pos{P_{\text{bottom}}}) = \min\{y_{P_0},y_{P_{1}}, \ldots, y_{P_{|P-1|}}\} $).

  We now define two helper points of $\R^2$, far enough from $P_{0,1,\ldots,j}$:
  \begin{eqnarray*}
    A &=& P_{\mathrm{left}} + \left(\begin{array}{c} -10\\0\end{array}\right)\\
    B &=& \left(\begin{array}{c} X(\midpoint{P_j}{P_{j+1}})\\Y(\pos{P_{\text{bottom}}}) - 10\end{array}\right)
  \end{eqnarray*}
where $X(x,y) = x$, $Y(x,y) = y$.

  Intuitively, $A$ has the same y-coordinate as, and is 10 units to the left of $P_{\mathrm{left}}$, and $B$ has the same x-coordinate as the glue between $P_j$ and $P_{j+1}$, and is 10 units below the lowest point of $P_{0,1,\ldots,j}$.\footnote{The choice of 10 units is arbitrary here, we simply need to define a curve with  segments that are strictly to the left of $\sigma \cup \asm{P_{0,1,\ldots,j}}$ and below $P_{0,1,\ldots,j}$.}

  Define  $\mathcal Q_A = \{(x,y) \in\R^2| x \leq x_A\wedge y\leq y_A \textrm{ where } A=(x_A,y_A) \}$ (the quarter-plane below and to the left of $A$) and $\mathcal Q_B = \{(x,y)\in\R^2 | x \leq x_B\wedge y\leq y_B  \textrm{ where }  B=(x_B,y_B) \}$ (the quarter-plane below and to the left of $B$).
  Also, let $s_j = [ \pos{P_j}, \midpoint{P_j}{P_{j+1}} ]$ be the half-unit horizontal segment of $\R^2$ from $\pos{P_j}$ to $\midpoint{P_j}{P_{j+1}}$.

  Moreover, let $c$ be the curve that is the concatenation of  $[(x_A, y_B), A]$, $[A, \pos{  P_{\mathrm{left}} }] $, $\frak{E}_{P_{\mathrm{left}, \mathrm{left}+1,\ldots,j}}$, $s_j$, $[\midpoint{P_j}{P_{j+1}}, B]$,  $[B, (x_A, y_B)]$.
By Observation~\ref{obs:ClosedSimpleCurve},  $c$ is a simple closed curve, hence defines a  bounded connected component $\mathcal{C}$ of $\R^2$. 

  Now, let $C_j = \mathcal Q_A\cup\mathcal Q_B\cup \mathcal{C}$, which is also a single  
  connected component of $\R^2$, because $(x_A, y_B)$ is in all components of the union ($(x_A,y_B)\in\mathcal Q_a\cap\mathcal Q_B\cap \mathcal{C}$) and each component is connected.

  {\bf Proof argument.}
  We claim that for all $k< j$, $l_k\subset C_j$ since $l_k \subset \mathcal Q_B \cup \mathcal{C}$ which can be seen as follows: (i) $l_k^B = \{ (x,y) \mid (x,y) \in l_k \wedge y \leq B) \}$ is contained in $\mathcal Q_B$ since $l_k$ is to the left of $l_j$ (by the contrapositive of Lemma \ref{lem:columns}) and $l_k^B$ reaches infinitely far to the south as does  $\mathcal Q_B$; and (ii) by the definition of visibility the segment $ l_k \setminus l_k^B $ is contained in $\mathcal{C}_j$.

  Now suppose, for the sake of contradiction, that  $\frak{E}_{Q_{j+1,j+2,\ldots}} $, the canonical embedding of the path $Q_{j+1,j+2,\ldots}$ in $\R^2$, intersects $l_k$. Then $\frak{E}_{Q_{j+1,j+2,\ldots}}$ needs to enter $C_j$, because $\pos{Q_{j+1} }$ is outside $C_j$ (by the contrapositive of Lemma \ref{lem:columns}).
    This crossing can happen at only four different parts of the border of $C_j$:
  \begin{itemize}
  \item $l_{\mathrm{left}}$, but this is impossible since $\pos{P_{\mathrm{left}}}$ is a leftmost point of $\sigma\cup \asm{ P_{0,1,\ldots,j} }$ and $\vect{P_iP_j}$ (as a vector) has a strictly positive x-coordinate (by the contrapositive of Lemma \ref{lem:columns}).
  \item $P_{\mathrm{left}, \mathrm{left}+1, \ldots, j}$ or $s_j$, but this would contradict the fact that $Q$ is simple.
  \item $l_j$. We claim that this is impossible: assume, for the sake of contradiction, that it is not, and let $k_0>j$ be the smallest integer such that $Q_{j+1,j+2,\ldots,k_0}$ intersects $l_j$.
    Moreover, $k_0-(j-i) > j$ because $Q_{i,i+1,\ldots,j}$ does not intersect its own visibility rays (except at their endpoints). 
    Therefore, $Q_{j+1,j+2,\ldots,k_0-(j-i)}$ would also intersect $l_i$, hence $Q_{j+1,j+2,\ldots,k_0-(j-i)}$ also enters $C_j$, contradicting the assumption that $k_0$ is the smallest integer such that $Q_{j,j+1,\ldots,k_0}$ intersects $l_j$ (since in the previous bullet points we have shown that $Q_{j+1,j+2,\ldots,k_0-(j-i)}$ cannot cross other parts of the border of $C_j$).
  \end{itemize}

  Therefore, $\frak{E}_{Q_{j+1,j+2,\ldots}}$ does not enter $C_j$, which is a contradiction, which in turn contradicts our assumption that $\frak{E}_{Q_{j+1,j+2,\ldots}}$ intersects $l_k$. 
   Hence the canonical embedding of $Q$ does not cross any of the visibility rays defining $V_{P_{0,1,\ldots,j}}^{+}$,
     thus      $V_{P_{0,1,\ldots,j}}^{+} \subseteq V_Q^{+}$. 
\end{proof}

By flipping the use of ``$+$'' and ``$-$'' in the statement of the previous lemma, and ``left'' and ``right'' in the proof,  we get the following corollary:

\begin{corollary}[Visibility survives pumping for $V_{P}^{-}$]
\label{cor:semiprecious-minus}
Let $P \in \prodpathsU$  be a path producible by any tile assembly system $\calU=(U,\sigma,1)$, and $i$, $j$ be two integers such that
$i < j$, $\glueP{i}{i+1} \in V_{P}^{-}$ and $\glueP{j}{j+1} \in V_{P}^{-}$,
and $\type{\glueP{i}{i+1}} = \type{\glueP{j}{j+1}}$.
Let~$\olq$ be the pumping of $P$ between $i$ and $j$, and let $Q$ be the maximal prefix of $\olq$ that is an assemblable path.

Then  $V_{P_{0,1,\ldots,j}}^{-} \subseteq V_Q^{-}$. Intuitively, this means that all ``-'' glues visible relative to  $P_{0,1,\ldots,j}$ are also visible relative to  $Q$ (note that $Q$ contains $P_{0,1,\ldots,j}$ as a prefix).
\end{corollary}

\subsection{Blocking any \hs path}
\label{sec:blocking-h-successful}
Keeping in mind that the seed assembly supertile  of $\calU$ includes the origin $(0,0) \in \Z^2$, for the rest of the paper fix a horizontal line at height $h=10 m$  above the origin.
\begin{definition}[The set of \hs paths of $\calU$]\label{def:hsuccpaths}
The set of  {\hs paths} $\mathbb{P}_\calU$  of $\calU$ is defined as: 
  \begin{equation*}%\label{eq:hsuccpaths}
  \mathbb{P}_\calU = \{ P \mid P \in \prodpathsU \text{ and } P \text{ contains exactly one tile at height } h=10m \text{, its last tile}\}
    \end{equation*}
    where $\prodpathsU$ is the set of producible paths of $\calU$ (defined in Section~\ref{sec:defs-paths}).
\end{definition}    

Note that {\em any} claimed successful simulation by $\calU$ of $\calT_{n}$ (defined in Section~\ref{sec:T}) must exhibit at least one path that has a  \hs prefix $P$. When we write ``$P$ is a \hs path'' we mean $P\in\mathbb{P}_\calU$. The set of \hs paths $\mathbb{P}_\calU$ is finite because the tileset $U$ is finite and the area of the simulation zone below the horizontal line at height $h$ is finite. 

\newcommand{\nhs}{nowhere-\hs}
\newcommand{\Nhs}{Nowhere-\hs}
We define a  {\em ``\nhs path''} to be a path that has no tile at height $h$.  In other words, a \nhs path has no \hs prefixes.

\subsubsection{Visibility setup}
We begin with the following lemma, which has a straightforward proof and is used merely to define $i$, $j$ and $\ell$, which are used extensively in later proofs.
Recall that  $\calU$ is a tile assembly system with tile set $U$ simulating  $\calT_{10|U|}$  at scale factor $m$.  

\newcommand{\hsU}{\ensuremath{\mathbb{P}_\calU}}

\begin{lemma}
\label{lem:visibility-setup}
Let $P \in\hsU$ be a \hs path, and let $\ell$ be a vertical line in $\R^2$ with x-coordinate $|U|(3m+1)+m+1.5$. 
If the visible glue placed by $P$ on $\ell$ is a $V_P^+$ glue (\resp a $V_P^-$ glue), then there exist
$i< j \in \N$ that satisfy all of the following properties:
\begin{enumerate}
\item $\pos{P_{i}}$ is to the right (\resp to the left) of $\ell$ (i.e.\ $x_{P_{i}} > \ell$, \resp $x_{P_i} < \ell$)\label{a}
\item $\pos{P_{j}}$ is horizontal distance at least $3m$ from $\pos{P_i}$ (i.e. $|x_{P_i} - x_{P_j}| \geq 3m$)\label{b1}
\item $\pos{P_{j}}$ is to the right (\resp to the left) of $\pos{P_{i}}$\label{b2}
\item $\glueP{i}{i+1} \in V_P^+$ and $\glueP{j}{j+1} \in V_P^+$ (\resp $\glueP{i}{i+1}\in V_P^-$ and $\glueP{j}{j+1}\in V_P^-$)\label{c}
\item $\type{\glueP{i}{i+1}} = \type{\glueP{j}{j+1}}$\label{d}
\item $\pos{P_{i}}$ and $\pos{P_{j}}$ are within vertical distance $3m$ (i.e.\ $|y_{P_{j}} - y_{P_{i}} | \leq 3m$)\label{e}
\end{enumerate}
\end{lemma}
\begin{proof}
Firstly, the path $P$ is of width (horizontal extent) $\geq 8|U|m+2$ (as we are simulating $\calT_{10|U|}$ at scale factor $m$, see Figure~\ref{fig:fuzz}).
Secondly, since $P$ is \hs, $P$ crosses $\ell$, and since $\ell$ is positioned at distance $\geq |U|(3m+1)+1$ to the right\footnote{Note that the $m\times m$ seed supertile region contains the point $(0,0) \in \mathbb{Z}^2$.} of the seed~$\sigma$ of $\calU$, there are at least $|U|(3m+1)+1$ visible glues to the left of $\ell$.

Moreover:
\begin{itemize}
\item If the visible glue placed by $P$ on $\ell$ is in $V_P^+$, then since $P$ has at least one tile in the rightmost $2m$ positions of the simulation zone, $P$ has at least $|U|(3m+1)+1$ visible glues to the right of $\ell$.
By Lemma \ref{lem:leftright} all of these glues are in $V_P^+$.
\item Else, since $\ell$ is at horizontal distance at least $|U|(3m+1)+1.5$ from the rightmost tile of $\sigma$, $P$ has at least $|U|(3m+1)+1$ visible glues to the left of $\ell$
By Lemma \ref{lem:leftright} all of these glues are in $V_P^-$.
\end{itemize}

Thus, if we look at the first $|U|(3m+1)+1$ visible glues of $P$ that are immediately to the right (\resp to the left) of $\ell$, by the pigeonhole principle, at least one of their glue types appears at least $3m+1$ times. Since each x-coordinate has exactly one visible tile, we can find two $V_P^+$ (\resp $V_P^-$) glues, $\glueP{i}{i+1}$ and $\glueP{j}{j+1}$ for some $i<j$, with the same type, that are at least horizontal distance $3m$ away from each other, which shows Conclusions~\ref{a},~\ref{b1},~\ref{c} and~\ref{d} of this lemma. Taking the contrapositive of Lemma~\ref{lem:columns} (in that lemma letting $P_k$ be the tile of $P$ at height $h$), we get that $x_{P_i } < x_{P_j}$ (\resp, of Corollary~\ref{cor:columns}, that $x_{P_j} < x_{P_i}$), which shows Conclusion~\ref{b2}.

Finally, since the region we chose to apply the pigeonhole argument is located immediately to the right of $\ell$ (\resp, left) and is of width merely $ |U|(3m+1)+1$  neither $\glueP{i}{i+1}$ nor $\glueP{j}{j+1}$ are in the ``vertical part'' of the simulation zone (Figure~\ref{fig:fuzz}), and since the ``horizontal part'' of the simulation zone is of height $3m$, this proves Conclusion~\ref{e}. 
\end{proof}

\subsubsection{Blocking any \hs path by growing a branch from it}

In this section we give Lemma~\ref{lem:onepath} which is the first main tool used in this paper. We also give Theorem~\ref{thm:onepath} whose short proof gives a method to block any \hs path. We begin with the definition of an \emph{enclosing branch}, which is a path $D$ branching from $P$, and enclosing a connected component of $\R^2$. The enclosing branch achieves this in one of two ways: (1) either by intersecting $\sigma\cup P$ (see Figure~\ref{fig:enclosing}(Left)) and hence the enclosure is bordered by $P$, the enclosing branch and possibly $\sigma$, or (2) by placing a new visible glue on $\ell$ (see Figure~\ref{fig:enclosing}(Right)) and hence the enclosure is bordered by $P$, the enclosing branch, a segment of $\ell$ and  possibly $\sigma$.

\begin{definition}\label{def:enclosingbranch}
  [Enclosing branch for a path]
  Let $P\in\hsU$ be a \hs path and for any~$\ell$  that satisfies the hypotheses of Lemma~\ref{lem:visibility-setup}: 
  Let $k \in \{0,1,\ldots , |P|-1\}$ be such that $P_{0,1,\ldots,k}$ includes $\namedGlue{P_\ell}{P_{\ell+1}}$, $P$'s visible glue  on $\ell$.  We call a path $D$ an \emph{enclosing branch} for $P$ at~$k$ if $P_{0,1,\ldots,k}D_{0,1,\ldots,|D|-2} \in\prodpathsU$ is an assemblable path, $\pos{P_{k+1}} = \pos{D_0}$,  $\glueP{k}{k+1} = \gluePD{k}{0}$ and is visible relative to  $P_{0,1,\ldots,k}D_{0,1,\ldots,|D|-2} \in\prodpathsU$, and at least one of the following is the case:
  \begin{enumerate}
  \item \label{def:enclosingbranch:DconflictsPrefixOfP} at $\pos{D_{|D|-1}}$ the path $D$ intersects  $\sigma\cup \, \asm{P_{0,1,\ldots,k}}$,\footnote{and therefore $D$ cannot grow completely from $\sigma \cup \, \asm{ P_{0,1,\ldots,k}}$ without intersecting $\sigma\cup \,\asm{P_{0,1,\ldots,k}}$} or
  
  \item \label{def:enclosingbranch:belowl} $\midpoint{D_{|D|-2}}{D_{|D|-1}} \in \R^2$ is on $\ell$ and strictly below the points in the set $\frak{E}_P \cap \ell$.\footnote{And therefore $D$ cannot grow completely from $\sigma \cup \, \asm{ P_{0,1,\ldots,k}}$ without placing a glue on $\ell$ that is lower than $P$'s visible glue on $\ell$.} 
  \end{enumerate}
\end{definition}

    \begin{figure}[ht]
      \begin{center}
        \includegraphics{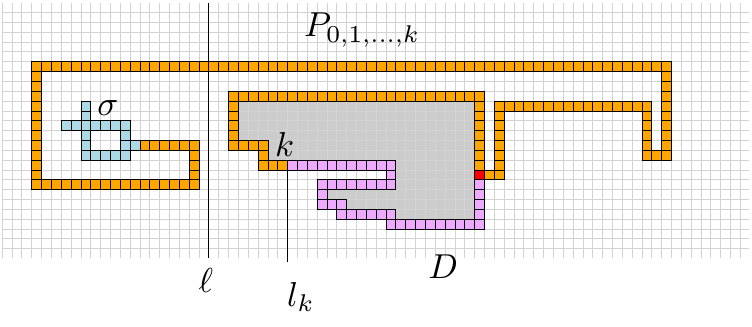}\hspace{14ex}
        \includegraphics{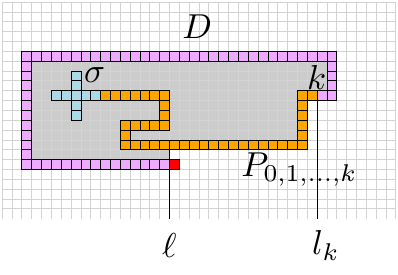}
      \end{center}
      \caption{Enclosing branches that enclose the part of the plane shaded grey. {\bf Left:} Example of an enclosing branch $D$ for $P$ at $i$ for   Definition~\ref{def:enclosingbranch}(\ref{def:enclosingbranch:DconflictsPrefixOfP}). $D$ is the pink and red path, the last tile of $D$, which intersects $P$, is drawn in red (hence $PD$ is not a path because its last tile intersects an earlier part of itself). 
      {\bf Right:} Example of an enclosing branch $D$ for $P$ at $i$ for  Definition~\ref{def:enclosingbranch}(\ref{def:enclosingbranch:belowl}). 
      $D$ is the pink and red path, the last tile of $D$ (immediately after $D$ places a new visible glue on $\ell$) is drawn in red. Note that Figure~\ref{fig:enclosingwithl} shows a concrete construction of such an enclosing branch.}
      \label{fig:enclosing}
    \end{figure}

Before stating Lemma~\ref{lem:onepath}, we give an intuitive, although imprecise, {\em summary} of its statement, of {\em how} we obtain such a statement and {\em why} we would want such a statement. 
Let $P$ be any \hs path, which implies that $P$ is long enough to repeat some glue type many times. This implies one of two things (which are the two conclusions of the lemma): 
  \begin{enumerate}
  \item There is a path, assemblable by $\calU$, that places tiles outside of the simulation zone.
    If the proof yields this conclusion, then such a path is found either by  repeating one part of $P$ enough times (this is point \ref{qinf} in the proof), or else by deriving a \hs path $P'$ from $P$, and ``shortcutting'' a part of $P'$, so as to translate a suffix of $P'$ by a large enough vector.
    In both cases, this contradicts that $\calU$ simulates $\calT_{10|U|}$, which gives the  proof of our main theorem.

  \item Else, we can grow a path of the form  $P_{0,1,\ldots,k}D_{0,1,\ldots |D|-2 }$, for some $k\in\{i,j \}$, which has the following properties:

    \begin{itemize}
    \item $\glue{k}{k+1}$, which is visible relative to $P$, is also visible relative to   $D$, and 
    \item any \hs path that turns to the right from $D_{0,1,\ldots}$ 
    must hide the visibility of $\glue{k}{k+1}$.
    \end{itemize}

    First we observe that $P$ must be blocked by such a  $P_{0,1,\ldots,k}D_{0,1,\ldots}$ (i.e.\ $P$ cannot grow from an assembly that already contains $P_{0,1,\ldots,k}D_{0,1,\ldots}$), because we can show  that $P$ turns to the right from $D_{0,1,\ldots,a}$\footnote{From the definition of enclosing branch $P$ places a tile at $\pos{D_0}$ therefore either $P$ branches to the right (see proof of Theorem~\ref{thm:onepath})  from $D$, or if it happens that  $\type{P_{k+1}} \neq \type{D_0} $ then $P$ is blocked by $D$.}, but this causes a contradiction as $P$ can not hide its own visible glue. This is one useful fact from Conclusion~\ref{onepath:conclusion:allbroken}. 

    The second useful fact (and the main reason for the particular form of Conclusion~\ref{onepath:conclusion:allbroken}), is that we have found a way to consume a finite resource: visible glues (note that the entire set of \hs paths is finite and each such path has a finite set of visible glues).
    Then, later in the proof of Theorem~\ref{thm:main}, we will use this ``consumption of visible glues''  property  to show that the ``supply of visible tiles'' must decrease with each new ``branch'' to the right starting from $P_{0,1,\ldots,k}D_{0,1,\ldots |D|-2}$; and the particular form of Conclusion~\ref{onepath:conclusion:allbroken} allows us to do this.
\end{enumerate}
The previous intuitive overview of the lemma statement is somewhat incomplete, but may serve as a guide. All the steps of the {\em proof} are given in Figure~\ref{fig:onepath}, with a running pictorial example, and short summaries of each step.

%\begin{minipage}{\textwidth} % minipage prevents page-breaking this important lemma
\begin{lemma}[Pumping or enclosing any \hs path]
  \label{lem:onepath}
  \label{lem:one-path}
  Let $P\in\hsU$ be a \hs path, let $i,j,\ell$ be as in Lemma~\ref{lem:visibility-setup} 
  with $P$'s visible glue on $\ell$ being in $V_P^+$ (\resp,  $V_P^-$).
  At least one of the following holds:
  \begin{enumerate}
  \item   There is an assemblable path (i.e.\ from $\prodpathsU$) 
  that places tiles outside of the simulation zone of~$\calU$, contradicting that $\calU$ simulates $\calT_{10|U|}$. 
    \label{onepath:conclusion:breakscale}

  \item (i) There is an enclosing branch $D = D_{0, 1, \ldots, |D|-1}$ for $P$ at some $k \in \{i,j \} $, 
  such that
  $P_{0,1,\ldots,k}D_{0,1,\ldots,|D|-2} \in \prodpathsU$ is an assemblable \nhs path that
  has the same visible glue on $\ell$ as $P$.%,

    (ii) Furthermore, for all paths $R$ such that for some $a\leq |D|-2$, $P_{0,1,\ldots,k}D_{0,1,\ldots,a}R$ is \hs and turns right (\resp, left) from $P_{0,1,\ldots,k}D_{0,1,\ldots,|D|-2} $,   
 at least one of the following is the case:
    \begin{itemize}
    \item $\glue{k}{k+1}$ is not visible relative to $P_{0,1,\ldots,k}D_{0,1,\ldots,a}R$, or 
    \item $R$ has a lower visible glue on $\ell$ than $P$.
    \end{itemize}

    \label{onepath:conclusion:allbroken}
\end{enumerate}
\end{lemma}
%\end{minipage}
\begin{sidewaysfigure}
    \begin{center}
      \input{shortcut-edited.tex}

      \caption{A summary of the major points in the proof of Lemma~\ref{lem:onepath}. The proof essentially performs a depth-first, left-first, search of this tree with each leaf ending in one of the lemma conclusions. Points in the tree and proof are numbered \ref{qinf}, \ref{qfin}, $\ldots$, \ref{conclusion} correspondingly.}
      \label{fig:onepath}
    \end{center}
\end{sidewaysfigure}
\begin{proof}
  A tree representing the proof is shown in Figure~\ref{fig:onepath}. The proof performs a depth-first, left-first, search of the tree. Points in the tree and proof are numbered \ref{qinf}, \ref{qfin}, $\ldots\,$, \ref{conclusion} correspondingly.

  \begin{figure}[H]
    \begin{center}
      \includegraphics{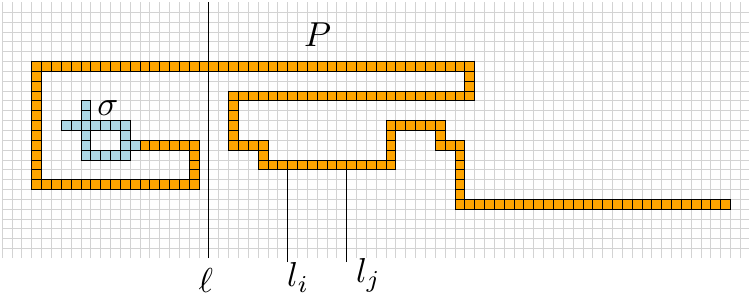}\hspace{3ex}
      \includegraphics{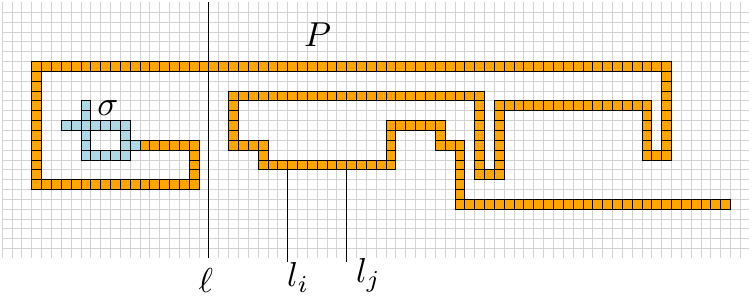}
    \end{center}
    \caption{Two example \hs paths. For brevity, their suffixes that grow to height $h$ are not shown.  
    Here, and in all figures in the paper, $l_i$ and $l_j$ denote the two visibility rays of the glues $\glueP{i}{i+1}$ and $\glueP{i}{i+1}$ respectively. 
    {\bf Left}: A path $P$ where the segment  $P_{i+1,i+2,\ldots,j}$ can be (infinitely) pumped (see Figure~\ref{fig:onepath-inf}).
    {\bf Right}: A path $P$ where the segment  $P_{i+1,i+2,\ldots,j}$ can not be (infinitely) pumped (see Figure~\ref{fig: onepath}).}
  \end{figure}

  Intuitively, let $\olq$ be the sequence composed of $P_{0,1,\ldots,i}$ followed by the ``infinite pumping'' of the segment $P_{i+1,i+2,\ldots,j}$. Formally,  let $\olq$ be the pumping of $P$ between $i$ and $j$ (Definition~\ref{def:pumpingPbetweeniandj}), 
  and let $Q \in \prodpathsU$ be the longest  prefix of $\olq$ that (a) is a path (i.e.\ non-self intersecting), (b) does not intersect $\sigma$,  and (c) is assemblable from the seed $\sigma$ of $\calU$.
  $Q$ is either infinite or  finite.
  \begin{enumerate}[label=P\arabic*]
  \item \label{qinf} If $Q$ is infinite (as in the example in Figure~\ref{fig:onepath-inf}):
      Then since $\vect{P_iP_j}$ has a nonzero horizontal component (by Lemma~\ref{lem:visibility-setup} the horizontal distance between  $\glueP{i}{i+1}$ and $\glueP{j}{j+1}$  is $\geq 3m$), this infinite assembly reaches infinitely to the right for glue $g$ in the lemma statement being in $V_{P}^{+}$ (\resp, to the left for $g$ being in $V_{P}^{-}$), and thus places tiles outside of the simulation zone, giving Conclusion~\ref{onepath:conclusion:breakscale}.

      \begin{figure}[h]
        \begin{center}
          \includegraphics{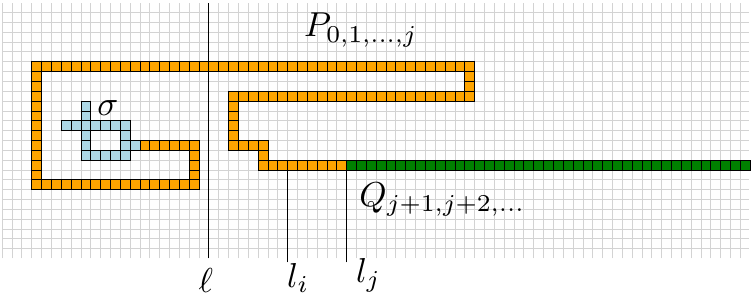}
        \end{center}
        \caption{Case \ref{qinf}:  
        $Q$ is an infinite path because the sequence $\olq$ does not intersect $P_{0,1,\ldots,j}$ nor $\sigma$. Note that $Q$ is defined so that it shares a prefix with $P$:  i.e.\   $P_{0,1,\ldots,j} = Q_{0,1,\ldots,j}$.\label{fig:onepath-inf}}
      \end{figure}
    \item \label{norrp}
   
   If some prefix $Q'$ of $Q$ is \hs we claim that we get Conclusion~\ref{onepath:conclusion:breakscale} of the lemma statement. 
 To see this, note that $P_{0,1,\ldots,j}$ is  \nhs. Therefore, if $Q'$ were \hs, $Q$ 
 would place a tile $Q_s = Q_{|Q'|-1}$ (where $|Q'| \in \N$ is the length of $Q'$) on the horizontal line $y=h$, which means that $Q$ would also have a tile at position $\pos{Q_s - \vect{P_i P_j}}$ (i.e.\ positioned one iteration of the ``pumping'' earlier, see Figure~\ref{fig:onepath-out-of-sim-zone}). 
 But  $\pos{Q_s - \vect{P_i P_j}}$ is outside of the simulation zone: $\pos{Q_s}$ is in the simulation zone on line~$h$, the ``vertical part'' of the simulation zone (that intersects line~$h$ and for distance $7m$ below~$h$)  is of width $3m$,  $\vect{P_i P_j}$ has horizontal length $\geq 3m$ and vertical length $\leq 3m$ (Lemma~\ref{lem:visibility-setup}). 
 This immediately gives Conclusion~\ref{onepath:conclusion:breakscale}.

      \begin{figure}[H]
        \begin{center}
          \includegraphics{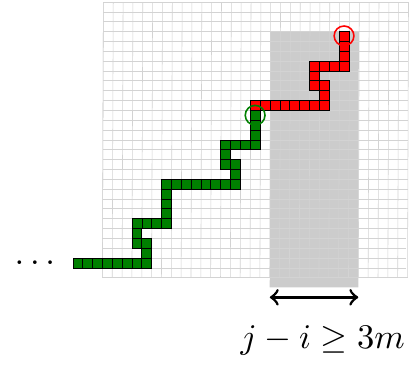}
        \end{center}
        \caption{Case \ref{norrp}: $Q$ (in green and red in the picture) is \hs. In this case, the last tile $Q_{|Q|-1}$ of $Q$ (circled in red) is at height $h$. Since $|x_{P_j}-x_{P_i}|\geq 3m$, and because $Q_{i,i+1,\ldots}$ is periodic, $Q$ also contains tile $Q_{|Q|-1} - \vect{P_jP_i}$ (in red in the picture), which is outside of the simulation zone (in gray).}
        \label{fig:onepath-out-of-sim-zone}
      \end{figure}

    \item \label{qfin}
      Else $Q$ is finite, and  \nhs. 
      We assert the claim that $\olq_{|Q|}$ intersects $\sigma\cup \asm{P_{0,1,\ldots,j}}$, which we will next prove.\footnote{Note that we already know  that $\olq_{|Q|}$ intersects $\sigma$ or $\olq_{0,1,\ldots,|Q|-1}$.
      Intuitively, if $\olq_{j+1,j+2\ldots}$ intersects itself, then since $\olq_{i,i+1,\ldots}$ is periodic and $i<j$, we can find another self-intersection $j-i$ indices earlier along $\olq$. Applying this argument repeatedly will show an intersection with $P_{i, i+1, \ldots, j}$.}

\newcommand{\kmin}{a}

      If $\olq_{|Q|}$ intersects $\sigma$, we are immediately done. Else, $\olq$ intersects itself. Let  $\kmin < |Q|$ be the {\em smallest} integer such that $ \pos{\olq_\kmin} = \pos{\olq_{|Q|}} $. We will show that $a\leq j$. 
      Assume, for the sake of contradiction that $\kmin > j$. 

      Note that $|Q|-(j-i) \geq \kmin-(j-i)\geq i+1$.
      Therefore, 
      we can apply Lemma~\ref{lem:torture} to show that
      $Q_{\kmin-(j-i)} = Q_\kmin - \vect{P_iP_j}$ and
      $Q_{|Q|-(j-i)} = \olq_{|Q|} - \vect{P_iP_j}$. 
      But  then since $\pos{Q_\kmin}  = \pos{\olq_\kmin} = \pos{\olq_{|Q|}}$, and since $a>j$,\footnote{intuitively, we are in the $\geq 2$nd iteration of pumping} we get that
      $\pos{Q_{\kmin-(j-i)}} = \pos{Q_{\kmin}} - \vect{P_iP_j} = \pos{\olq_{|Q|}} - \vect{P_iP_j} = \pos{Q_{|Q|-(j-i)}}$.
      Therefore, $Q$ intersects itself which is a contradiction. 
      Hence $\kmin\leq j$ and thus
      $\olq_{|Q|}$ intersects $\sigma\,\cup\, \asm{P_{0,1,\ldots,j}}$ as claimed.

      An example of this case is shown in Figure~\ref{fig: onepath}.
      \begin{figure}[H]
        \begin{center}
          \includegraphics{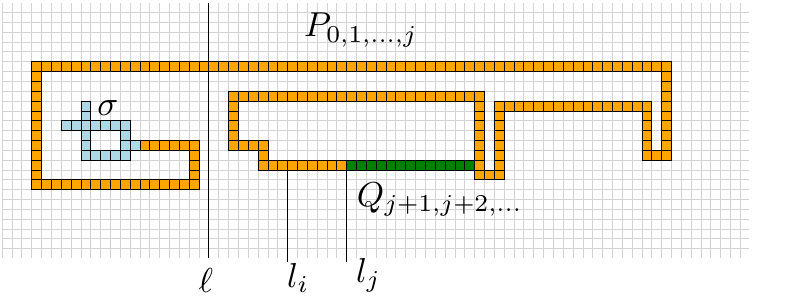}
        \end{center}
        \caption{Case \ref{qfin}: $\olq$ conflicts with $P_{0,1,\ldots,j}$. The figure shows $\sigma$ and the path $Q$. Note that $Q$ is defined so that it shares a prefix with $P$:  i.e.\   $P_{0,1,\ldots,j} = Q_{0,1,\ldots,j}$. }
        \label{fig: onepath}
      \end{figure}

    \item \label{qnoths}
      Let $\mathcal R$ be the set of all \hs paths of the form $Q_{0,1,\ldots,r}R \in \prodpathsU$, for any $r$ such that $j < r \leq |Q|-1$ and any path $R$, such that all of the following hold:
      \begin{enumerate}[label=\textnormal{(\roman*)}]
      \item $Q_{0,1,\ldots,r}R$ turns right (\resp, left) from $Q$,\label{qnoths:turnright}
      \item \label{qnoths:jVis} $\glueP{j}{j+1}=\glueQ{j}{j+1}$ is visible  relative to $Q_{0,1,\ldots,r}R$,
      \item \label{qnoths:ell} $Q_{0,1,\ldots,r}R$ has the same  visible glue as $P$ on $\ell$.
      \end{enumerate}

      {\bf If $\mathcal R$ is empty\footnote{which includes in particular the case where no \hs path can branch from $Q$}}, then we claim that we get Conclusion~\ref{onepath:conclusion:allbroken} of the lemma statement, with $k=j$ and $D=\olq_{j+1,j+2,\ldots,|Q|}$. Indeed, $D$ is an enclosing branch for $P$ at $j$ with the {\em same visible glue} on $\ell$ as~$P$.
      Moreover, by Lemma~\ref{lem:semiprecious} (\resp, Corollary~\ref{cor:semiprecious-minus}), $Q$ does not hide the visibility of $\glueP{i}{i+1}$ or $\glueP{j}{j+1}$, hence $D$ satisfies Conclusion~\ref{onepath:conclusion:allbroken}(i).
      Then, Conclusion~\ref{onepath:conclusion:allbroken}(ii) is immediately satisfied since $\mathcal{R}$ is empty (intuitively, $\mathcal R$ is the set of paths that do not meet that conclusion).

  \item\label{push}
    {\bf Else, $\mathcal R$ is not empty}.
    Until the end of this proof, let $r$ and $R$ be such that $Q_{0,1,\ldots,r}R$ is the most right-priority (\resp, left-priority) of $\mathcal R$. Note that since all paths of $\mathcal R$ are \hs, the last tile of  $R_{|R|-1}$ is  the only tile of $R$ at height $h$. Such a path is shown in Figure~\ref{fig:onepathR}.
    Notice that, because $R$ is the most right-priority path of $\mathcal R$, $R$ does not turn left from $Q$ (before turning right from $Q$), or else we could find a \hs path turning right earlier than $R$.

    \begin{figure}[H]
      \begin{center}
        \includegraphics{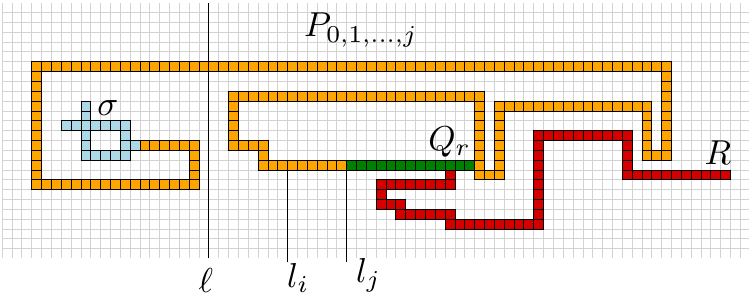}
      \end{center}
      \caption{Case~\ref{push}: there is at least one path in $\mathcal R$. Let $Q_{0,1,\ldots,r}R \in \mathcal{R}$ be the most right-priority one.}
      \label{fig:onepathR}
    \end{figure}
    We grow $\sigma\cup \asm{P_{0,1,\ldots,i}} = \sigma\cup \asm{Q_{0,1,\ldots,i}}$, and then the maximal assemblable prefix of the following ``translated path'': $(Q_{j+1,j+2,\ldots,r}R)+\vect{P_jP_i}$. In other words, starting from $\sigma$, we grow the following path: $Q_{0,1,\ldots,i}((Q_{j+1,j+2,\ldots,r}R)+\vect{P_jP_i})$ (Figure~\ref{fig:onepath-pull-cangrow} below highlights this ``backwards translation'').
  \item\label{rrpgrows}
      If all of $(Q_{j+1,j+2,\ldots,r}R)+\vect{P_jP_i}$ is assemblable from $\sigma\cup \asm{ P_{0,1,\ldots,i} }$, as in the example in Figure~\ref{fig:onepath-pull-cangrow}, then we claim that $\pos{ R_{|R|-1} +\vect{P_jP_i}} $ is outside the simulation zone, yielding  Conclusion~\ref{onepath:conclusion:breakscale}. To see this
      first note that since $R$ is \hs the tile $R_{|R|-1}$ (i.e.\ untranslated) is on  the horizontal line $y=h$ and thus in the ``vertical part'' of the simulation zone (see Figure~\ref{fig:fuzz}) which is of width $3m$.  
      But, by Lemma~\ref{lem:visibility-setup}, $\vect{P_jP_i}$ has  horizontal length  $\geq 3m $  and  vertical length $\leq 3m$. This means that $\pos{ R_{|R|-1} +\vect{P_jP_i}} $ is outside of the simulation zone, giving  Conclusion~\ref{onepath:conclusion:breakscale}. 

      \begin{figure}[H]
        \begin{center}
          \includegraphics{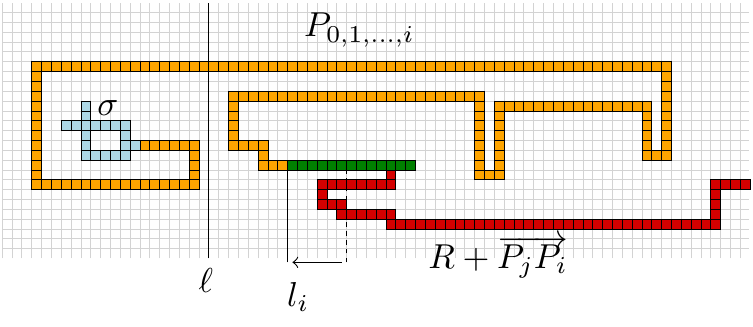}
        \end{center}
        \caption{Case~\ref{rrpgrows}: $(Q_{j+1,j+2,\ldots,r}R)+\vect{P_jP_i}$ is assemblable from $\sigma\,\cup\, \asm{ P_{0,1,\ldots,i} }$.}
        \label{fig:onepath-pull-cangrow}
      \end{figure}

    \item \label{rrpfails}
      Else not all of $(Q_{j+1,j+2,\ldots,r}R)+\vect{P_jP_i}$ is assemblable from $\sigma\,\cup\, \asm{ P_{0,1,\ldots,i} }$.
      Let $R'$ be the longest prefix of $R$ such that $R'+\vect{P_jP_i}$ does not intersect  $\sigma\cup \asm{ P_{0,1,\ldots,i} }$ and $R'+\vect{P_jP_i}$ does not have any visible glue on $\ell$ below the visible glue of $P$ on $\ell$. 
      (See Figure~\ref{fig:onepath-pull} for an example where the longest assemblable prefix of $R+\vect{P_jP_i}$ does not place any new visible glue on $\ell$, and Figure~\ref{fig:enclosingwithl} for an example where it does.) 
      
      In this case, we make the more specific claim that $R+\vect{P_jP_i}$ conflicts with $\sigma\cup \asm{ P_{0,1,\ldots,i} }$: indeed, by its definition in \ref{push} above, $R$ does not conflict with $Q_{0,1,\ldots,r}$, and hence does not conflict with
      $Q_{j+1,j+2,\ldots,r}$, 
      hence $R+\vect{P_jP_i}$ does not conflict with $Q_{j+1,j+2,\ldots,r}+\vect{P_jP_i} = Q_{i+1} Q_{i+2} \ldots Q_{r-(j-i)}$.  
      The only part of the assembly $\sigma\cup\asm{ P_{0,1,\ldots,i}} \cup (\asm{Q_{j+1,j+2,\ldots,r}  +\vect{P_jP_i} } )$ that $R+\vect{P_jP_i}$ can conflict with is therefore $\sigma\cup \asm{ P_{0,1,\ldots,i}}$.

      Observe that $(Q_{j+1,j+2,\ldots,r}R_{0,1,\ldots,|R'|})+\vect{P_jP_i}$ is an enclosing branch for $P$ at $i$ (Definition~\ref{def:enclosingbranch}\footnote{Note that both of the cases (1) and (2) of Definition~\ref{def:enclosingbranch} can happen here.}).

      \begin{figure}[H]
        \begin{center}
          \includegraphics{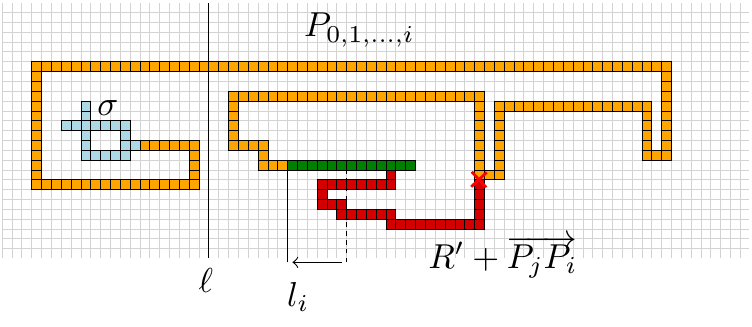}
        \end{center}
        \caption{Case~\ref{rrpfails}: Not all of $(Q_{j+1,j+2,\ldots,r}R)+\vect{P_jP_i}$ is assemblable from $\sigma\,\cup\, \asm{ P_{0,1,\ldots,i} }$.
        }
        \label{fig:onepath-pull}
      \end{figure}

      \begin{figure}[h] % option H from float pckage
        % \vspace{-\baselineskip}
        \begin{center}
          \centerline{
            \includegraphics[valign=b]{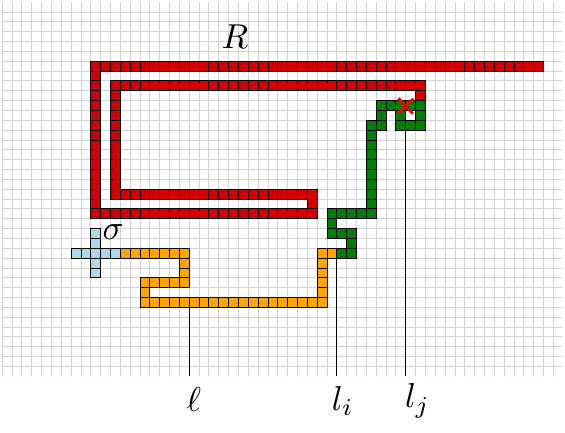}\hspace{3ex}
            \includegraphics[valign=b]{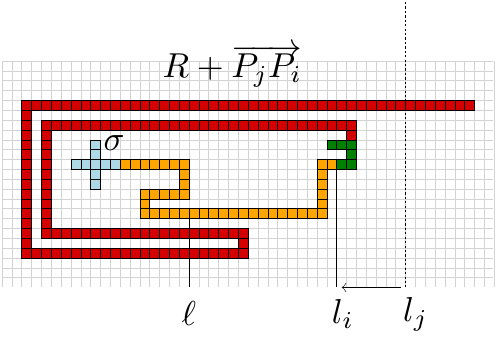}\hspace{3ex}
            \includegraphics[valign=b]{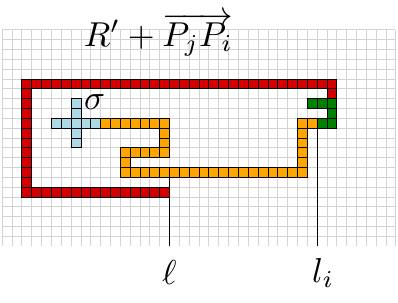} }
        \end{center}
        \vspace{-1.0\baselineskip}
        \caption{{\bf Left}: an example path $R$ from \ref{rrpfails} (the red cross is an intersection of $Q$ with itself). {\bf Centre}: Example  for \ref{rrpfails},  i.e.\ $R+\vect{P_jP_i}$ places a glue on $\ell$ below the visible glue of $P$ on $\ell$.
          {\bf Right}:  $R'+\vect{P_jP_i}$  is the largest prefix of $R+\vect{P_jP_i}$ that does not place a visible glue on $\ell$ below the visible glue of $P$ on $\ell$.}\label{fig:enclosingwithl}
      \end{figure}

    \item \label{norestart}
      We now consider the set $\mathcal S$ of \hs paths of the form $P_{0,1,\ldots,i}XS$
      for some prefix $X$ of $(Q_{j+1,j+2,\ldots,r}R') +\vect{P_jP_i}$ and some path $S$, such that all of the following hold:
      \begin{itemize}
        \item $P_{0,1,\ldots,i}XS$ turns right (\resp, left) from $P_{0,1,\ldots,i}((Q_{j+1,j+2,\ldots,r}R') +\vect{P_jP_i})$,
        \item $\glueP{i}{i+1}$ is visible relative to $P_{0,1,\ldots,i}XS$, and 
        \item the visible glue of $P$ on $\ell$ is visible relative to $P_{0,1,\ldots,i}X S$.
      \end{itemize}

      There are two cases:
      \begin{enumerate}

      \item \label{norestart1} If $\mathcal S$ is empty this gives Conclusion~\ref{onepath:conclusion:allbroken} with $D= Q_{j+1,j+2,\ldots,r}R_{0,1,\ldots,|R'|} +\vect{P_jP_i}$ and $k=i$. In particular, as noted in case~\ref{rrpfails}, $D$ is an enclosing branch for $P$ at $i$, which satisfies Conclusion~\ref{onepath:conclusion:allbroken}(i).

        Moreover, the fact that $\mathcal S$ is empty immediately shows Conclusion \ref{onepath:conclusion:allbroken}(ii) (intuitively, $\mathcal S$ is the set of paths that do not meet that conclusion).

      \item Else, there is at least one path $P_{0,1,\ldots,i}X S$  in $\mathcal S$.
        See Figure~\ref{fig:onepath-pull-restart} for two examples.
        \label{restart}

        \begin{figure}[ht]
          \begin{center}
            \includegraphics{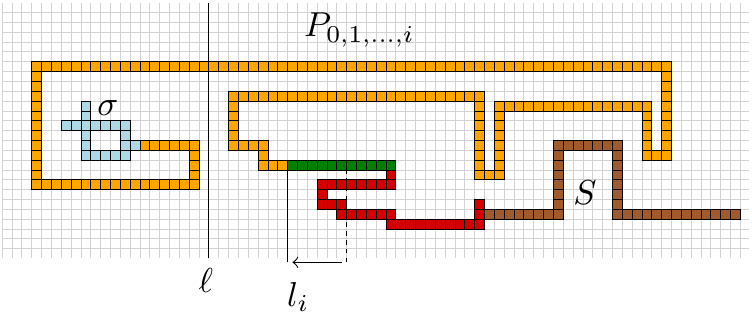}
            \hspace*{0.1\textwidth}
            \includegraphics{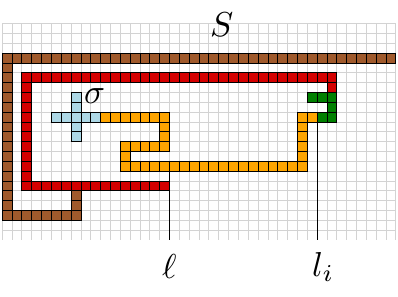}
          \end{center}
          \caption{Two different example paths $S$ that illustrate Case~\ref{restart}: At least one \hs path $XS$ can turn right from $(Q_{j+1,j+2,\ldots,r}R')+\vect{P_jP_i}$ without hiding the visibility of $\glueP{i}{i+1}$ nor of the visible glue of $P$ on $\ell$.}
          \label{fig:onepath-pull-restart}
        \end{figure}
    \end{enumerate}

  \item \label{conclusion}
    The only unresolved case after step~\ref{norestart} is therefore case~\ref{restart}, which we now reason about with the goal of obtaining Conclusion~\ref{onepath:conclusion:allbroken}.

    We next  grow the longest assemblable prefix of the ``forward translated'' segment $XS+\vect{P_i P_j} = (XS)+\vect{P_i P_j}$  and use this to  show  that in all remaining cases we get Conclusion~\ref{onepath:conclusion:allbroken}. 
     See Figure~\ref{fig:onepath-push} for an example.

    For notation, let  $X' = X + \vect{P_iP_j}$ and $S' = S + \vect{P_iP_j}$. 
    Suppose that the visible glue of $P$ on $\ell$ is in $V_P^+$ (respectively, in $V_P^-$).

    From \ref{restart},  $S$ does not hide the visibility of $ \namedGlue{P_i}{P_{i+1}}$.
    Therefore,  $S'$ does not place a glue directly below
    $\midpoint{P_j}{P_{j+1}}$
    either. 
    Notice that $P_{0,1,\ldots,j}X'$ is assemblable (as it is a prefix of $P_{0,1,\ldots,j} Q_{j+1 ,\ldots,r}  R $), and let $s$ be the largest integer such that
    $P_{0,1,\ldots,j}X'S'_{0,1,\ldots,s-1}$ is assemblable and has the same visible glue as $P$ on $\ell$.

    If $s = |S|$, then the last point of $S'$ is outside of the simulation zone, because the last tile of~$S$ (i.e.\ untranslated) is at height $h$. This yields Conclusion~\ref{onepath:conclusion:breakscale}.

    Else, $s < |S|$.
    Moreover, we picked $S$ in case~\ref{restart} so that $XS$ turns right (\resp, left) at least once from $(Q_{j+1,j+2,\ldots,r}R')+\vect{P_jP_i}$, therefore $S'$ also turns right (\resp, left) at least once from $ Q_{j+1,j+2,\ldots,r}R'  $.

    \begin{figure}[ht]
      \begin{center}
        \includegraphics{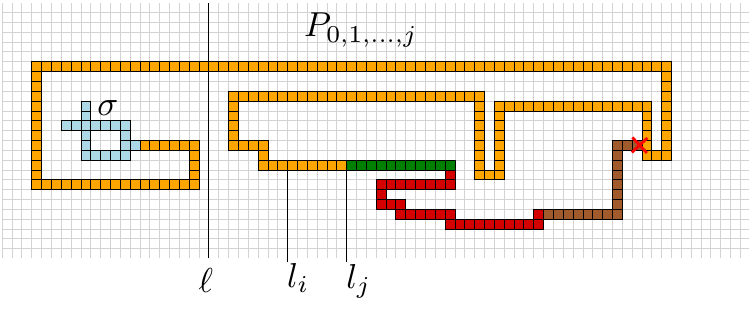}
      \end{center}
      \caption{Case~\ref{conclusion}: We have just constructed an enclosing branch for $P$ at $j$. In the drawing, that enclosing branch is made of the path following the green, red and then brown tiles. In this example, $X$ is the concatenation of the green ($Q_{j,j+1,\ldots,r}$) and red (a prefix of $R'$) paths and $S$ is in brown (but in general  $S$ can start from any of the green or red tiles).}
      \label{fig:onepath-push}
    \end{figure}

    Let $D = X'S'_{0,1,\ldots,s}$, and let  $Z$ be any path such that $P_{0,1,\ldots,j}D_{0,1,\ldots,a}Z$ is \hs and turns right (\resp, left) from $P_{0,1,\ldots,j}D_{0,1,\ldots,|D|-2}$, for some $a\geq 0$.
    Note that $P_{0,1,\ldots,j}D_{0,1,\ldots,a}Z$ cannot be in $\mathcal R$, because then $P_{0,1,\ldots,j}D_{0,1,\ldots,a}Z$ would more right-priority (\resp, left-priority) than $Q_{0,1,\ldots,r} R $, $r>j$ (contradicting our choice of $R$ in~\ref{push}) since:  $P_{0,1,\ldots,j}D_{0,1,\ldots,a}Z$ either turns right (\resp, left) from $Q$ earlier than $R$, or turns right (\resp, left) from $R$, or turn right (\resp, left) from $S'$ (which turns right (\resp, left) from $Q_{j+1, j+2 \ldots r} R'$). 

    Therefore,\footnote{I.e. we have a path $P_{0,1,\ldots,j}D_{0,1,\ldots,a}Z$ of the form $Q_{0,1,\ldots r} R $ (see~\ref{qnoths}), and that turns right from $Q$ (satisfies \ref{qnoths}\ref{qnoths:turnright}), yet is not in $\mathcal{R}$ hence $P_{0,1,\ldots,j}D_{0,1,\ldots,a}Z$ violates Conditions \ref{qnoths}\ref{qnoths:jVis} and \ref{qnoths}\ref{qnoths:ell}.}  $Z$ must either have its visible glue on $\ell$ lower than that of $P$, or hide the visibility of $\glueP{j}{j+1}$. This is precisely Conclusion~\ref{onepath:conclusion:allbroken} with $k=j$. Notice that~$D$ is in fact an enclosing branch for $P$ at $j$ because:  neither $R'$ nor $S'$ hide the visibility of $\glueP{j}{j+1}$, and by Lemma~\ref{lem:semiprecious} (\resp Corollary~\ref{cor:semiprecious-minus}), neither does $Q_{j,j+1,\ldots}$, and finally that $D$ is composed of subpaths from these paths.

  \end{enumerate}
\end{proof}

The following theorem (\ref{thm:onepath}) essentially states that for any path $P$, we can grow an assembly containing no \hs path, conflicting with $P$. 
 The proof is almost a direct consequence of Lemma~\ref{lem:onepath}.  Note that we can think of Theorem~\ref{thm:onepath} as a weaker  version of our main result (Theorem~\ref{thm:main}). That main result (Theorem~\ref{thm:main}) builds a single assembly containing no \hs path and that conflicts with {\em all} possible \hs paths.

\begin{theorem}
  Let $P\in\hsU$ be a \hs path. Then either there is a producible assembly $\alpha\in\prodasm{\mathcal{U}}$ with tiles outside of the simulation zone, or else there is an assemblable \nhs path of the form %$D$ such that
  $P_{0,1,\ldots,k}D_{0,1,\ldots,|D|-2}$ that conflicts with~$P$, and thus $\sigma\cup\asm{P_{0,1,\ldots,k}D_{0,1,\ldots,|D|-2}}$ prevents $P$ from growing to be \hs.
  Moreover, $D$ is constructed as in Lemma~\ref{lem:onepath}. 
\label{thm:onepath}
\end{theorem}
\begin{proof}
  We apply Lemma~\ref{lem:onepath}. If we get Conclusion~\ref{onepath:conclusion:breakscale}, we are immediately done (we get $\alpha$ in the statement). Else, let $i,j,\ell$ be as defined in Lemma~\ref{lem:visibility-setup},  let $k\in\{i,j\}$ and let  $D$ be the enclosing branch  constructed in Conclusion~\ref{onepath:conclusion:allbroken} of Lemma~\ref{lem:onepath}, assuming $P$ places a $V_P^+$ (\resp, $V_P^-$) glue on $\ell$.
  We begin by defining a connected component $\mathcal C$. 
  There are two cases: 
  \begin{itemize}
  \item If $D$  intersects $\sigma\cup \asm{ P_{0,1,\ldots,k}}$ let
  $\overline{P}$ be any path from $\pos{D_{|D|-1}}$ (the first such intersection) to $\pos{P_k}$ in the grid graph of $\sigma\cup \asm{ P_{0,1,\ldots,k}}$. 
 Let $\mathcal C$ be the bounded connected component of $\mathbb{R}^2$ enclosed by the concatenation of $\frak{E}_{\overline{P}}$ and $\frak{E}_{D}$
   (the canonical embedding of the paths $\overline{P}$ and $D$, respectively).

  \item Else, $D=D_{0,1,\ldots,|D|-1}$ places a glue, denoted $\namedGlue{D_{|D|-2}}{D_{|D|-1}}$, on $\ell$ below the visible glue of $P$ on $\ell$, and $D$ does not intersect  $\sigma\cup \asm{ P_{0,1,\ldots,k}}$. In this case, we let $c$ be the concatenation of the following four curves  (where $\glue{\ell}{\ell+1}$ is the visible glue of $P$ on $\ell$):  
     \begin{equation*}
       \begin{split}
     &     [ \midpoint{P_\ell}{P_{\ell+1}} , \pos{P_{\ell+1}} ] \\
    & \frak{E}_{P_{\ell +1,\ldots,k}D_{0,1,\ldots,|D|-2}} \\
  & [ \pos{D_{|D|-2}} ,  \midpoint{D_{|D|-2}}{D_{|D|-1}}  ] \\
  & [  \midpoint{D_{|D|-2}}{D_{|D|-1}}  ,  \midpoint{P_\ell}{P_{\ell+1}}] 
       \end{split}
       \end{equation*}
          By Observation~\ref{obs:ClosedSimpleCurve}, 
     $c$ is a finite closed simple curve and thus defines a bounded connected component $\mathcal{C}$. 
  \end{itemize}
In both cases, by the fact that $D$ is an enclosing branch, $\gluePD{k}{0}$ is visible relative to $P_{0, 1 ,\ldots,k}D_{0,1,\ldots,|D|-2}$, which implies that the left-hand side (\resp, right-hand side)  of $\gluePD{k}{0}$ is {\em inside} $\mathcal C$ (when walking in the direction from $P_k$ to $D_0$).

Also,  in both cases,  $P_{0,1,\ldots,k} D_{0,1,\ldots,|D|-2}$ is not \hs (by Lemma~\ref{lem:onepath}).  
Therefore, if~$P$ can still grow to be \hs after $P_{0,1,\ldots,k}D_{0,1,\ldots,{|D|-2}}$ is grown, then $P_{k,k+1,...}$ turns right (\resp, left) from $P_{k}D_{0,1,\ldots,{|D|-2}}$, and thus $R$ from the statement of  Lemma~\ref{lem:onepath} is a suffix of $P$. But this implies that $P$  places a lower visible glue than $\glue{\ell}{\ell+1}$ (on $\ell$)  and/or $\namedGlue{P_k}{D_0}$ (on the visibility ray $l_k$), which  contradicts the visibility of  $\glue{\ell}{\ell+1}$ and/or $\namedGlue{P_k}{P_{k+1}} = \namedGlue{P_k}{D_0}$ relative to~$P$.  

  Therefore, $P$ conflicts with $P_{0,1,\ldots,k}D_{0,1,\ldots,{|D|-2}}$ and thus $P$ cannot grow to be \hs from the assembly $\sigma \cup \asm {P_{0,1,\ldots,k}D_{0,1,\ldots,{|D|-2}}}$.
\end{proof}

\section{Blocking all paths}\label{sec:manypaths}

We restate our main theorem here: 
\begin{reptheorem}{thm:main}
\thmMain
\end{reptheorem}

This result is  an immediate corollary of Theorem~\ref{thm:shapes} below.  
 Intuitively, Theorem~\ref{thm:shapes} states that there is no tile set that, at temperature 1,  produces (or simulates) the ``shapes'' of all $\calT_{N}$ systems,\footnote{The class of ``flipped-L'' tile assembly systems $\{ \calT_{N} \mid N \in \mathbb{Z}^{+} \}$ were defined earlier in Section~\ref{sec:T}.}
 even if the simulator is allowed to use spatial rescaling. Thus Definition~\ref{def:equiv-shape} is violated  which immediately implies (via Observation~\ref{obs:prodshapes}) that there is no tile set that, at temperature 1, simulates the productions of all $\calT_{N}$ systems (thus contradicting Definition~\ref{def:equiv-prod}, ``equivalent productions''), which in turn contradicts Definition~\ref{def:iu-specific-temp} (``intrinsicially universal'', at temperature 1), giving Theorem~\ref{thm:main}.  

\begin{theorem}
\label{thm:shapes}
There is no tileset $U$, scale factor $m \in \mathbb{Z}^{+}$, seed $\sigma$ and $m$-block supertile representation function $R_m$ such that for all $N \geq 10 |U|$,  
$\dom{R^{*}_m(\sigma)} = \dom{\sigma_{N}}$ and $\{\dom{R^{*}_m(\alpha)} \mid \alpha\in\termasm{\mathcal U}\} = \{\dom{\beta} \mid \beta\in\termasm{{\mathcal T}_{N}}\}$ where  ${\mathcal U} = (U,\sigma,1)$ and $\mathcal{T}_{N} = (T_{N},\sigma_{N},1)$.
\end{theorem}

\begin{proof}
\newcommand\n{10}
% !TEX root = t1notiu.tex

Assume, for the sake of contradiction, that there is a tileset $U$ such that for $N=10|U|$, there is an integer $m$, a seed assembly $\sigma\in\asm{\cal U}$, and an $m$-block representation function $R_{m}$ such that the terminal assemblies of $\calU = (U,\sigma,1)$ map cleanly to the terminal assemblies of $\calT_{10|U|}$ under $R_m$, where $\calT_{10|U|}$ is the flipped-L tile assembly system defined in Definition~\ref{def:tn}.

We will show that~$\calU$ also produces terminal assemblies mapping to non-terminal or non-producible assemblies of $\calT_{10|U|}$ under $R_m$. More specifically, we will show that either some of the assemblies of~$\calU$ map cleanly to \emph{non-producible} assemblies of $\calT_{10|U|}$ under $R_m$, or else we will construct one \emph{producible} assembly $\alpha\in\prodasm{\calU}$ conflicting with all \hs paths of~$\calU$. This will then conclude the proof  since $\alpha$ grows into a \emph{terminal} assembly, i.e.\ $\alpha \rightarrow^{\mathcal{U}} \alpha' $ where $\alpha'\in \termasm{\calU}$, that does not map cleanly to a \emph{terminal} assembly of $\termasm{\calT_{10|U|}}$ under $R_m$ (since all tiles of $\alpha'$ are below the horizontal line at height $h$).

\paragraph{Blocking \hs paths individually} For the remainder of the proof, let  $\ell$ be a vertical (glue) line at x-coordinate $|U|(3m+1)+m+1.5$ (in other words, at distance $\geq |U|(3m+1)+1$ to the right of the rightmost tile of~$\sigma$), as defined by Lemma~\ref{lem:visibility-setup}.
We  apply Lemma~\ref{lem:onepath} on each \hs path $P$, individually\footnote{By ``individually'' we mean that we are currently merely looking at the case where we grow each path separately: of course it may be the case that not all of these paths can be simultaneously grown as they may conflict with each other---the main point of this proof is to handled this.}. For each such~$P$,  Lemma~\ref{lem:onepath} has one of two
conclusions, numbered Conclusion~\ref{onepath:conclusion:breakscale} and Conclusion~\ref{onepath:conclusion:allbroken}. If we get Conclusion~\ref{onepath:conclusion:breakscale} for any of the \hs paths, we can conclude the proof immediately, because that conclusion shows that it is possible to grow a path from $\sigma$ that places tiles outside of the simulation zone of $\calU$, contradicting that $\calU$ simulates $\calT_{10|U|}$,
and hence assemblies of $\calU$ do not simulate the shape of $\calT_{10|U|}$ and we are done with the proof of Theorem~\ref{thm:shapes}.

Therefore, in the rest of this proof, we assume that for \emph{all \hs paths} of~$\calU$ we get Conclusion~\ref{onepath:conclusion:allbroken}
of Lemma~\ref{lem:onepath}. That conclusion gives, for each \hs path $P$, a \nhs enclosing branch $D$ for~$P$ at some integer $k_P$.

If it were the case that the entire set of these enclosing branches could grow together in the same assembly, we would immediately be done: indeed, the union of the seed with all of these enclosing branches (and their prefixes from $P$) would be an assembly conflicting with all \hs paths of $\calU$ (implying in particular that this union does not contain any \hs path).

The rest of the proof deals with the situation where this is not the case, i.e.\ at least one (and possibly very many) enclosing branches $D$ from Lemma~\ref{lem:onepath} conflict with other paths or with other enclosing branches, and thus not all enclosing branches $D$ can grow completely together in the same assembly.

\paragraph{Path order}
We will build an assembly that does not reach height $h$ and that blocks all of the paths from the set of \hs paths  $\hsU$ of $\calU$. Recall that the set of \hs paths of $\calU$ is finite.  In order to block them all, we will tackle \hs paths in a specific order, called the ``path order,'' defined as follows. Let $\prec$ be {\em the path order} relation on the set  $\hsU$ of \hs paths of $\calU$ where for $P,Q\in \hsU$ with $P\neq Q$ we say that $P\prec Q$ if and only if at least one of (A) or (B) holds:
\begin{enumerate}[label=(\alph*),leftmargin=*,align=left]
\item[(A)] the visible glue of $P$ on $\ell$ is strictly higher than the visible glue of $Q$ on $\ell$, or
\item[(B)] the visible glues $\glueP{p}{p+1}$ and $\glueQ{q}{q+1}$ of $P$ and $Q$ on $\ell$ are at the same position\footnote{I.e.\ $\midpoint{{P_p}}{{P_{p+1}}} = \midpoint{{Q_q}}{{Q_{q+1}}} $.} and one of the following holds:  
  \begin{itemize}
  \item  $ \glueP{p}{p+1} \in V_P^+$ and $ \glueQ{q}{q+1} \in V_Q^-$, or 
  \item  $\glueP{p}{p+1} \in V_P^+$, $\glueQ{q}{q+1} \in V_Q^+$, and $Q_{q,q+1,\ldots,|Q|-1}$ is the right-priority path of $P_{p,p+1,\ldots,|P|-1}$ and $Q_{q,q+1,\ldots,|Q|-1}$, or 
  \item   $\glueP{p}{p+1} \in  V_P^-$, $\glueQ{q}{q+1} \in V_Q^-$, and $Q_{q,q+1,\ldots,|Q|-1}$ is the left-priority path $P_{p,p+1,\ldots,|P|-1}$ and $Q_{q,q+1,\ldots,|P|-1}$, or 
  \item  Else, notice that $P$ and $Q$ share their suffix from their visible glue on $\ell$ onwards until their last tile at height $h$ (because none of these suffixes is the right-priority or left-priority one).
    Then $P$ is the right priority path of $P$ and $Q$ if $P$ and $Q$ share two consecutive tiles $P_{a}P_{a+1}= Q_{b}Q_{b+1}$ before\footnote{By ``before'' we mean with respect to the order of tiles along the path $Q$, and along the path $P$.} disagreeing (note that $P\neq Q$), and if they do not share such a pair then $P$ is the lexicographically first path of $P$ and $Q$ if we describe both using some canonical encoding of $P$ and $Q$ as two binary strings.
  \end{itemize}
\end{enumerate}
Here is an intuitive description of the path order: we first consider paths by the height of their visible glue on $\ell$ (highest first), and then if  both visible glues on $\ell$ are in the same direction, we first consider the most right-priority of $P$ and $Q$ after they cross $\ell$ if these glues are in $V_P^+$ and $V_Q^+$, or the most left-priority if these glues are in $V_P^-$ and $V_Q^-$, and if the glues are at the same position with different +/- orientations, the one with a ``+'' visible glue on $\ell$ comes first. Finally if $P$ and $Q$ happen to agree (are equal) on their suffix from their visible glue on $\ell$  onwards, then we (arbitrarily) choose the right priority path (note that in this latter case all of the differences between $P$ and $Q$ must be before their respective visible glues on $\ell$).

Note that the relation $\prec$ is a total order on the set  $\hsU$ of \hs paths, since the last case of the definition of $\prec$ covers all remaining cases using right-priority, and right-priority is itself a total order. Also, recall that the set of \hs paths is a finite number (see Section~\ref{sec:blocking-h-successful}), and let $H$ be that number.  

Thus let  $P^0\prec P^1 \prec P^2\prec \ldots \prec P^{H-1}$ be the list of all \hs paths according to path  order (so that no path has a higher  visible glue on $\ell$ than $P^0$).

\paragraph{Enclosing branch $D^n$.}
For each path $P^{n}$, applying Lemma~\ref{lem:onepath} gives an index~$k_n$ and an ``enclosing branch $D^n$ for $P^n$ at $k_n$'' such that $P^n_{0,1,\ldots,k_n}D^n_{0,1,\ldots, |D|-2}$ conflicts with $P^n$ (by Theorem~\ref{thm:onepath}).
Let  $$E^n \defeq P_{0,1,\ldots,k_n}D^n_{0,1,\ldots, |D|-2}$$

\paragraph{The ``path blocking'' assembly $\alpha_n$.}
Let the notation $\asmprefix{\alpha}{P}$ denote the path that is the longest assemblable prefix of $P$ that can be grown from the assembly $\alpha$.\footnote{Observe that if $\alpha$ is producible by some tile assembly system then for all paths $P$ it is (trivially) the case that  $\alpha \cup \asm{\asmprefix{\alpha}{P}}$ is an assembly producible by that same tile assembly system.}

We define an assembly $\alpha_n$ which has a special form (composed of $\sigma$ and assemblable prefixes of~$E^k$ grown in path, i.e. $k$, order), to be used in our induction hypothesis:
$$\alpha_n = \sigma \cup \left( \bigcup_{k=0}^n \, \asm{F^{k}}\right)$$ where $F^0 = \asmprefix{\sigma}{E^0} = E^0_{0,1,\ldots,|E^0|-1}$ and   for all $k \geq 1 $, $ F^{k} = \asmprefix{\alpha_{k-1}}{E^k}$.

\paragraph{Claim:  for all $n\geq 0$, $\alpha_n \in \prodasm{\calU}$.} 
To see this claim note that: 
\begin{itemize}
\item First, $\alpha_0$ is producible: indeed, $F^0 = E^0 = P^0_{0,1,\ldots,k_0}D^0_{0,1,\ldots, |D|-2}$ is a producible path of $\mathcal U$, by Lemma~\ref{lem:onepath}. Therefore, $\alpha_0 = \sigma\cup \asm{F^0}$ is producible.
\item Then, assuming $\alpha_n$ is producible, i.e. $\alpha_n\in\prodasm{U}$, remember that $\asmprefix{\alpha_n}{E^{n+1}}$ is the maximal prefix of $E^{n+1}$ that can grow from $\alpha_n$. Therefore, $\alpha_n\rightarrow^\calU \left( \alpha_n\cup \asm{\asmprefix{\alpha_n}{F^{n+1}}} \right)= \alpha_{n+1}$, and therefore $\alpha_{n+1}\in\prodasm{\mathcal U}$.
\end{itemize}
Hence $\alpha_n \in \prodasm{\calU}$ as claimed.

\vspace{\baselineskip}
To conclude the proof we will consider the assembly\footnote{Recall that $H$ is the number of $h$-successful paths of $\calU$.} $\alpha_{H-1} \in \prodasm{\calU}$ which we claim has the (as yet unproven)   property that all producible \hs paths conflict with it. Then allowing tiles to attach to $\alpha_{H-1}$ will eventually yield\footnote{Growth can only happen within the finite area simulation zone below height $h$ so must eventually stop.}  a terminal assembly $\alpha \in \termasm{\calU}$ with no tiles at height $h$ and thus no tiles above height $h$, which contradicts that $\calU$ simulates (the shape of) $\calT_{10|U|}$.  We will use induction to show that 
all producible \hs paths conflict with  $\alpha_{H-1}$.

\paragraph{Induction hypothesis:}
All of the paths $P^{0},P^{1},\ldots,P^{n}$  conflict with  $\alpha_n$.
\vspace{\baselineskip}

Some intuition and implications of our induction hypothesis: The induction hypothesis implies that for $\forall k \leq n, \asmprefix{\alpha_n}{P^k}$  is not  \hs. To see this note that since for all $ k$ the last tile of $P^k$ is its only tile at height $h$, and the induction hypothesis implies that $P^k$  can not grow from $\alpha_n$ to be \hs. 
Also, no tile of $\alpha_n$ reaches height $h$ (because $E^k$ is constructed via Lemma~\ref{lem:onepath}) which implies that none of the~$F^k$ (of which  $\alpha_n$ is composed)   are \hs.
Finally,  since there are a finite number $H$ of \hs paths  the induction exhausts those $H$ paths in $H$ steps, and thus yields an assembly $\alpha_{H-1}$ which is a finite union of finite (path) assemblies, and thus $\alpha_{H-1}$ is a finite producible assembly that blocks all \hs paths.

\paragraph{Initial step of induction ($P^{0}$ and $\alpha_{0}$).} At the initial step of the induction, we apply Lemma~\ref{lem:onepath} to~$P^0$, to obtain an enclosing branch $P^0_{0,1,\ldots,k_0}D^0_{0,1,\ldots,|D^0|-1}$. 
This proves our induction hypothesis for the initial step:  
  by Theorem~\ref{thm:onepath}, $P^0$ conflicts with  $\alpha_0$ (i.e.\ $P^0$ cannot grow to be \hs from $\alpha_0$), 
 and we have already defined $\alpha_0=\sigma\cup \asm{F^0}$ where $F^0 = \asmprefix{\sigma}{E^0} = E^0_{0,1,\ldots,|E^0|-2}  = P^0_{0,1,\ldots,k_0}D^0_{0,1,\ldots,|D^0|-2}$.

\paragraph{Inductive step ($P^{n+1}$ and $\alpha_{n+1}$).} The remainder of the proof is concerned with the inductive step. For any $n\geq 0$ suppose the induction hypothesis holds,\footnote{Recall that $\alpha_n$ contains only the seed $\sigma$ and assembled paths $F^{0}, F^{1}, \ldots ,F^{n}$ (i.e.\ $\alpha_n = \sigma \cup \left( \bigcup_{k=0}^n \, \asm{F^{k}}\right)$) that are respective prefixes  of $E^{0}, E^{1}, \ldots ,E^{n}$, 
none of which are \hs.} i.e.\ 
 all of  first $n+1$ paths $P^{0},P^{1},\ldots,P^{n}$ conflict with  $\alpha_n$.

We recall that $F^{n+1} = \asmprefix{\alpha_n}{E^{n+1}}$ is  the maximal prefix of $E^{n+1}$ that can grow from $\alpha_{n}$, and that $\alpha_{n+1} = \alpha_n\cup \asm{F^{n+1}}$.
If $P^{n+1}$ conflicts with $\alpha_n$,  then we are immediately done with the induction step for $\alpha_{n+1}$, because this proves that $P^{n+1}$ cannot grow from $\alpha_{n+1}$. Hence from now we will assume that $P^{n+1}$ does not conflict with $\alpha_n$.  

If $P^{n+1}$ conflicts with $F^{n+1}$, then we are immediately done with the induction step for $\alpha_{n+1}$, because this proves that $P^{n+1}$ cannot grow from $\alpha_{n+1}$. 
Otherwise $P^{n+1}$ does not conflict with $F^{n+1}$. This implies that $F^{n+1}$ is a {\em strict} prefix of $E^{n+1}$ (otherwise we would contradict  Theorem~\ref{thm:onepath})  and therefore 
$E^{n+1}$ (and in particular $D^{n+1}$) conflicts with~$F^k$ for some $k\in \{0,1,\ldots,n \}$.
We will reason about this $F^k$. 

\newcommand{\mainthmcaseone}{\pn and $F^k$ share the position of their visible glue on $\ell$, one of these glues is a $V^+$ glue, the other one is a $V^-$ glue}
\newcommand{\mainthmcasetwo}{\pn and $F^k$ do not share the position of their visible glue on $\ell$ (which includes the case where $F^k$ does not reach $\ell$)}
\newcommand{\mainthmcasethree}{\pn and $F^k$ share the position of their visible glue on $\ell$, and either both are $V^+$ glues, or both are $V^-$ glues}
We next split the inductive step into three cases\footnote{For the sake of proof simplicity, we present them in the order in which we handle these cases.} each of which will be concluded independently:
\begin{enumerate}[leftmargin=*,align=left]
\item[\em (Case 1)] \mainthmcaseone.
\item[\em (Case 2)] \mainthmcasetwo.
\item[\em (Case 3)] \mainthmcasethree.
\end{enumerate}

In all three cases, let $F^k$'s visible glue on $\ell$ be denoted $\namedGlue{F^k_f}{F^k_{f+1}}$ and let \pn's  visible glue on   on $\ell$ be denoted \pnviz.

\paragraph{Case 1: \mainthmcaseone.}

Exactly one of $\namedGlue{F^k_f}{F^k_{f+1}}$ and \pnviz is in $V^+$ and the other is in $V^-$. (See Figure~\ref{fig:ex1} for an example.)
At the beginning of the inductive step, we assumed that \pn does not conflict with $\alpha_n$, hence  $\pn$ and $F^k$ agree on all points where they intersect. Moreover, since $\pn$ and $F^k$ share their visible glue on~$\ell$ (i.e.\ $\midpnviz = \midpoint{F^k_f}{F^k_{f+1}}$, i.e. their visible glues on $\ell$ are at the same position) we know they agree on at least two tiles each, specifically $F^k_f = \pn_{p+1}$ and $F^k_{f+1}=\pn_{p}$. 
Now let $b$ be the largest integer such that there is an integer $a\leq f$ where ${F^{k}_a} = P^{n+1}_b$.
Since $\midpnviz = \midpoint{F^k_f}{F^k_{f+1}}$ we know that $p+1\leq b$, also $b \leq |\pn|-2$  since $\pn$ is \hs and \fk is not.
Consider the sequence $Q = F^{k}_{0,1,\ldots,a}  P^{n+1}_{b+1,b+2,\ldots , |P|-1}$. First note that the positions of $Q$ form a connected sequence of positions in $\Z^2$:  this follows from the fact that the positions of $F^{k}_{0,1,\ldots,a}$ are connected, the positions of  $ P^{n+1}_{b+1,b+2,\ldots,|P|-1}$ are connected, and that ${F^{k}_a}  = P^{n+1}_b$. Also, we claim that $Q$ is simple:  to see this,  note that (i) $F^{k}_{0,1,\ldots,a}$ is simple, (ii) $P^{n+1}_{b+1,b+2,\ldots, |P|-1}$ is simple, and finally that (iiix) $P^{n+1}_{b+1,b+2,\ldots, |P|-1}$ does not intersect $F^{k}_{0,1,\ldots,a}$ (by definition of $b$). Since $Q$ has a connected simple set of positions, $Q$ is a path.  Furthermore it is the case that $Q \in \prodpaths{\calU} $, which follows immediately from the following facts:  $Q$ is a path, $ F^{k}_{0,1,\ldots,a} \in \prodpaths{\calU}$, $P^{n+1} \in \prodpaths{\calU}$ and ${F^{k}_a}  = P^{n+1}_b$.

 \begin{figure}
   \begin{center}
     \includegraphics{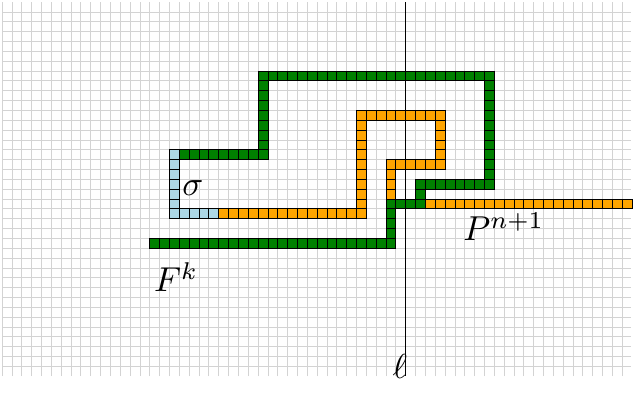}\hspace{4ex}
     \includegraphics{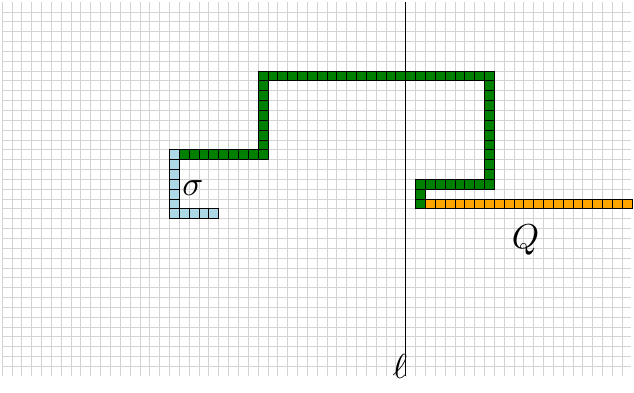}
   \end{center}
   \caption{{\bf Left:} An example of Case 1, where $F^k$ places a $V^-$ glue on $\ell$. {\bf Right:} The path $Q = F^{k}_{0,1,\ldots,a}  P^{n+1}_{b+1,b+2,\ldots} $ which has the property that it comes earlier in path order than  $P^{n+1}$, i.e. $Q \prec P^{n+1}$, since the visible glue of $Q$ on $\ell$ is higher than the visible glue of $P^{n+1}$ on $\ell$.}
   \label{fig:ex1}
 \end{figure}
 
Next we claim that $Q \prec P^{n+1}$. 
First, we know that, since $b \geq p+1$ all of the glues that $\pn_{b+1,b+2,\ldots}$ places on $\ell$ are at height strictly higher than the height of $\pnviz$ on~$\ell$ which is $\pn$'s visible (i.e.\ lowest) glue on~$\ell$.
Second, since $a \leq f$ we know that all of the glues that $F^k_{0,1,\ldots,f}$ places on~$\ell$ are at height strictly higher than the height of $\namedGlue{F^k_f}{F^k_{f+1}}$ on~$\ell$ which is $F^k$'s visible glue on~$\ell$ which is at the same height as $\pn$'s visible glue on~$\ell$. 
Since $Q = F^{k}_{0,1,\ldots,a}  P^{n+1}_{b+1,b+2,\ldots}$, then~$Q$'s visible glue on~$\ell$
is strictly higher than the visible glue of $P^{n+1}$ on $\ell$. 
Thus $Q \prec P^{n+1}$.

Since \pn is \hs, and since \pn and $Q$ share a nonempty suffix $\pn_{b+1,b+2,\ldots, |\pn|-1}$,  this implies that $Q$ is also \hs. Moreover, no strict prefix of $Q$ is \hs, because $F^k$ is not \hs and by the definition of \emph{\hs} no strict prefix of \pn is \hs. But since $Q \prec P^{n+1}$, this means that $Q$ 
satisfies the induction hypothesis,  
meaning that $Q$ conflicts with $\alpha_n$. Since $\asm{F^{k}_{0,1,\ldots,a}}$ is a subassembly of $\alpha_n$ then the prefix $F^{k}_{0,1,\ldots,a}$ of $Q$  does not conflict with \an, which in turn  implies that the suffix $\pn_{b+1,b+2,\ldots, |\pn|-1}$ of $Q$ conflicts with $\alpha_n$, which  implies that $\pn$ conflicts with~$\alpha_n$, 
satisfying the induction hypothesis.\footnote{Recall that we have already defined  $F^{n+1} = \asmprefix{\an}{E^{n+1}}$ and $\alpha_{n+1} = \an \cup \asm{F^{n+1} }$.}

\paragraph{Case 2: \mainthmcasetwo.}
Moreover, the glue placed by $\pn$ on~$\ell$ is a $V^+_{\pn}$ glue (respectively a $V^-_{\pn}$ glue).
We assumed that \pn does not conflict with~$\alpha_n$,  
hence \pn does not conflict with $F^k$. 
We first show that $\pn_{p+1,p+2,\ldots,|\pn|-1}$ intersects and agrees with $F^k$, and then use an argument similar to Case~1 above:
\begin{itemize}
\item
  Assume, for the sake of contradiction, that $\pn_{p+1,p+2,\ldots,|\pn|-1}$ does not intersect $F^k$. Let $a\geq p$ be the smallest\footnote{There is at least one such integer since we know that $E^{n+1}$  conflicts with $ F^k$, and we know that this conflict happens after (in $E^{n+1}$ order) the visible glue ($\pnviz$) of $\pn$ and $E^{n+1}$.} integer such that $\pos{E^{n+1}_a} = \pos{F^k_b}$ for some $b$.
  We are going to define a closed connected component in which \pn starts to grow. First note that $F^k$ is connected, connected to $\sigma$, and $E^{n+1}$ is connected to $\sigma$. Therefore, $\sigma\cup \asm{F^k} \cup\asm{ \pn_{0,1,\ldots,p}}$ contains at least one  path from $F^k_b$ to $\pn_p$ (note that $\pn_{0,1,\ldots,p}$ is a prefix of $E^{n+1}$). Let $\mathcal P$ be any shortest such path.

Let then $c$ be the closed curve defined by the concatenation of ${c}^k = \frak{E}_{\mathcal P}$ and ${c}^{n+1} = \frak{E}_{E^{n+1}_{p,p+1,\ldots,a}}$. Curve $c$ is simple because ${c}^{n+1}$ and $c^k$ only intersect at their endpoints because $a$ was chosen to be the smallest integer ($\geq p$) such that $\pos{E^{n+1}_a}=\pos{F^k_b}$ and because $\mathcal P$ is a shortest path.

Therefore, by the Jordan Curve Theorem, $c$ encloses a single bounded connected component $\mathcal{C}$ of $\R^2$. (This connected component is shown in gray in the example in Figure~\ref{fig:case2}.)

\begin{figure}[h]
   \begin{center}
     \includegraphics{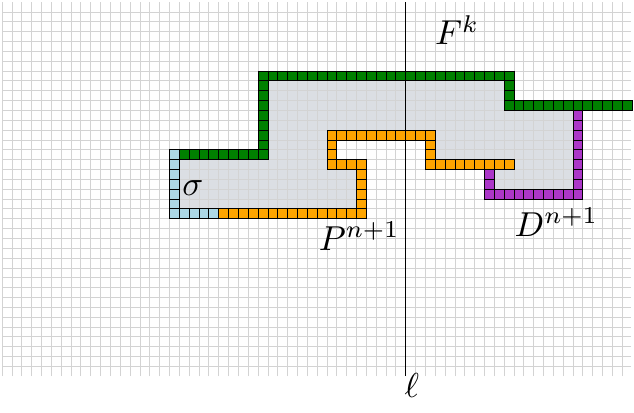}
   \end{center}
   \caption{An example of Case 2, where $P^{n+1}$ and $F^k$ both place a $V^+$ glue on $\ell$, at different heights.}
   \label{fig:case2}
\end{figure}

Now,  $\midpnviz$, the position of the visible glue $\pnviz$ of $\pn$ on $\ell$, is on curve~$c$. Moreover, since no other point of $c$ intersects $\ell$ at the height of, or below $\midpnviz$, then $\midpnviz$ is the unique lowest intersection of $c$ and $\ell$.

Then, since $\pnviz$ is in $V^+_{\pn}$ (respectively in $V^-_{\pn}$), the left-hand side (respectively right-hand side) of $c$ is inside $\mathcal{C}$.  
Therefore, since $\pn$ is \hs and places tiles (with positions) on $c$, $\pn_{p,p+1,\ldots, |P|-1}$ needs to turn from $c$ (because all points of $c$ are below height $h$).
However, by Lemma~\ref{lem:onepath}, $\pn$ cannot turn right (\resp, left) from  $D^{n+1}$, hence from $E^{n+1}$; if it did  \pn would hide at least one of its own visible glues, which is impossible. Therefore, $\pn_{p,p+1,\ldots ,|\pn|-1}$ must turn from, and thus intersect, other parts of $c$, i.e. $\sigma$ or $F^k$,  which is a contradiction. 
Hence $\pn_{p,p+1,\ldots ,|\pn|-1}$ intersects $F^k$. 

\item We have shown that $\pn_{p+1,p+2,\ldots,|\pn|-1}$ intersects (and agrees with) $F^k$ at least once. In fact all such intersections are agreements because $\pn$ does not conflict with $F^k$. Let $a\geq p+1$ be the largest integer such that $\pn_a = F^k_b$ for some integer $b$. We claim that $Q = F^k_{0,1,\ldots,b}\pn_{a+1,a+2,\ldots}$ is a path: indeed, $\pn_{a+1,a+2,\ldots}$ does not intersect $F^k_{0,1,\ldots,b}$ by the definition of~$a$, and $Q$ is connected. Moreover, $Q \in \prodpaths{\calU}$. Furthermore, $Q\prec \pn$, because since $a>p$ and \pn is simple, and since the visible glue of $F^k$ on $\ell$ is not shared with that of \pn, all glues of $Q$ on $\ell$ are strictly higher than $\pnviz$. Therefore, by the induction hypothesis, $Q$ conflicts with $\alpha_{n}$, and hence $\alpha_{n+1}$, which means that $\pn_{a+1,a+2,\ldots}$ conflicts with $\alpha_{n+1}$ (since $F^k_{0,1,\ldots,b}$ does not conflict with $\alpha_{n+1}$).

\end{itemize}

\paragraph{Case 3:  \mainthmcasethree.}
In this case, because $\pk\prec \pn$, and $\pn$ and $F^k$, and hence $\pk$, share their visible glue at the same height on $\ell$, we know by the definition of $\prec$ that either:
\begin{itemize}
\item $\pn_{p,p+1,\ldots}$ is more right-priority (\resp, left-priority) than $\pk_{f,f+1,\ldots}$ if $\gglue{\pn}{p}{p+1}$, $\gglue{\pk}{f}{f+1}$  are both  $V^+$ glues (\resp,  $V^-$ glues), where $\gglue{\pn}{p}{p+1}$ is the visible glue of $\pn$ on $\ell$, and $\gglue{\pk}{f}{f+1}$ is the visible glue of $\pk$ on $\ell$.
\item $\pn_{p,p+1,\ldots} = \pk_{f,f+1,\ldots}$
\end{itemize}

However, in the second case, since $\alpha_n$ conflicts with $\pk$ (by the induction hypothesis, since $k\leq n$), and $F^k$ places the visible glue of \pk on $\ell$, then \an conflicts with  $\pk_{f,f+1,\ldots}$, hence $\an$ also conflicts with $\pn_{p,p+1,\ldots}$, and we are done with Case 3 by simply letting $\alpha_{n+1}=\alpha_n \cup \asm{\asmprefix{\an}{E^{n+1}}}$.

To conclude this proof, we will therefore handle the first case, i.e. the case where $\pn_{p,p+1,\ldots}$ is more right-priority (\resp, left-priority) than $\pk_{f,f+1,\ldots}$.

We assumed that \pn does not conflict with $\alpha_n$, hence in particular $\pn_{p,p+1,\ldots}$ does not conflict with $\gk_{f,f+1,\ldots}$, and  does not conflict with $\pk_{f,f+1,\ldots}$.
Let $q\geq p$ and $g\geq f$ be the smallest integers such that $\pn_q \neq \gk_g$.
An example of this situation is shown in Figure~\ref{fig:case3}. 

\begin{figure}
   \begin{center}
     \includegraphics{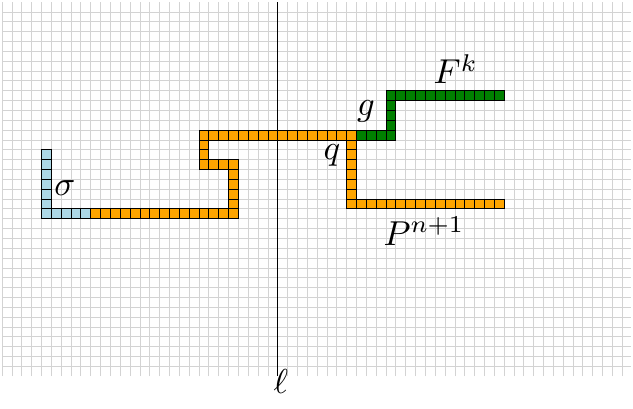}
   \end{center}
   \caption{An example of Case 3, where $P^{n+1}$ and $F^k$ (and by implication $P^k$) both place a $V^+$ glue on $\ell$, at the same height.}
   \label{fig:case3}
\end{figure}

The argument follows along the same lines as Case~2 (building a closed connected component in which \pn starts to grow), but requires a new technique to identify the inside and outside of that connected component.

\begin{itemize}

\item Assume, for the sake of contradiction, that $\pn_{q,q+1,\ldots}$ does not intersect \gk.

We  now describe a closed connected component inside which a suffix of \pn starts to grow.  We first introduce a new variant of embedding of paths into $\R^2$, which we call the \emph{nano-embedding of a path $P$}, denoted $\frak N_P$. This is illustrated in Figure~\ref{fig:nanoembed} and defined as follows.
  For a path $P$ consider its canonical embedding $\frak E_P$. Then, 
  define  $\frak N_P$ to be the curve in $\R^2$ where all of the points of $\frak N_P$ are at $L_\infty$ distance exactly $0.25$ from their closest point on~$\frak E_P$, and are positioned on the right (\resp, left) hand side of  $\frak E_P$ as we walk  along  $\frak E_P$ from $\frak E_P(0)$ to  $\frak E_P(1)$. 
  For tiles on $P$ with input side being their west side, we show in Figure~\ref{fig:nanoembed}(top) all three cases of nano-embeddings. 
  Other cases where the input side is north, east or south are rotations of these three cases. Special cases for start and end tiles of a path are illustrated in Figure~\ref{fig:nanoembed}. 
  
  \begin{figure}[p]
    \begin{center}
      \begin{tikzpicture}[scale=2]
        \draw(0,0)rectangle(1,1);
        \draw[blue,very thick,->](-0.0,0.5)--(0.5,0.5)--(0.5,1.);
        \draw(-0.0,0.5)node[anchor=east]{$\frak E_P$};
        \draw[red,very thick, ->](-0.0,0.25)--(0.75,0.25)--(0.75,1.);
        \draw(-0.0,0.25)node[anchor=east]{$\frak N_P$};
        %\draw(-0.25,0.78)node[]{input};
        %\draw(0.15,1.14)node[]{output};
        \draw(0.5,-0.4)node[anchor=south]{input west, output north};
      \end{tikzpicture}\hspace{8ex}
      \begin{tikzpicture}[scale=2]
    %    \draw[fill=none,draw=none](-0.25,-0.25) rectangle(1.25,1.25);
        \draw(0,0)rectangle(1,1);
        \draw[blue,very thick,->](-0.0,0.5)--(1.,0.5);
        \draw(-0.0,0.5)node[anchor=east]{$\frak E_P$};
        \draw[red,very thick, ->](-0.0,0.25)--(1.,0.25);
        \draw(-0.0,0.25)node[anchor=east]{$\frak N_P$};
       % \draw(-0.25,0.78)node[]{input};
        %\draw(1.33,0.78)node[]{output};
        \draw(0.5,-0.4)node[anchor=south]{input west, output east};
      \end{tikzpicture}\hspace{8ex}
      \begin{tikzpicture}[scale=2]
     %   \draw[fill=none,draw=none](-0.25,-0.25) rectangle(1.25,1.25);
        \draw(0,0)rectangle(1,1);
        \draw[blue,very thick,->](-0.0,0.5)--(0.5,0.5)--(0.5,-0.);
        \draw(-0.0,0.5)node[anchor=east]{$\frak E_P$};
        \draw[red,very  thick,->](-0.0,0.25)--(0.25,0.25)--(0.25,-0.);
        \draw(-0.0,0.25)node[anchor=east]{$\frak N_P$};
       % \draw(-0.25,0.78)node[]{input};
       % \draw(0.9,-0.12)node[]{output};
        \draw(0.5,-0.4)node[anchor=south]{input west, output south};
      \end{tikzpicture}
    \end{center}
    \begin{center}
      \begin{tikzpicture}[scale=1.2]
        \draw[fill=none,draw=none](-0.25,-0.25) rectangle(1.25,1.25);
   
        \draw[fill=red,opacity=0.5](-1,0)rectangle(0,1);   
        \draw[fill=red,opacity=0.5](0,0)rectangle(1,1);
        \draw[fill=red,opacity=0.5](1,0)rectangle(2,1);
        \draw[fill=lightgray](2,0)rectangle(3,1);
        \draw[fill=lightgray](3,0)rectangle(4,1);
        \draw[fill=lightgray](4,0)rectangle(5,1);

        \draw[fill=red](1,-1)rectangle(2,0);
        \draw[fill=red](1,-2)rectangle(2,-1);
        \draw[fill=red](2,-2)rectangle(3,-1);
        \draw[fill=red](3,-2)rectangle(4,-1);
        \draw[fill=red](4,-2)rectangle(5,-1);

        \draw[fill=red](4,-1)rectangle(5,0);
        
        \draw(-0.5,-0.2)node[]{start};

        \draw(-0.5,0.65)node[]{\footnotesize$P_0$};
        \draw(-0.5,0.35)node[]{\footnotesize$Q_0$};

        \draw(0.5,0.65)node[]{\footnotesize$P_1$};
        \draw(0.5,0.35)node[]{\footnotesize$Q_1$};

        \draw(1.5,0.65)node[]{\footnotesize$P_2$};
        \draw(1.5,0.35)node[]{\footnotesize$Q_2$};

    %    \draw(0.5,0.5)node[]{\footnotesize$P_1 , Q_1$};
     %   \draw(1.5,0.5)node[]{\footnotesize$P_2 Q_2$};
        \draw(2.5,0.5)node[]{\footnotesize$P_3$};
        \draw(3.5,0.5)node[]{\footnotesize$P_4$};
        \draw(4.5,0.5)node[]{\footnotesize$P_5$};
        \draw(1.5,-0.5)node[]{\footnotesize$Q_3$};
        \draw(1.5,-1.5)node[]{\footnotesize$Q_4$};
        \draw(2.5,-1.5)node[]{\footnotesize$Q_5$};
        \draw(3.5,-1.5)node[]{\footnotesize$Q_6$};
        \draw(4.5,-1.5)node[]{\footnotesize$Q_7$};
        \draw(4.5,-0.5)node[]{\footnotesize$Q_8$};

      \end{tikzpicture} \hspace{7ex}
      \begin{tikzpicture}[scale=1.2]
        \draw[fill=none,draw=none](-0.25,-0.25) rectangle(1.25,1.25);
   
        \draw(-1,0)rectangle(0,1);   
        \draw(0,0)rectangle(1,1);
        \draw(1,0)rectangle(2,1);
        \draw(2,0)rectangle(3,1);
        \draw(3,0)rectangle(4,1);
        \draw(4,0)rectangle(5,1);

        \draw(1,-1)rectangle(2,0);
        \draw(1,-2)rectangle(2,-1);
        \draw(2,-2)rectangle(3,-1);
        \draw(3,-2)rectangle(4,-1);
        \draw(4,-2)rectangle(5,-1);

        \draw(4,-1)rectangle(5,0);
 
        \draw[blue,very thick,<-](-0.5,0.5)--(4.5,0.5);
        
        \draw[red,very thick, ->](-0.5,0.25)--(1.25,0.25)--(1.25,-1.75)--(4.75,-1.75)--(4.75,0);
        \draw[magenta,very  thick, ->](-0.5,0.5)--(-0.5,0.25);
        \draw[green,very  thick, ->](4.75,0)--(4.5,0.5);

        \draw(-0.76,0.4)node[]{$ s_m$};
        \draw(4.8,0.2)node[]{$ s_g$};
        \draw(1.5,0.72)node[]{$\frak E_P^\leftarrow$};
        \draw(1.58,-0.5)node[]{$\frak N_Q$};
      \end{tikzpicture}\hspace{8ex}
    \end{center}
    \caption{Nano-embedding and canonical embedding of a path in $\R^2$. Top:  Three tiles on some path $P$ that have their input side as their west side, and their output sides as north, east and south respectively, as indicated. 
    In each of the three cases, the 
    canonical embedding $\frak{E}_P$ of $P$ is shown in blue, and the nano-embedding $\frak{N}_P$ is shown in red. Each point in $\frak{N}_P$ is $L_\infty$ distance exactly 0.25 from its closest point in $\frak{E}_P$. 
    Rotating these diagrams by $90^\circ$, $180^\circ$ and $270^\circ$ give the other 9 cases needed to define the nano-embedding of any tile on a path, except for the first and last tile. 
    Bottom left: two paths that start at the common tile $P_0 = Q_0$; path $P$ is shown in pink and grey, $Q$ is shown in pink and red. 
    Bottom right: example showing how we combine the canonical embedding $\frak{E}_P$ of the path $P$ and the nano-embedding $\frak{N}_Q$ of the path $Q$ to make a simple closed curve $c$ in $\R^2$. 
    The start tile on the left  is a special case (in our construction it is always the case that start tile of a nano-embedding has its output side on the east). For the start tile a short vertical magenta segment $s_m$ is drawn 
    so that it ends at the start point of $\frak{N}_Q$. We follow $\frak{N}_Q$ until it ends (``in'' the last tile of $Q$), then draw a short green segment~$s_g$ to the position (``center'') of  the last tile of $P$. From there the reverse of~$\frak{E}_P$, denoted~$\frak{E}_P^\leftarrow$, traces backwards through the positions of tiles of $P$ to the start point of the magenta segment.    
The resulting curve $c$ is the concatenation of the curves $s_m, \frak{N}_Q, s_g, \frak{E}_P^\leftarrow$ and  is a simple closed curve in $\R^2$.   }
    \label{fig:nanoembed}
  \end{figure}
  \begin{figure}[p]
    \begin{center}
       \includegraphics[trim={0.2cm 0.6cm 0.5cm 0.2cm},clip=true]{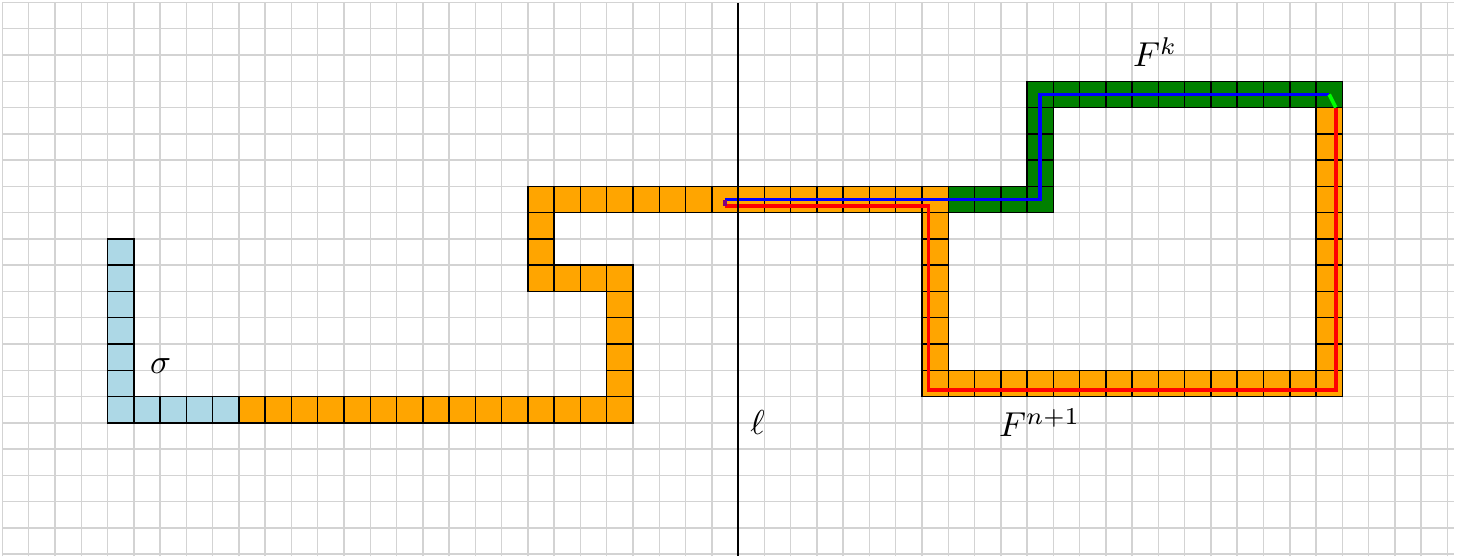}
    \end{center}
    \caption{An example of case 3.1, showing the paths 
    $F^{n+1}$ and $F^k$. The nano-embedding of $F^{n+1}$ is shown as the red curve, and the canonical embedding of $F^k$ is shown as the blue curve. Together with the two small (length $\leq 1$) magenta and green segments, these four curves form a simple closed curve~$c$ in $\R^2$.}
    \label{fig:case3a}
  \end{figure}
  Since $E^{n+1}$ intersects $F^k$ (because in particular, $E^{n+1}$ conflicts with $F^k$), let $b > q$ be the smallest integer such that $\pos{E^{n+1}_b} = \pos{F^k_d}$ for some integer $d$.

  We now define a simple closed curve $c$ inside which a suffix of $P^{n+1}$ starts to grow: let $c$ be the concatenation of $\frak N_{F^{n+1}_{p,p+1,\ldots,b-1}}$, then a length $<1$ line segment from the final~point~of~$\frak N_{F^{n+1}_{p,p+1,\ldots,b-1}}$ to $\pos{F^k_d}$, then $\frak E_{F^k_{f,f+1,\ldots,d}}^\leftarrow$, and finally a line segment of length 0.25 from $\pos{F^k_f}$ to  $\frak N_{F^{n+1}_{p,p+1,\ldots,b-1}}\!(0)$, which is the first point of $\frak N_{F^{n+1}_{p,p+1,\ldots,b-1}}$.
  (Figure~\ref{fig:case3a} shows an example $c$.)

  We claim that $c$ is simple: indeed, since
  $F^k$ only turns left from $F^{n+1}$, and since $\frak N_{F^{n+1}}$ stays immediately to the right of $F^{n+1}$, the four curves used to construct $c$ intersect each other only at the last and first endpoints of each pair of consecutive curves.
  Notice that $c$ is also closed. Therefore, by the Jordan Curve Theorem, $c$ encloses a bounded connected component $\mathcal C$ of $\R^2$. Moreover, at the visible glue of $P^{n+1}$ and $P^k$ on $\ell$, the nano-embedding $\frak N_{\pn}$ of $P^{n+1}$ is below $\frak E_{P^{n+1}}$, and since $P^{n+1}$ places a $V_{\pn}^+$ (\resp, $V_{\pn}^-$) glue on $\ell$, the left-hand (\resp, right-hand) side of $c$ is the inside of $\mathcal{C}$.

  Finally, since $\pn$ does not turn right from $F^{n+1}$ (otherwise, by Lemma~\ref{lem:onepath}   $\pn$ would hide the visibility of at least one of its own glues, which is impossible), a suffix of $\pn$ starts inside $\mathcal{C}$ or on $c$. But since no point of $\mathcal{C}$ is at or above height $h$, and $\pn$ is \hs, the last tile of ${\pn}$  is  positioned outside of $\mathcal C$, which can happen in only two different ways: either $\pn$ turns right from  $F^{n+1}_{p,p+1,\ldots,b-1}$  (contradicting Lemma~\ref{lem:onepath}),  or else $\frak E_{\pn_{q,q+1,\ldots}}$ intersects  $\frak E_{F^k_{g,g+1,\ldots,d}}^\leftarrow$ which is also a contradiction.

\item Therefore, $\pn_{q,q+1,\ldots}$ intersects $F^k$. Moreover, that intersection is necessarily an agreement.
  Let therefore $q'$ be the largest integer such that $ \pn_{q'} = F^k_{g'} $ for some $g'\geq 0$ (notice that  $q'\geq q$), and let $Q = F^k_{0,1,\ldots,g'}\pn_{q'+1,q'+2,\ldots,|\pn|-1}$. 
  Note that  $Q$ is  connected, 
  and that $\pn_{q'+1, q'+2, \ldots}$ does not intersect $F^k_{0,1,\ldots,g'}$ (because of our choice of $q'$), and that the tiles $F^k_{g'}$ and $\pn_{q'+1}$ bind (since $ \pn_{q'} = F^k_{g'} $). 
Therefore, $Q$ is an assemblable path in $\calU$.

  \begin{itemize}

  \item If $g'\geq f$, then $Q$ and \pn have the same visible glue on $\ell$ (this visible glue is on $F^k_{0,1,\ldots,g'}$), and the first difference between $Q_{f,f+1,\ldots}$ and $\pn_{p,p+1,\ldots}$ is a right turn of $\pn_{p,p+1,\ldots}$ from $Q_{f,f+1,\ldots}$, meaning that $Q_{f,f+1,\ldots}$ is less right-priority than $\pn_{p,p+1,\ldots}$. Therefore, $Q\prec \pn$.

  \item Else, $g'<f$, and hence the visible glue of $Q$ on $\ell$ is not the same as the visible glue of \pn on $\ell$. Therefore, the visible glue of $Q$ on $\ell$ is strictly higher than that of \pn. This means that $Q\prec \pn$.
  \end{itemize}

  In both cases, $Q\prec \pn$,  
   hence $Q$ conflicts with $\alpha_n$ by the induction hypothesis. 
  Recall that $Q = F^k_{0,1,\ldots,g'}\pn_{q'+1,q'+2,\ldots,|\pn|-1}$, and that $\asm{F^k} \sqsubseteq \alpha_n $, 
   and therefore $\pn_{q'+1,q'+2,\ldots,|\pn|-1}$ conflicts with $\alpha_n$. 
  
  Hence $\pn$ conflicts with $\alpha_{n}$ which 
  proves the induction hypothesis for $\alpha_{n+1}$.
\end{itemize}

\end{proof}
% !TEX root = t1notiu.tex

\section{Noncooperative tile assembly: Impossibility of bounded Turing machine simulation}\label{sec:no_finite_TM} 

We begin by restating Theorem~\ref{thm:no_finite_TM}. Note that the ``bounding function'' $B_M$ in the statement is an arbitrary upperbound on the space usage  of the Turing machine $M$ as we wish to allow any claimed temperature 1 simulator of Turing machines to be arbitrarily 2D-space-inefficient in it's attempt to do so. 
\begin{reptheorem}{thm:no_finite_TM}
\thmTM
\end{reptheorem}
\begin{proof} 
Intuitively, the proof proceeds  by supposing for the sake of contradiction that there is such a tileset~$V$ and then  modifying  $V$ to get another tile set that can be instantiated as an infinite set of tile assembly systems each one of which 
produces terminal assemblies that have the same scaled (simulated) shape as some  system  $\mathcal{T}_N$ defined in Section~\ref{sec:T}. But this violates Theorem~\ref{thm:shapes}, giving a contradiction. We argue this as follows.

 Let  $M$ be a Turing machine with input alphabet $\{ 1\}^\ast$ that accepts all of its inputs $x\in\{1\}^n$  using space~$s(n)$ and time $t(n)$.  So suppose for the sake of contradiction that there is a tileset~$V$ that simulates $M$ on all inputs using some ``bounding function'' $B_M$ as described in the theorem statement.

We will modify the tile set $V$. 
Since the tile type $H$ is on the rightmost vertical  column of every {\em terminal} assembly $\alpha$, $H$'s east glue type   $g_E$ is either (a)  of strength $0$, or else (b) of strength~$1$ and matches no west glue in the tile set~$V$.  On tile type $H$ we replace $g_E$ with a new glue type~$g_{E}'$ that is of strength~1 and where~$g_{E}'$  appears on no other tile type of~$V$. We also add two new tile types $t_1,t_2$ to $V$: 
the west side of $t_1$ has the glue type $g_{E}'$ and so binds to the east side of $H$, and the south side of  $t_2$ binds to the north side of $t_1$, and the south side of $t_2$ binds to the north side of itself. When $H$ appears in some assembly $\alpha$ it is always possible to bind a tile of type~$t_1$ to the east side of $H$ in $\alpha$ (by hypothesis $H$ is placed in the rightmost column of $\alpha$ hence there is always sufficient (unit) space to the right of $H$ to place $t_1$). There may be a number of places where tiles of type $t_1$ can bind (each to the east of a tile of type~$H$), nevertheless every terminal assembly will have an infinite vertical line of $t_2$ tiles growing to the north of some instance of $t_1$ (i.e.\ since $H$ is in the rightmost column, there can be nothing to the north of a tile of type $t_1$, other than possibly another tile of type $t_1$, and hence there is nothing to stop some tile of type $t_1$ growing an the infinite vertical line of $t_2$ tiles  to its north). 

For each $x$, $|x|=n$, using the modified tileset $V$, the system $\mathcal{V}_x=(V,\sigma_{M,x},1)$ builds an assembly that has (roughly) the same rescaled shape as $\calT_{B_M(n)/B_M(n) t(n)} = \calT_{t(n)}$ (defined in Section~\ref{sec:T}), but with some spatial rescaling (by a factor of $B_M(n)$). Let~$x$ be any input such that, with  $|x| = n$, $t(n) \geq 10|V|$  (e.g.\ choosing $|x| \geq 10|V|$ does the trick). Finally, setting  $m=B_M(n)$, let $R_{m}$ be the $m$-block supertile representation function that is undefined on empty $m$-blocks and maps nonempty $m$-blocks to the  tile $\sigma\in T$.\footnote{$m$-block supertile representation functions were defined in Section~\ref{sec:simulation_def}.    Since the proof of  Theorem~\ref{thm:no_finite_TM}   reasons merely about the {\em shapes} of terminal assemblies, we do not even require that $R_m$ sometimes maps different nonempty $m$-blocks to different tile types of $T$. In other words, having $R_m$ map nonempty $m$-blocks to some tile type of $T$ (here $\sigma\in T$) is sufficient to reason about the shape of assemblies under $R_m$.}
 Then  
$\dom{R^{*}_{m}(\sigma_{\mathcal{V}_x})} = \dom{\sigma_{\mathcal{T}_n}}$ and $\{\dom{R^{*}_{m}(\alpha)} \mid \alpha\in\termasm{\mathcal{V}_x}\} = \{\dom{\beta} \mid \beta\in\termasm{{\mathcal{T}}_{n}}\}$ for $n\geq 10 |V|$. 
Hence our modified $V$ violates violates Theorem~\ref{thm:shapes} where in the theorem statement we set 
$U=V$, 
$m=B_M$, 
$\sigma = \sigma_{M,x}$ and 
$R_m = R_{B_M}$. 
\end{proof}

% !TEX root = t1notiu.tex
\section*{Acknowledgements}  

We thank Damien Regnault, Matthew Patitz, Trent Rogers and Andrew Winslow for important technical feedback and the following people for interesting discussions that helped to improve our thinking: 
Nicolas Schabanel, Dave Doty, Robert Schweller and Jacob Hendricks.  We also thank for Nicolas Schabanel for hosting both authors at LIAFA (Paris 7, France) for the months of September 2014 and May 2015.
A special thanks to our wives Elisa and Beverley. 

\appendix
  \section{Additional simulation definitions}\label{sec:appendix:addSimDefs}
   This appendix continues the definitions in Section~\ref{sec:simulation_def} and were used in previous work on intrinsic universality~\cite{USA, IUSA, Meunier-2014, OneTile, 2HAMIU, Fekete2014}. % They are not required for our results but are included here for completeness. 
  % !TEX root = t1notiu.tex

\begin{definition}
\label{def:equiv-prod} We say that $\mathcal{S}$ and $\mathcal{T}$ have
\emph{equivalent productions} (under $R$), and we write $\mathcal{S}
\Leftrightarrow \mathcal{T}$ if the following conditions hold:
\begin{enumerate}
\item $\left\{R^*(\alpha') | \alpha' \in \prodasm{\mathcal{S}}\right\} =
  \prodasm{\mathcal{T}}$.
\item $\left\{R^*(\alpha') | \alpha' \in \termasm{\mathcal{S}}\right\} = \termasm{\mathcal{T}}$.
\item For all $\alpha'\in \prodasm{\mathcal{S}}$, $\alpha'$ maps cleanly to $R^*(\alpha')$.
\end{enumerate}
\end{definition}
%\dwm{bring in DW's new---more correct---def of equiv productions}

\begin{observation}\label{obs:prodshapes}
If $\mathcal{S}$ and $\mathcal{T}$ that have equivalent productions (they satisfy Definition~\ref{def:equiv-prod}) then  they have equivalent shapes  (they satisfy Definition~\ref{def:equiv-shape}).
%$\{\dom{R^{*}(\alpha)} \mid \alpha\in\termasm{\mathcal S}\} = \{\dom{\beta} \mid \beta\in\termasm{{\mathcal T}_{n}}\}$. 
\end{observation}

The following two definitions (for simulation {\em dynamics}) can be safely ignored by the reader and are included only for the sake of completeness. % The main result of this paper is a negative result about simulation. Although simulation (Definition~\ref{def-s-simulates-t} below), makes use of the following two Definitions concerning the dyn, in fact our proof requires only knowledge of Definition~\ref{def-equiv-prod}(1).  For this paper the following two definitions are ignoreed, but % Our proof proceeds needs only Definition \ref{def-equiv-prod}.  to break claimed simulators, we include the following definitions of simulations between dynamics, for the sake of completeness.
\begin{definition}
\label{def-t-follows-s} We say that $\mathcal{T}$ \emph{follows}
$\mathcal{S}$ (under $R$), and we write $\mathcal{T} \dashv_R \mathcal{S}$ if
$\alpha' \rightarrow^\mathcal{S} \beta'$, for some $\alpha',\beta' \in
\prodasm{\mathcal{S}}$, implies that $R^*(\alpha') \to^\mathcal{T} R^*(\beta')$.
\end{definition}
\begin{definition}
\label{def-s-models-t}
We say that $\mathcal{S}$ \emph{models} $\mathcal{T}$ (under $R$), and we write
$\mathcal{S} \models_R \mathcal{T}$, if for every $\alpha \in
\prodasm{\mathcal{T}}$, there exists $\Pi \subset \prodasm{\mathcal{S}}$ where
$R^*(\alpha') = \alpha$ for all $\alpha' \in \Pi$, such that, for every $\beta
\in \prodasm{\mathcal{T}}$ where $\alpha \rightarrow^\mathcal{T} \beta$, (1) for
every $\alpha' \in \Pi$ there exists $\beta' \in \prodasm{\mathcal{S}}$ where
$R^*(\beta') = \beta$ and $\alpha' \rightarrow^\mathcal{S} \beta'$, and (2) for
every $\alpha'' \in \prodasm{\mathcal{S}}$ where $\alpha''
\rightarrow^\mathcal{S} \beta'$, $\beta' \in \prodasm{\mathcal{S}}$,
$R^*(\alpha'') = \alpha$, and $R^*(\beta') = \beta$, there exists $\alpha' \in
\Pi$ such that $\alpha' \rightarrow^\mathcal{S} \alpha''$.
\end{definition}

\noindent The previous definition essentially specifies that whenever $\mathcal{S}$
simulates an assembly $\alpha \in \prodasm{\mathcal{T}}$, there must be at least one valid growth path in $\mathcal{S}$ for each of the possible next steps %that 
$\mathcal{T}$ could make from~$\alpha$. % which results in an assembly in $\mathcal{S}$ that maps to that next step.

\begin{definition}
\label{def-s-simulates-t} We say that $\mathcal{S}$ \emph{simulates}
$\mathcal{T}$ (under $R$) if $\mathcal{S} \Leftrightarrow_R \mathcal{T}$
(equivalent productions), $\mathcal{T} \dashv_R \mathcal{S}$ and $\mathcal{S}
\models_R \mathcal{T}$ (equivalent dynamics).
\end{definition}

\subsection{Intrinsic universality}\label{sec:iu_def}
Now that we have a formal definition of what it means for one tile assembly system to
simulate another, we can proceed to formally define the concept of intrinsic universality, i.e.\ when there is one general-purpose tile set that can be appropriately programmed to simulate any other tile system from a specified
class of tile assembly systems.
Let $\REPL$ denote the set of all supertile representation functions (i.e.\
$m$-block supertile representation functions for all $m\in\Z^+$).
Define $\frakC$ to be a class of tile assembly
systems, $\frakC_1$ to be the class of all temperature 1 tile assembly systems, and let $U$ be a tileset.

\begin{definition}\label{def:iu-specific-temp}
We say $U$ is \emph{intrinsically universal} for $\frakC$ \emph{at temperature}
$\tau' \in \Z^+$ if there are functions $\mathcal{R}:\frakC \to
\REPL$ and $S:\frakC \to \mathcal{A}^U_{< \infty}$ such that, for each
$\mathcal{T} = (T,\sigma,\tau) \in \frakC$, there is a constant $m\in\N$ such
that, letting $R = \mathcal{R}(\mathcal{T})$,
$\sigma_\mathcal{T}=S(\mathcal{T})$, and $\mathcal{U}_\mathcal{T} =
(U,\sigma_\mathcal{T},\tau')$, $\mathcal{U}_\mathcal{T}$ simulates~$\mathcal{T}$
at scale $m$ and using supertile representation function~$R$.
\end{definition}
That is, $R = \mathcal{R}(\mathcal{T})$ is a representation function that
interprets assemblies of $\mathcal{U}_\mathcal{T}$ as assemblies of
$\mathcal{T}$, and $\sigma_\calT = S(\mathcal{T})$ is the seed assembly used to program
tiles from $U$ to represent the seed assembly of $\mathcal{T}$. In this paper, we disprove the existence of an intrinsically universal tileset for~$\frakC_1$ (the set of all temperature~1 tile assembly systems) at temperature $\tau'=1$.

%\dwmpem{DW deleted a definition}
%\pemm{Ok with me}
%\begin{definition}
%\label{def:iu-general}
%We say that~$U$ is \emph{intrinsically universal} for $\frakC$ if it is
%intrinsically universal for~$\frakC$ at some temperature $\tau'\in Z^+$.
%\end{definition}

\bibliographystyle{abbrv} 
\bibliography{t1notiubib}

\begin{thebibliography}{10}

\bibitem{BMS-DNA2012a}
B.~Behsaz, J.~Ma\v{n}uch, and L.~Stacho.
\newblock Turing universality of step-wise and stage assembly at temperature 1.
\newblock In {\em DNA18: Proc. of International Meeting on DNA Computing and
  Molecular Programming}, volume 7433 of {\em LNCS}, pages 1--11. Springer,
  2012.

\bibitem{Bousquet2010}
M.~Bousquet{-}M{\'{e}}lou.
\newblock Families of prudent self-avoiding walks.
\newblock {\em J. Comb. Theory, Ser. {A}}, 117(3):313--344, 2010.

\bibitem{Versus}
S.~Cannon, E.~D. Demaine, M.~L. Demaine, S.~Eisenstat, M.~J. Patitz,
  R.~Schweller, S.~M. Summers, and A.~Winslow.
\newblock Two hands are better than one (up to constant factors).
\newblock In {\em STACS: Proceedings of the Thirtieth International Symposium
  on Theoretical Aspects of Computer Science}, pages 172--184. LIPIcs, 2013.
\newblock \href{http://arxiv.org/abs/1201.1650}{Arxiv preprint:
  \texttt{1201.1650}}.

\bibitem{Reif-2012}
H.~Chandran, N.~Gopalkrishnan, and J.~Reif.
\newblock Tile complexity of approximate squares.
\newblock {\em Algorithmica}, 66(1):1--17, 2013.

\bibitem{Cook-2011}
M.~Cook, Y.~Fu, and R.~T. Schweller.
\newblock Temperature 1 self-assembly: deterministic assembly in {3D} and
  probabilistic assembly in {2D}.
\newblock In {\em SODA: Proceedings of the 22nd Annual ACM-SIAM Symposium on
  Discrete Algorithms}, pages 570--589, 2011.
\newblock Arxiv preprint: \href{http://arxiv.org/abs/0912.0027}{\tt
  arXiv:0912.0027}.

\bibitem{OneTile}
E.~D. Demaine, M.~L. Demaine, S.~P. Fekete, M.~J. Patitz, R.~T. Schweller,
  A.~Winslow, and D.~Woods.
\newblock One tile to rule them all: Simulating any tile assembly system with a
  single universal tile.
\newblock In {\em ICALP: Proceedings of the 41st International Colloquium on
  Automata, Languages, and Programming}, volume 8572 of {\em LNCS}, pages
  368--379. Springer, 2014.
\newblock Arxiv preprint: \href{http://arxiv.org/abs/1212.4756}{\tt
  arXiv:1212.4756}.

\bibitem{2HAMIU}
E.~D. Demaine, M.~J. Patitz, T.~A. Rogers, R.~T. Schweller, S.~M. Summers, and
  D.~Woods.
\newblock The two-handed tile assembly model is not intrinsically universal.
\newblock In {\em ICALP: Proceedings of the 40th International Colloquium on
  Automata, Languages, and Programming}, volume 7965 of {\em LNCS}, pages
  400--412. Springer, July 2013.
\newblock Arxiv preprint: \href{http://arxiv.org/abs/1306.6710}{\tt
  arXiv:1306.6710}.

\bibitem{IUSA}
D.~Doty, J.~H. Lutz, M.~J. Patitz, R.~T. Schweller, S.~M. Summers, and
  D.~Woods.
\newblock The tile assembly model is intrinsically universal.
\newblock In {\em FOCS: Proceedings of the 53rd Annual IEEE Symposium on
  Foundations of Computer Science}, pages 439--446. IEEE, Oct. 2012.
\newblock Arxiv preprint: \href{http://arxiv.org/abs/1111.3097}{\tt
  arXiv:1111.3097}.

\bibitem{USA}
D.~Doty, J.~H. Lutz, M.~J. Patitz, S.~M. Summers, and D.~Woods.
\newblock Intrinsic universality in self-assembly.
\newblock In {\em STACS: Proceedings of the 27th International Symposium on
  Theoretical Aspects of Computer Science}, pages 275--286, 2009.
\newblock Arxiv preprint: \href{http://arxiv.org/abs/1001.0208}{\tt
  arXiv:1001.0208}.

\bibitem{Doty-2011}
D.~Doty, M.~J. Patitz, and S.~M. Summers.
\newblock Limitations of self-assembly at temperature 1.
\newblock {\em Theoretical Computer Science}, 412(1--2):145--158, 2011.
\newblock Arxiv preprint: \href{http://arxiv.org/abs/0906.3251}{\tt
  arXiv:0906.3251}.

\bibitem{Fekete2014}
S.~P. Fekete, J.~Hendricks, M.~J. Patitz, T.~A. Rogers, and R.~T. Schweller.
\newblock Universal computation with arbitrary polyomino tiles in
  non-cooperative self-assembly.
\newblock In {\em SODA: ACM-SIAM Symposium on Discrete Algorithms}, pages
  148--167. SIAM, 2015.

\bibitem{Flory53}
P.~J. Flory.
\newblock {\em {Principles of Polymer Chemistry}}.
\newblock Cornell University Press, 1953.

\bibitem{geotiles}
B.~Fu, M.~J. Patitz, R.~T. Schweller, and R.~Sheline.
\newblock Self-assembly with geometric tiles.
\newblock In {\em ICALP: Proceedings of the 39th International Colloquium on
  Automata, Languages, and Programming}, volume 7391 of {\em LNCS}, pages
  714--725. Springer, 2012.

\bibitem{gilbert2015continuous}
O.~Gilbert, J.~Hendricks, M.~J. Patitz, and T.~A. Rogers.
\newblock Computing in continuous space with self-assembling polygonal tiles.
\newblock In {\em SODA: ACM-SIAM Symposium on Discrete Algorithms}, pages
  937--956. SIAM, 2016.
\newblock Arxiv preprint: \href{http://arxiv.org/abs/1503.00327}{\tt
  arXiv:1503.00327}.

\bibitem{Hendricks-2014}
J.~Hendricks, M.~J. Patitz, T.~A. Rogers, and S.~M. Summers.
\newblock The power of duples (in self-assembly): It's not so hip to be square.
\newblock In {\em COCOON: Proceedings of 20th International Computing and
  Combinatorics Conference}, pages 215--226, 2014.
\newblock Arxiv preprint: \href{http://arxiv.org/abs/1402.4515}{\tt
  arXiv:1402.4515}.

\bibitem{Jonoska2014}
N.~Jonoska and D.~Karpenko.
\newblock Active tile self-assembly, part 1: Universality at temperature 1.
\newblock {\em Int. J. Found. Comput. Sci.}, 25(2):141--164, 2014.

\bibitem{knuth:math}
D.~E. Knuth.
\newblock {Mathematics and computer science: coping with finiteness}.
\newblock {\em Mathematics: people, problems, results}, 2, 1984.

\bibitem{Manuch-2010}
J.~Ma\v{n}uch, L.~Stacho, and C.~Stoll.
\newblock Two lower bounds for self-assemblies at temperature~1.
\newblock {\em Journal of Computational Biology}, 17(6):841--852, 2010.

\bibitem{meunier2015}
P.-{\'E}. Meunier.
\newblock Non-cooperative algorithms in self-assembly.
\newblock In {\em UCNC: Unconventional Computation and Natural Computation},
  volume 9252 of {\em LNCS}, pages 263--276. Springer, 2015.

\bibitem{Meunier-2014}
P.-{\'E}. Meunier, M.~J. Patitz, S.~M. Summers, G.~Theyssier, A.~Winslow, and
  D.~Woods.
\newblock Intrinsic universality in tile self-assembly requires cooperation.
\newblock In {\em SODA: Proceedings of the ACM-SIAM Symposium on Discrete
  Algorithms}, pages 752--771, 2014.
\newblock Arxiv preprint: \href{http://arxiv.org/abs/1304.1679}{\tt
  arXiv:1304.1679}.

\bibitem{Signals}
J.~E. Padilla, M.~J. Patitz, R.~T. Schweller, N.~C. Seeman, S.~M. Summers, and
  X.~Zhong.
\newblock Asynchronous signal passing for tile self-assembly: Fuel efficient
  computation and efficient assembly of shapes.
\newblock {\em International Journal of Foundations of Computer Science},
  25(4):459--488, 2014.
\newblock \href{http://arxiv.org/abs/1202.5012}{Arxiv preprint:
  \texttt{arxiv:1202.5012}}.

\bibitem{Patitz-2011}
M.~J. Patitz, R.~T. Schweller, and S.~M. Summers.
\newblock Exact shapes and {T}uring universality at temperature 1 with a single
  negative glue.
\newblock In {\em DNA 17: Proceedings of the Seventeenth International
  Conference on DNA Computing and Molecular Programming}, LNCS, pages 175--189.
  Springer, Sept. 2011.
\newblock Arxiv preprint: \href{http://arxiv.org/abs/1105.1215}{\tt
  arXiv:1105.1215}.

\bibitem{Roth01}
P.~W.~K. Rothemund.
\newblock {\em Theory and Experiments in Algorithmic Self-Assembly}.
\newblock PhD thesis, University of Southern California, December 2001.

\bibitem{RotWin00}
P.~W.~K. Rothemund and E.~Winfree.
\newblock The program-size complexity of self-assembled squares (extended
  abstract).
\newblock In {\em STOC: Proceedings of the thirty-second annual ACM Symposium
  on Theory of Computing}, pages 459--468, Portland, Oregon, United States,
  2000. ACM.

\bibitem{SolWin07}
D.~Soloveichik and E.~Winfree.
\newblock Complexity of self-assembled shapes.
\newblock {\em SIAM Journal on Computing}, 36(6):1544--1569, 2007.

\bibitem{Winf98}
E.~Winfree.
\newblock {\em Algorithmic Self-Assembly of {D}{N}{A}}.
\newblock PhD thesis, California Institute of Technology, June 1998.

\bibitem{Winfree98simulationsof}
E.~Winfree.
\newblock Simulations of computing by self-assembly.
\newblock Technical Report Caltech CS TR:1998.22, California Institute of
  Technology, 1998.

\bibitem{woods2015ntrinsic}
D.~Woods.
\newblock Intrinsic universality and the computational power of self-assembly.
\newblock {\em Philosophical Transactions of the Royal Society A: Mathematical,
  Physical and Engineering Sciences}, 373(2046), 2015.
\newblock
  \href{http://dx.doi.org/10.1098/rsta.2014.0214}{\texttt{dx.doi.org/10.1098/rsta.2014.0214}}.

\end{thebibliography}

\end{document}